\documentclass[10pt,a4wide]{article}

\usepackage{geometry,fullpage}
\usepackage{titling}
\usepackage{bm}
\usepackage{amsthm,amssymb,amsmath,amsfonts,rotate}
\usepackage{graphicx,algorithmic,algorithm}
\usepackage[]{datetime}
\usepackage{aliascnt}
\usepackage{hyperref,todonotes}
\usepackage{lmodern}
\usepackage{tikz}
\usetikzlibrary{calc,3d}
\usetikzlibrary{intersections,decorations.pathmorphing,shapes,decorations.pathreplacing,fit}
\usepackage{amsmath,ifthen,eurosym}
\usepackage{wrapfig,cite}
\usepackage{latexsym,amssymb,url}
\usepackage{subcaption}
\usepackage{cleveref}
\graphicspath{{}}

\usepackage{amsmath,latexsym,amsthm,amsfonts,amssymb,%
	 ifthen,color,wrapfig,hyperref,cite,rotate,lmodern,aliascnt,%
	 datetime,graphicx,algorithmic,algorithm,enumitem,%
	todonotes,cleveref,bm,subcaption,tabularx,xcolor,colortbl,xspace%
}

\everymath{\color{MidnightBlack}}

\hypersetup{
 pdftitle = {A complexity dichotomy for hitting connected minors on bounded treewidth graphs: the chair and the banner draw the boundary},
 backref=true,
 colorlinks = true,
 urlcolor = black!50!red,
 linkcolor = black!50!red,
 citecolor = black!50!green,
 menucolor = {red}
 filecolor = {cyan},
 anchorcolor = {black}
 urlcolor = {magenta},
 runcolor = {cyan},
 linkbordercolor = {blue}
}

\usepackage{wasysym}
\newcommand{\restr}[1]{|_{#1}}

\usetikzlibrary{arrows.meta}
\newcommand{\tbox}[4]{\node (#1) at (#2)[draw,thick,text centered, minimum height=1cm]
  {\shortstack{~\\#3 \\ \vspace{0.1cm} \\ {[}#4{]} \vspace{0.1cm}}}
}

\newcommand{\tfbox}[4]{\node (#1) at (#2)[draw,line width= 3pt,text centered,minimum height=1cm]
  {\shortstack{~\\ #3 \\ \vspace{0.1cm} \\ {[}#4{]}  \vspace{0.1cm}}}}
\newcommand{\tibox}[3]{\node (#1) at (#2)[draw, dotted, thick,text centered,minimum height=2.2cm,minimum width=3cm]
  {\shortstack{~\\#3\\~}}}

\usepackage{tikz}
\usetikzlibrary{calc,3d}
\usetikzlibrary{intersections,decorations.pathmorphing,shapes,decorations.pathreplacing,fit,backgrounds,decorations.pathreplacing}
\tikzset{black node/.style={draw, circle, fill = black, minimum size = 4pt, inner sep = 0pt}}
\tikzset{hblack node/.style={draw, circle, fill = black, minimum size = 5pt, inner sep = 0pt}}

\usepackage[greek,english]{babel}
\usepackage[utf8]{inputenc}
\usepackage{alphabeta}
\usepackage{relsize}

\newcommand{\remove}[1]{}

\usepackage{color}

\definecolor{MidnightBlack}{rgb}{0.1,0.1,0.30}
\definecolor{MidnightBlue}{rgb}{0.1,0.1,0.44}
\definecolor{Black}{rgb}{0,0, 0}
\definecolor{Blue}{rgb}{0, 0 ,1}
\definecolor{Red}{rgb}{1, 0 ,0}
\definecolor{White}{rgb}{1, 1, 1}
\definecolor{Grey}{rgb}{.6, .6, .6}
\definecolor{Mygreen}{rgb}{.0, .7, .0}
\definecolor{Yellow}{rgb}{.55,.55,0}
\definecolor{Mustard}{rgb}{1.0, 0.86, 0.35}
\definecolor{applegreen}{rgb}{0.55, 0.71, 0.0}
\definecolor{darkturquoise}{rgb}{0.0, 0.81, 0.82}
\definecolor{celestialblue}{rgb}{0.29, 0.59, 0.82}
\definecolor{green_yellow}{rgb}{0.68, 1.0, 0.18}
\definecolor{crimsonglory}{rgb}{0.75, 0.0, 0.2}
\definecolor{darkmagenta}{rgb}{0.30, 0.0, 0.30}
\definecolor{internationalorange}{rgb}{1.0, 0.31, 0.0}
\definecolor{darkorange}{rgb}{1.0, 0.55, 0.0}

\newcommand{\red}[1]{{\color{Red}#1}}

\DeclareMathOperator{\tw}{{\sf tw}}

\newcommand{\mynewtheorem}[2]{
\newaliascnt{#1}{dummy}
\newtheorem{#1}[#1]{#2}
\aliascntresetthe{#1}
}

\theoremstyle{plain}
\mynewtheorem{theorem}{Theorem}
\mynewtheorem{proposition}{Proposition}
\mynewtheorem{corollary}{Corollary}
\mynewtheorem{lemma}{Lemma}

\theoremstyle{definition}
\mynewtheorem{definition}{Definition}
\theoremstyle{remark}
\mynewtheorem{sublemma}{Sublemma}
\mynewtheorem{claim}{Claim}
\mynewtheorem{remark}{Remark}
\mynewtheorem{fact}{Fact}
\mynewtheorem{conjecture}{Conjecture}
\mynewtheorem{question}{Question}
\mynewtheorem{observation}{Observation}
\mynewtheorem{problem}{Problem}
\mynewtheorem{note}{Note}
\newcommand{\cupall}{\pmb{\pmb{\bigcup}}}

\newcommand{\bd}{{\sf bd}}
\newcommand{\bor}{{\sf bd}}

\newcommand{\ann}{{\sf ann}}
\newcommand{\inter}{{\sf int}}
\newcommand{\bigO}{\mathcal{O}}

\newcommand{\Ostar}{\mathcal{O}^*}

\newcommand{\Fcal}{\mathcal{F}}

\newcommand{\Ocal}{\mathcal{O}}
\newcommand{\Pcal}{\mathcal{P}}

\newcommand{\Nbb}{\mathbb{N}}

\newcommand{\frR}{{\frak{R}}}

\newcommand{\ETH}{{\sf ETH}\xspace}

\newcommand{\bound}[1]{{\bf #1}}

\theoremstyle{plain}

\newcommand{\paraprobl}[5]
{
  \begin{flushleft}
    \fbox{
      \begin{minipage}{#5cm}
        \noindent {\textsc {#1}}\\
        {\bf Input:} #2\\
        {\bf Parameter:} #4\\
        {\bf Output:} #3
      \end{minipage}
    }
  \end{flushleft}
}

\newcommand{\pretp}{\preceq_{\sf tm}}
\newcommand{\prem}{\preceq_{\sf m}}

\newcommand{\paw}{{\sf paw}\xspace}

\newcommand{\chair}{{\sf chair}\xspace}

\newcommand{\butterfly}{{\sf butterfly}\xspace}
\newcommand{\claw}{{\sf claw}\xspace}

\newcommand{\banner}{{\sf banner}\xspace}

\newcommand{\house}{{\sf px}\house}
\newcommand{\ourdiamond}{{\sf diamond}\xspace}

\usepackage{aecompl}
\usepackage{color}
\definecolor{linkcol}{rgb}{0,0,0.8}
\definecolor{citecol}{rgb}{0.65,0,0}
\definecolor{titlecol}{rgb}{0.65,0,0}

\hypersetup {bookmarksopen=true, pdftitle="A complexity dichotomy for hitting connected minors on bounded treewidth graphs", pdfauthor="Julien Baste - Ignasi Sau - Dimitrios M. Thilikos",
  pdfsubject="A complexity dichotomy for hitting connected minors on bounded treewidth graphs", %subject of the document
  pdfmenubar=true, %menubar shown
  pdfhighlight=/O, %effect of clicking on a link
  colorlinks=true, %couleurs sur les liens hypertextes
  pdfpagemode=None, %aucun mode de page
  pdfpagelayout=DoublePage, %ouverture en simple page
  pdffitwindow=true, %pages ouvertes entierement dans toute la fenetre
  linkcolor=linkcol, %couleur des liens hypertextes internes
  citecolor=citecol, %couleur des liens pour les citations
  urlcolor=linkcol %couleur des liens pour les url
}

\newcounter{func}
\newcommand{\newfun}[1]{f_{\refstepcounter{func}\label{#1}\thefunc}}
\newcommand{\funref}[1]{\hyperref[#1]{f_{\ref*{#1}}}} % print a
\newcounter{con}

\newcommand{\conref}[1]{\hyperref[#1]{c_{\ref*{#1}}}} % print a
\definecolor{gray0}{gray}{0.875}
\definecolor{gray1}{gray}{0.775}
\definecolor{gray2}{gray}{0.75}

\newcommand\cuparrow{%
  \mathrel{\ooalign{\hss$\cup$\hss\cr%
      \kern0.3ex\raise0.7ex\hbox{\scalebox{0.7}{$\downarrow$}}}}}

\newcommand\bigcuparrow{%
  \mathrel{\ooalign{\hss$\bigcup$\hss\cr%
      \kern0.55ex\raise0.7ex\hbox{\scalebox{0.7}{$\downarrow$}}}}}

 \usepackage{titling}
\thanksmarkseries{arabic}
\newcommand*\samethanks[1][\value{footnote}]{\footnotemark[#1]}

\begin{document}

% OLD TITLE: Optimal algorithms for hitting (topological) minors on \\ graphs of bounded treewidth
\title{\vspace{-.5cm}Hitting minors on bounded treewidth graphs.\\IV. An optimal algorithm\thanks{An extended abstract of this article appeared in the \emph{Proceedings of the 31st Annual ACM-SIAM Symposium on Discrete Algorithms (SODA), pages 951-970, Salt Lake City, Utah, U.S., January 2020}. The first author was supported by the Deutsche Forschungsgemeinschaft (DFG, German Research Foundation) - 388217545. The two last authors were supported  by   the ANR projects DEMOGRAPH (ANR-16-CE40-0028), ESIGMA (ANR-17-CE23-0010), ELIT (ANR-20-CE48-0008), the French-German Collaboration ANR/DFG Project UTMA (ANR-20-CE92-0027), and the French Ministry of Europe and Foreign Affairs, via the Franco-Norwegian project PHC AURORA.}}

\author{\bigskip Julien Baste\thanks{
	Univ. Lille, CNRS, Centrale
	Lille, UMR 9189 - CRIStAL - Centre de Recherche en Informatique Signal
	et Automatique de Lille, F-59000 Lille, France. Email: \texttt{julien.baste@univ-lille.fr}.}$\ ^{,}$\samethanks[3] \and
Ignasi Sau\thanks{LIRMM, Univ.  Montpellier, CNRS, Montpellier, France. Emails:  \texttt{ignasi.sau@lirmm.fr}, \texttt{sedthilk@thilikos.info}.}
%~\thanks{Departamento de Matem\'atica, Universidade Federal do Cear\'a, Fortaleza, Brazil.}
\and
Dimitrios  M. Thilikos\samethanks[3]%$\,\,^,$\samethanks[4]$\,\,^,$\samethanks[5]%\thanks{Department of Mathematics, National and Kapodistrian University of Athens, Athens, Greece.}}
}

%\date{\vspace{-1cm}}
\date{}

\maketitle

%\removed{
\begin{abstract}
	\noindent 	 For a fixed finite collection of graphs $\Fcal,$ the  \textsc{${\cal F}$-M-Deletion}  problem asks, given an $n$-vertex input graph $G,$ for the minimum number of vertices that intersect all minor models in $G$ of the graphs in $\Fcal.$ by Courcelle Theorem,  this problem can be solved in time $f_{\Fcal}(\tw)\cdot n^{\Ocal(1)},$  where $\tw$ is the treewidth of $G,$ for some function $f_{\Fcal}$ depending on $\Fcal.$ %\sed{Update for ${\cal H}$}
	In a recent series of articles, we have initiated the programme of optimizing asymptotically the function $f_{\Fcal}.$
	Here we provide an algorithm showing that $f_{\Fcal}(\tw) = 2^{\Ocal(\tw\cdot \log\tw)}$ for every collection $\Fcal.$ Prior to this work, the best known function $f_{\Fcal}$ was double-exponential in $\tw.$ In particular, our  algorithm vastly extends the results of Jansen et al.~[SODA 2014] for the particular case $\Fcal=\{K_5,K_{3,3}\}$ and of Kociumaka and  Pilipczuk~[Algorithmica 2019] for graphs of bounded genus, and answers an open problem posed by Cygan et al.~[Inf Comput 2017]. We combine several ingredients such as the machinery of boundaried graphs in dynamic programming via representatives, the Flat Wall Theorem,  Bidimensionality,  the irrelevant vertex technique, treewidth modulators, and protrusion replacement. Together with our previous results providing single-exponential algorithms for particular collections~$\Fcal$~[Theor Comput Sci 2020] and general lower bounds [J Comput Syst Sci 2020], our algorithm yields the following complexity dichotomy when $\Fcal = \{H\}$ contains a single connected graph $H,$ assuming the Exponential Time Hypothesis:  $f_H(\tw)
		=2^{\Theta(\tw)}$ if $H$ is a contraction of the \chair\ or the \banner, and $f_H(\tw)
		=2^{\Theta(\tw\cdot \log\tw)}$ otherwise.

	\vspace{.5cm}

	\vspace{.5cm}

	%\ig{We should remove all the lower bounds from this article}

\end{abstract}\thispagestyle{empty}
\newpage

%\ig{TO DO BEFORE DEADLINE (July 9th)}

%\begin{enumerate}
%\item Use ``autoref'' everywhere, change size of Figure 2
%%\item Explain example of set of representatives of size $2^{2^h}.$

%\sed{Upload the part containing the algorithm to arXiv?}
%\end{enumerate}

\tableofcontents

%\tableofcontents

\thispagestyle{empty}
\newpage\clearpage

\setcounter{page}{1}
\section{Introduction}
\label{label_soggiugnendo}

Let ${\cal F}$ be a finite non-empty collection of non-empty graphs.  In the \textsc{$\Fcal$-M-Deletion} problem, we are given a graph $G$ and an integer $k,$ and the objective is to decide whether there exists a set $S \subseteq V(G)$ with $|S| \leq k$ such that $G \setminus S$ does not contain any of the graphs in ${\cal F}$ as a minor. This problem belongs to the  family of {\sl graph modification problems} and  has a big expressive power, as instantiations of it correspond, for instance,  to \textsc{Vertex Cover} (${\cal F}= \{K_2\}$), \textsc{Feedback Vertex Set} (${\cal F}= \{K_3\}$), and \textsc{Vertex Planarization} (${\cal F}= \{K_5,K_{3,3}\}$). Note that if ${\cal F}$ contains a graph with at least one edge, then \textsc{$\Fcal$-M-Deletion} is {\sf NP}-hard~\cite{LeYa80}.
%by the classical classification result of Lewis and Yannakakis~\cite{LeYa80}.

We study  the parameterized complexity of \textsc{$\Fcal$-M-Deletion} in terms of the treewidth of the input graph (while the size of $k$ may be unbounded). Since the property of containing a graph as a minor can be expressed in monadic second-order logic~\cite{KimLPRRSS16line}, by Courcelle Theorem~\cite{Courcelle90}, \textsc{$\Fcal$-M-Deletion} can be solved in time\footnote{The notation $\Ostar(\cdot)$  suppresses polynomial factors depending on the size of the input graph.} $\Ostar(f_{\Fcal}(\tw))$ on graphs with treewidth at most $\tw,$ where $f_{\Fcal}$ is some computable function depending on $\Fcal.$ As the function $f_{\Fcal}(\tw)$ given by Courcelle Theorem is typically enormous, our goal is to determine, for a fixed collection ${\cal F},$ which is the {\sl best possible} such function $f_{\Fcal}$ that one can (asymptotically) hope for, subject to reasonable complexity assumptions. Besides being an interesting objective in its own, optimizing the running time of algorithms parameterized by treewidth has usually side effects. Indeed, black-box subroutines parameterized by treewidth are nowadays ubiquitous in parameterized~\cite{CyganFKLMPPS15}, exact~\cite{FominK10}, and approximation~\cite{WiSh11-approx} algorithms.

\medskip
\noindent
\textbf{Previous work}. This line of research has attracted considerable attention  in the parameterized complexity community during the last years. For instance, \textsc{Vertex Cover} is easily solvable in time $\Ostar(2^{\Ocal(\tw)}),$ called \emph{single-exponential}, by standard dynamic programming techniques, and no algorithm with running time $\Ostar(2^{o(\tw)})$ exists, unless the Exponential Time Hypothesis (\ETH) fails~\cite{ImpagliazzoP01}. (The \ETH implies that 3-\textsc{Sat} on $n$ variables cannot be solved in time $2^{o(n)}$; see~\cite{ImpagliazzoP01} for more details.) For \textsc{Feedback Vertex Set}, standard dynamic programming techniques give a running time of $\Ostar(2^{\Ocal(\tw \cdot \log \tw)}),$ while the lower bound under the \ETH~\cite{ImpagliazzoP01} is again $\Ostar(2^{o(\tw)}).$ This gap remained open for a while, until Cygan et al.~\cite{CyganNPPRW11} presented an optimal (randomized) algorithm running in time $\Ostar(2^{\Ocal(\tw)}),$ introducing  the celebrated \emph{Cut {\sl \&} Count}
technique. This article triggered several other (deterministic) techniques to obtain single-exponential algorithms for so-called \emph{connectivity problems} on graphs of bounded treewidth, mostly based on algebraic tools~\cite{BodlaenderCKN15,FominLPS16}.

Concerning \textsc{Vertex Planarization}, Jansen et al.~\cite{JansenLS14} presented an algorithm running in time $\Ostar(2^{\Ocal(\tw \cdot \log \tw)})$ as a crucial subroutine in an   algorithm running in time $\Ostar(2^{\Ocal(k \cdot \log k)})$ where $k$ is the solution size. Marcin Pilipczuk~\cite{Pili15} proved afterwards that this running time is optimal under the \ETH, by using the framework introduced by Lokshtanov et al.~\cite{permuclique}  for proving superexponential lower bounds.

Generalizing the above algorithm, the main technical contribution of the recent paper of Kociumaka and  Pilipczuk~\cite{KociumakaP17} is an algorithm running in time $\Ostar(2^{\Ocal((\tw +g)\cdot \log (\tw + g))})$ to solve the \textsc{Genus Vertex Deletion} problem, which consists in deleting the minimum number of vertices from an input graph in order to obtain a graph embeddable on a surface of Euler genus at most $g.$

Cygan et al.~\cite{CyganMPP17} studied the problem of hitting subgraphs (instead of minors), and proved that the problem of hitting all copies of a fixed {\sl path} as a subgraph can be solved in time $\Ostar(2^{\Ocal(\tw \cdot \log \tw)}).$ As a path occurs as a subgraph if and only if it occurs as a minor, their result implies that \textsc{$\{P_h\}$-M-Deletion} can be solved in time $\Ostar(2^{\Ocal(\tw \cdot \log \tw)})$ for every fixed integer $h \geq 2,$ where $P_h$ is the path on $h$ vertices. They left as an open problem whether the algorithm in time $\Ostar(2^{\Ocal(\tw \cdot \log \tw)})$  of  Jansen et al.~\cite{JansenLS14} for \textsc{Vertex Planarization} could be generalized to more minor-closed graph classes, other than planar graphs.

%\ig{Say that we answer an open question of Cygan et al.~\cite{CyganMPP17}: ``It was recently shown that $\Fcal$ being the class of planar graphs, a $2^{\Ocal(t \log t)} |V (G)|$-time algorithm exists. Can this result be generalized to more graph classes?''}
%\ig{We may also say that an optimal algorithm for hitting paths was already given by Cygan et al.~\cite{CyganMPP17}}

%\newpage

In a recent series of three papers~\cite{BasteST20-monster1,BasteST20-monster2,BasteST20-monster3}, we initiated a systematic study of the complexity of \textsc{$\Fcal$-M-Deletion},  parameterized by treewidth\footnote{In these papers~\cite{BasteST20-monster1,BasteST20-monster2,BasteST20-monster3}, in some results we also considered the version of the problem where the graphs in $\Fcal$ are forbidden as {\sl topological} minors; in the current paper we focus exclusively on the minor version.}.  Before stating these results, we say that a collection $\mathcal{F}$ is \emph{connected} if it contains only connected graphs.

In~\cite{BasteST20-monster1} we showed that, for every fixed collection $\mathcal{F},$ \textsc{$\Fcal$-M-Deletion} can be solved in time  $\Ostar\left(2^{2^{\Ocal(\tw \cdot \log \tw)}}\right)$ by a natural dynamic programming algorithm, and that if $\mathcal{F}$ contains a planar graph, the running time can be improved\footnote{In the conference version of~\cite{BasteST20-monster1} we additionally required $\mathcal{F}$ to be connected; in the journal version we  proved this result without  this assumption.} to $\Ostar(2^{\Ocal(\tw  \cdot \log \tw)}).$ If the input graph $G$ is planar or, more generally, embedded in a surface of bounded genus, then we showed that the running time can be further
%\red{(optimally)}
improved to $\Ostar(2^{\Ocal(\tw)}).$

In~\cite{BasteST20-monster2} we provided single-exponential algorithms for \textsc{$\{H\}$-M-Deletion} when $H$ is either  $P_4,$ $C_4,$ the \claw, the \paw, the \chair, or the \banner; see \autoref{label_thoughtlessness} in \autoref{label_apreciadores} for an illustration of these graphs.

In~\cite{BasteST20-monster3} we focused on lower bounds under the \ETH. We proved that for any connected $\Fcal$ containing graphs on at least two vertices, \textsc{$\Fcal$-M-Deletion} cannot be solved in time $\Ostar(2^{o(\tw)}),$ even if the input graph $G$ is planar. More notably, we proved that \textsc{$\Fcal$-M-Deletion} cannot be solved in time $\Ostar(2^{o(\tw \cdot \log \tw)})$ for collections $\Fcal$ satisfying some generic conditions. In particular, these conditions apply when ${\cal F}$ contains a single connected graph $H$ that is {\sl not} a contraction of the \chair or the \banner. Note that the connected graphs $H$ with $|V(H)| \geq 2$ that are a contraction of the \chair or the \banner are those on the left in  \autoref{label_thoughtlessness}, and for each of them \textsc{$\{H\}$-M-Deletion} can be solved in (optimal) single-exponential time~\cite{BasteST20-monster2,BasteST20-monster3}.

% We also provided single-exponential algorithms for the cases where $\Fcal \in \{ \{P_3\}, \{P_4\}, \{C_4\} \}.$ Concerning lower bounds under the  \ETH, we proved that for any connected $\Fcal,$ \textsc{$\Fcal$-M-Deletion} cannot be solved in time $\Ostar(2^{o(\tw)}),$ even if the input graph $G$ is planar. Inspired by the reduction of Pilipczuk~\cite{Pili15}, we proved that the problem cannot be solved in time $\Ostar(2^{o(\tw \cdot \log \tw)})$  for some families of collections ${\cal F},$ for example, when all graphs in ${\cal F}$ are planar and 3-connected. In the subsequent paper~\cite{BasteST18}, we focused on small planar graphs. Namely, we classified the optimal asymptotic
%complexity of \textsc{$\{H\}$-M-Deletion} when $H$ is a connected {\sl planar} graph on at most five vertices. This classification is illustrated in \autoref{label_thoughtlessness} in \autoref{label_apreciadores}, where it should be noted that $K_5$ is {\sl not} covered by the results in~\cite{BasteST18}. To achieve that, we provided single-exponential algorithms for a number of small patterns not considered in~\cite{BasteST20-monster1} and  superexponential lower bounds for the remaining cases, this time inspired by a reduction of Bonnet et al.~\cite{BonnetBKM-IPEC17} for generalized feedback vertex set problems. Full proofs of the results in~\cite{BasteST20-monster1,BasteST18} are available at~\cite{monster-arXiv}.

\medskip
\noindent
\textbf{Our results}. In this article we present an algorithm to solve \textsc{$\Fcal$-M-Deletion} in time $\Ostar(2^{\Ocal(\tw \cdot \log \tw)})$ for {\sl every} collection $\Fcal,$ thus making a significant step towards a complete classification of the complexity of the \textsc{$\Fcal$-M-Deletion} problem parameterized by treewidth. That is, we drop the condition that $\Fcal$ contains a planar graph, which was critically needed in the algorithm presented in~\cite{BasteST20-monster1} in order to bound the treewidth of an $\Fcal$-minor-free graph. Our algorithm can be interpreted as an exponential ``collapse'' of the natural dynamic programming algorithm running in time $\Ostar\left(2^{2^{\Ocal(\tw \cdot \log \tw)}}\right)$ given in~\cite{BasteST20-monster1}. Besides largely improving our previous results~\cite{BasteST20-monster1}, this algorithm  generalizes the ones for ${\cal F}= \{K_5,K_{3,3}\}$ given by Jansen et al.~\cite{JansenLS14} %\red{by following an approach different from the generalization of
%Kociumaka and  Pilipczuk~\cite{KociumakaP17},}
and for the \textsc{Genus Vertex Deletion} problem given by Kociumaka and  Pilipczuk~\cite{KociumakaP17},
which are based on embeddings, and answers positively the open problem of Cygan et al.~\cite{CyganMPP17} for {\sl every} minor-closed graph class.
%\ig{the obstructions for Euler genus at most $g$ may be disconnected~\cite{MoharT01}, so we DON'T generalize the results of Kociumaka and  Pilipczuk~\cite{KociumakaP17}}.
Since the algorithm is quite involved, we provide an overview of it in \autoref{label_voixoderncrow}.

This algorithm in time $\Ostar(2^{\Ocal(\tw \cdot \log \tw)})$ for every collection $\Fcal,$ together with the lower bounds under the \ETH given in~\cite{BasteST20-monster3}, the single-exponential algorithms given in~\cite{BasteST20-monster2}, and the known cases $\Fcal = \{P_2\}$~\cite{ImpagliazzoP01,CyganFKLMPPS15}, $\Fcal = \{P_3\}$~\cite{P3-cover,P3-cover-improved}, and $\Fcal = \{C_3\}$~\cite{CyganNPPRW11,BodlaenderCKN15}, imply the following complexity dichotomy when $\Fcal$ consists of a single connected graph $H,$ which we suppose to have at least one edge.

\begin{theorem}\label{label_purposelessness}
	Let $H$ be a connected graph. Under the \ETH, $\{H\}$-\textsc{M-Deletion} is solvable in time\footnote{We use $n$ and $\tw$ for the number of vertices and the treewidth of the input graph, respectively.}\vspace{-2mm}

	\begin{itemize}\setlength{\itemsep}{-0.7mm}
		\item $2^{\Theta(\tw)} \cdot n^{\Ocal(1)},$ if $H$ is a contraction of the \chair or the \banner, and
		\item $2^{\Theta(\tw \cdot \log \tw)} \cdot n^{\Ocal(1)},$  otherwise.
	\end{itemize}
	%In both cases, the running time is asymptotically optimal under the \ETH.
\end{theorem}

%\ig{Dimitrios, it is NOT true, as we said in Montpellier,  that the ``hard'' cases are those that contain one of the graphs $\{P_5,K_{1,4},\ourdiamond\}$ as a minor; the \banner is a counterexample}

\vspace{-.05cm}

This dichotomy is depicted in \autoref{label_thoughtlessness}, containing all  connected graphs $H$ with $2 \leq |V(H)| \leq 5$; note that if $|V(H)| \geq 6,$ then $H$ is not a contraction of the \chair or  the \banner, and therefore the second item above applies. Note also that $K_4$ and the {\sf diamond} are the only graphs on at most four vertices for which the problem is solvable in time $\Ostar(2^{\Theta (\tw \cdot \log \tw)})$ and that the \chair and the \banner are the only graphs on at least  five vertices for which the problem is solvable in time $\Ostar(2^{\Theta (\tw)}).$

The crucial role played by the \chair and the \banner  in the complexity dichotomy may seem surprising at first sight. In fact, we realized a posteriori that the ``easy'' cases can be succinctly described in terms of the \chair and the \banner by taking a look at \autoref{label_thoughtlessness}. Note that the ``easy'' graphs can be equivalently characterized as those that are minors of the \banner, with the exception of $P_5.$ Nevertheless, there is some intuitive reason for which excluding the \chair or the \banner constitutes the horizon on the existence of single-exponential algorithms. Namely, focusing on the \banner, every connected component (with at least five vertices) of a graph that excludes the  \banner as a  minor is either a cycle (of any length) or a tree in which some vertices have been replaced by triangles; both such types of components can be maintained by a dynamic programming algorithm in single-exponential time \cite{BasteST20-monster2}. A similar situation occurs when excluding the \chair. It appears that if the characterization of the allowed connected components is enriched in some way, such as restricting the length of the allowed cycles or forbidding certain degrees, the problem becomes inherently more difficult, inducing a transition from time $\Ostar(2^{\Theta (\tw)})$  to $\Ostar(2^{\Theta (\tw \cdot \log \tw)}).$

\medskip
\noindent
\textbf{Organization of the paper}. In \autoref{label_voixoderncrow} we provide a high-level overview of the algorithm running in time $\Ostar(2^{\Ocal(\tw \cdot \log \tw)}).$ In \autoref{label_laterralavolpelcel} we give some preliminaries. In \autoref{label_meditatively} we deal with flat walls, in \autoref{label_donnescamente}
we apply the irrelevant vertex technique in the context of boundaried graphs, and in  \autoref{label_pertenecerles}
we use this in order to bound the size of the dynamic programming tables.
%The lower bounds are presented in \autoref{ridicule}.
We conclude the article in \autoref{label_pelopponesian} with some open questions for further research. In \autoref{label_reproductions} we present an estimation of the constants depending on the (fixed) collection $\Fcal$ in our algorithm (cf. \autoref{label_geographical}).

%Due to space limitations, most of the material has been moved to the appendices, in particular the proofs of all the results marked with `($\star$)'.
%\vspace{-2.mm}

\section{Overview of the algorithm}
\label{label_voixoderncrow}
%\vspace{-1.5mm}

In order to obtain our algorithm of time $\Ostar(2^{\Ocal(\tw \cdot \log \tw)})$ for every  collection~$\Fcal,$ our approach can be streamlined as follows. We use the machinery of boundaried graphs, equivalence relations, and representatives originating in the seminal work of Bodlaender et al.~\cite{BodlaenderFLPST16} and  subsequently used, for instance, in~\cite{GarneroPST15,F.V.Fomin:2010oq,KimLPRRSS16line,BasteST20-monster1}. Let $h$ be a constant depending only on the collection $\Fcal$ (to be defined in the formal description of the algorithm) and let $t$ be a positive integer that is at most the treewidth of the input graph plus one. Skipping several technical details, a \emph{$t$-boundaried} graph is a graph with a distinguished set of vertices --its \emph{boundary}-- labeled bijectively with integers from the set $[t].$ We say that two $t$-boundaried graphs are \emph{$h$-equivalent} if for {\sl any} other $t$-boundaried graph that we can ``glue'' to each of them, resulting in graphs $G_1$ and $G_2,$ and every graph $H$ on at most $h$ vertices, $H$ is a minor of $G_1$ if and only if it is a minor of $G_2$ (cf.~\autoref{label_laterralavolpelcel} for the precise definitions).
Let ${\cal R}_{h}^{(t)}$ be a set of {\sl minimum-sized} representatives of this equivalence relation. Since $h$-equivalent (boundaried) graphs have the same behavior in terms of eventual occurrences of minors of size up to $h,$ there is a generic dynamic programming algorithm (already used in~\cite{BasteST20-monster1}) to solve \textsc{$\Fcal$-M-Deletion} on a rooted tree decomposition of the input graph, via a typical bottom-up approach: at every bag $B$ of the tree decomposition, naturally associated with a $t$-boundaried graph $G_B,$ and for every
representative $R \in {\cal R}_{h}^{(t)},$ store the minimum size of a set $S \subseteq V(G_B)$ such that the graph $G_B \setminus S$ is $h$-equivalent to $R$ (cf.~\autoref{label_sechseckigen} for some more details\footnote{This step was the only reason for which in the conference version of this article we required the collection $\Fcal$ to be connected. As mentioned in \autoref{label_soggiugnendo}, in the full version
	of~\cite{BasteST20-monster1} we dropped the connectivity assumption, which implies that in the current article we can drop it as well.}). This yields an algorithm running in time $\Ostar(|{\cal R}_{h}^{(t)}|^2),$ and therefore it suffices to prove that $|{\cal R}_{h}^{(t)}| = 2^{\Ocal_h(t \cdot \log t)},$ where the notation `$\Ocal_h(\cdot)$' means that the hidden constants
depend only on $h.$ Since we may assume that the graphs in ${\cal R}_{h}^{(t)}$ exclude some
graph on at most $h$ vertices as a minor (as all those that do not are $h$-equivalent), hence
they have a linear number of edges~\cite{MyersT05}, it is enough to prove that, for every $R \in {\cal R}_{h}^{(t)},$ it holds that
\begin{eqnarray}\label{label_cdeefccdfliiih}
	|V(R)| = \Ocal_h(t).\label{label_privateering}
\end{eqnarray}

%\ig{Do we want to include the following sentence?}
Note that this is indeed sufficient in order to obtain an algorithm within the claimed running time, as there are at most $\binom{|V(R)|^2}{|E(R)|}=2^{\Ocal_h(|V(R)| \cdot \log |V(R)||)}$ representatives, and $t! = 2^{\Ocal(t \cdot \log t)}$ possible labelings for the vertices in the boundary. In order to prove~\autoref{label_privateering}, we combine a number of different techniques, which we proceed to discuss informally, and that are schematically summarized in \autoref{label_unprosperous}:

\begin{figure}[htb]
	\centering
	\scalebox{.8}{\begin{tikzpicture}[scale=.975]
			\tibox{init}{13,0}{ $ t \leq \tw(G) + 1$\\$h = f(\mathcal{F})$\\$R \in \mathcal{R}_h^{(t)}$};

			\tbox{le18}{0,0}{Embedding\\ with dispersed\\vertices}{\autoref{label_dispenseront}};
			\tbox{pr16}{4,0}{Confinement of\\ models inside\\ a railed annulus}{\autoref{label_interessiert}};
			\tbox{th20}{2,-3}{Collapse of \\topological minor \\models inside a wall}{\autoref{label_constitutivos}};
			\tbox{fwt}{8,0}{Flat Wall Theorem\\ \cite{RobertsonS95b,KawarabayashiTW18,accurate}}{\autoref{label_intercanvien}};
			\tbox{th14}{7,-3}{Large\\ $h$-homogeneous\\ subwall}{\autoref{label_incominciando}};
			\tbox{le21}{3,-6}{$R$ contains\\ no irrelevant vertex}{\autoref{label_pretendientes}};
			\tbox{le22}{3,-9}{${\bf p}_{h,r}(R) \leq t$}{\autoref{label_dispensaries}};
			\tbox{le24}{8,-6}{${\bf p}_{h,r}$ is\\ bidimensional}{\autoref{label_substantiality}};
			\tbox{le15}{12,-6}{${\bf p}_{h,r}$ is\\ separable}{\autoref{label_concerniente}};
			\tbox{le25}{10,-9}{$R$ has a treewidth modulator\\ of size $\bigO_h(t)$\\ containing the boundary}{\autoref{label_murmuradores}};
			\tbox{pr26}{10,-12.2}{Linear protrusion\\ decomposition of $R$}{\autoref{label_unintentional}};
			\tbox{le27}{3,-12.2}{$|V(R)| = \bigO_h(t)$}{\autoref{label_intentionnel}};
			\tbox{div}{3,-15.4}{$|\mathcal{R}_h^{(t)}| = 2^{\bigO_h(t \cdot \log t)}$}{\autoref{label_encyclopedia}};
			\tfbox{th1}{10,-15.4}{Algorithm in time\\ $\Ostar(2^{\bigO_h(\tw \cdot \log \tw)})$\\ for any collection $\mathcal{F}$}{\autoref{label_instintivamente}};

			\draw[-{Kite[{length=4mm}]},line width=2pt] (fwt) -- (th14);
			\draw[-{Kite[{length=4mm}]},line width=2pt] (le18) -- (th20);
			\draw[-{Kite[{length=4mm}]},line width=2pt] (pr16) -- (th20);
			\draw[-{Kite[{length=4mm}]},line width=2pt] (th20) -- (le21);
			\draw[-{Kite[{length=4mm}]},line width=2pt] (th14) -- (le21);
			\draw[-{Kite[{length=4mm}]},line width=2pt] (le21) -- (le22);
			\draw[-{Kite[{length=4mm}]},line width=2pt] (le22) -- (le25);
			\draw[-{Kite[{length=4mm}]},line width=2pt] (le24) -- (le25);
			\draw[-{Kite[{length=4mm}]},line width=2pt] (le15) -- (le25);
			\draw[-{Kite[{length=4mm}]},line width=2pt] (le25) edge node[right] {\cite{KimLPRRSS16line}} (pr26);
			\draw[-{Kite[{length=4mm}]},line width=2pt] (pr26) edge node[above] {\shortstack{Reduce \\ protrusions \cite{BasteST20-monster1}}} (le27);
			\draw[-{Kite[{length=4mm}]},line width=2pt] (le27) edge node[left] {\shortstack{Sparsity \\ of the\\ representatives}} (div);
			\draw[-{Kite[{length=4mm}]},line width=2pt] (div) edge node[above] {\shortstack{DP algorithm\\ from \cite{BasteST20-monster1}}} (th1);

		\end{tikzpicture}}
	\caption{Diagram of the algorithm in time $\Ostar(2^{\Ocal(\tw \cdot \log \tw)})$ for any collection $\Fcal.$}
	\label{label_unprosperous}
\end{figure}
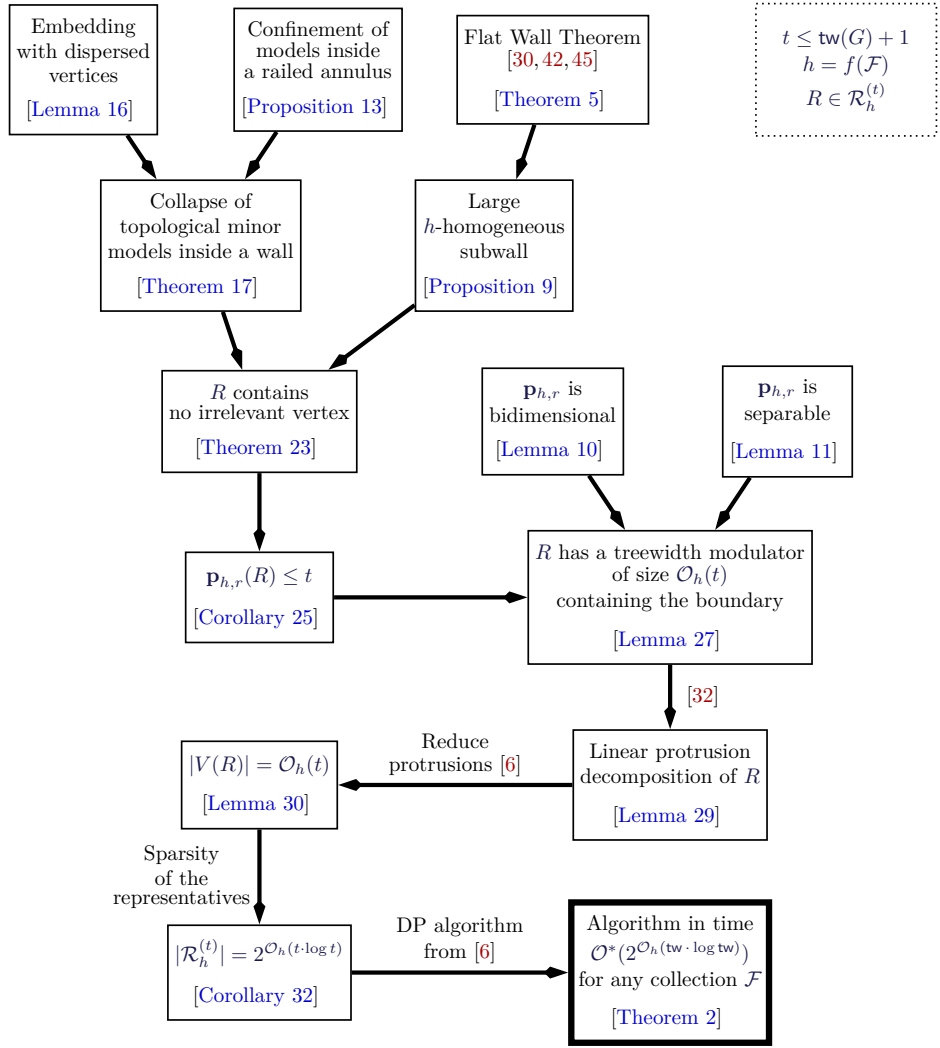

%\begin{itemize}

%\newpage

\bigskip

\noindent$\red{\blacktriangleright}$ We use the Flat Wall Theorem of Robertson and Seymour~\cite{RobertsonS95b}, in particular a version (\autoref{label_intercanvien}) that has been recently proved in~\cite{accurate} (and is based on the framework of Kawarabayashi et al.~\cite{KawarabayashiTW18}), which
incorporates the so-called
{\sl regularity property}.
%\sed{Some cases the one and some cases the other \ig{in fact, do we ever use Kawarabayashi et al.~\cite{KawarabayashiTW18}?}}
%Chuzhoy~\cite{Chuzhoy15}.
In a nutshell, this theorem says that every $K_h$-minor-free graph $G$ has a set of vertices $A \subseteq V(G)$ --called \emph{apices}-- with $|A| = \Ocal_h(1)$ such that $G\setminus A$ contains a \emph{flat wall} of height $\Omega_h (\tw(G)).$ Here, the definition of ``flat wall'' is quite involved and is detailed in \autoref{label_meditatively}; it essentially means a subgraph that has a  bidimensional grid-like structure, separated from the rest of the graph by its perimeter, and that is ``close'' to being planar, in the sense that it can be embedded in the plane in a way that its potentially non-planar pieces, called \emph{flaps}, have a well-defined structure along larger pieces called \emph{bricks}.

\medskip
\noindent$\red{\blacktriangleright}$ A subwall of a flat wall is \emph{$h$-homogeneous}  if for every brick of the subwall, the flaps within that brick have the same variety of \emph{$h$-folios}, that is, the same sets of ``boundaried'' minors of \emph{detail} at most $h$ (the detail of a boundaried graph is the maximum between its number of edges and its number of non-boundary vertices). This notion is inspired (but is not the same) by the one defined by Robertson and Seymour in~\cite{RobertsonS95b}.
Using standard ``zooming'' arguments,  we can prove that, given a flat wall, we can find a large $h$-homogeneous subwall inside it (\autoref{label_incominciando}). Homogeneous subwalls are very useful because, as we explain below, they permit the application of the irrelevant vertex technique adapted to our purposes.

\medskip
\noindent$\red{\blacktriangleright}$ We say that a vertex set $S$ \emph{affects} a   flat wall if some vertex within the wall has a neighbor in $S$ that is {\sl not} an apex. With these definitions at hand, we define a parameter, denoted by ${\bf p}_{h,r}$ in this informal description, mapping every graph $G$ to the smallest size of a vertex set  that affects all  $h$-homogeneous  flat walls with at most $h$ apices and height at least $r$ in $G.$ It is not hard to prove that the parameter ${\bf p}_{h,r}$  has a ``{bidimensional}’’ behavior~\cite{DemaineFHT05sube,FominDHT16,F.V.Fomin:2010oq}, in the sense that its value on a flat wall depends quadratically on the height of the wall (\autoref{label_substantiality}) and it is \emph{separable}~\cite{BodlaenderFLPST16,F.V.Fomin:2010oq,FominDHT16} (\autoref{label_concerniente}).

%\sed{Declare that the minors of a representative do not go deeply in the boundary graph. We have the mention  \autoref{label_disgustingly}  here! \ig{ \autoref{label_disgustingly} is mentioned in the next item. Do you want to do it before?}}

%they allow to reroute (topological) minor models and
%, which will be crucial in order to find an \emph{irrelevant} vertex;
% that is, a vertex whose removal does not affect \red{the type (positive of negative) of our instance}\sed{This is not exactly like that!}; see \autoref{label_donnescamente} for the precise definition.

\medskip
\noindent$\red{\blacktriangleright}$ The most complicated step towards proving \autoref{label_privateering}
is to find an ``irrelevant’’ vertex inside a sufficiently large (in terms of $h$) flat wall of a boundaried graph that is {\sl not} affected by its boundary  (\autoref{label_pretendientes}).
Informally, here ``irrelevant’’ means a non-boundary vertex of $R$
that can be avoided by {\sl any} minor model of a graph on at most  $h$ vertices and edges that
traverses the boundary of $R,$ {\sl  no matter} the graph that may be glued to it
and {\sl no matter how} this model traverses the boundary of $R$; see \autoref{label_donnescamente} for the precise definition.
The {\sl  irrelevant vertex technique} originated in the seminal work of Robertson and Seymour~\cite{RobertsonSGM22,RobertsonS95b} and has become a very useful tool used in various kinds of linkage and cut problems~\cite{KratschW12,Reed95,JansenLS14,KociumakaP17,AdlerKKLST17}. Nevertheless, given the nature of our setting, it is critical
that the size of the flat wall where the irrelevant vertex appears
does {\sl not depend} on the boundary size. To the best of our knowledge, this property is not guaranteed by the existing results
on the irrelevant vertex technique (such as \cite[(10.2)]{RobertsonS95b}  and its subsequent proof in~\cite{RobertsonSGM22}).
To achieve it and, moreover, in order to make an
estimation of the parametric dependencies, we develop a self-reliant theoretical framework that uses
the following ingredients:
\begin{itemize}
	\item[$\circ$] With a flat wall $W$ we associate a bipartite graph $\tilde{W},$ which we call its \emph{leveling}  as defined in~\cite{accurate}; cf.~\autoref{label_disintegrated} for the precise definition. In particular, this graph has a vertex for every flap of the flat wall, and can be embedded in a disk in a planar way.
	      \vspace{-.12cm}

	\item[$\circ$] It turns out to be more convenient to work with \emph{topological} minor models instead of minor models; we can afford it since for every graph $H$ there are at most $f(H)$ different topological minor
	      minimal graphs  that contain $H$ as a minor (\autoref{label_daskalojannes}). The reason for this is that it is easier to deal with the branch vertices of a topological minor model in the analysis. Given a topological minor model, we say that a flap of a wall is \emph{dirty} if it contains a branch vertex of the model, or there is an edge from the flap to an apex vertex of the wall. We also define the leveling of a topological minor model, and we equip its dirty flaps with colors that encode their $h$-folios. We now proceed to explain how to reroute the colored leveling of a topological minor model.
	      \vspace{-.12cm}

	\item[$\circ$] In order to reroute (colored levelings of) topological minor models, it will be helpful to use \emph{railed annuli}, a structure introduced in~\cite{KaminskiT12} that occurs as a subgraph inside a flat wall (\autoref{label_simultaneously}) and that has the following nice property, recently proved in~\cite{GolovachST20-SODA} (\autoref{label_interessiert}): if a railed annulus is large enough compared to $h,$ every topological minor model of a graph on at most $h$ vertices traversing it can be rerouted so that the branch vertices are preserved and such that, more importantly, the intersection of the new model with a large {\sl prescribed} part of the railed annulus is \emph{confined}, in the sense that it is only allowed to use a well-defined set of paths in that part, which does not depend on the original model.
	      \vspace{-.12cm}
	      % \ig{Dimitrios, do you agree with the previous explanation?}.

	\item[$\circ$]  We also need a technical result with a graph drawing flavor (\autoref{label_dispenseront}) guaranteeing that large enough railed annuli contain topological minor models of every graph of maximum degree three
	      % \ig{check also this explanation}
	      with the property, in particular, that certain vertices are pairwise far apart in the embedding. Using this result and the one proved in~\cite{GolovachST20-SODA}  mentioned above, we can finally prove  (\autoref{label_constitutivos}) that every topological minor model of a graph $H$ inside a graph with a large flat wall $W$ can be ``collapsed'' inside the wall, in the following sense: $G$ contains another topological minor model of a graph $H',$ such that $H$ is a minor of $H',$ and such that the new model avoids the central part of the annulus; here is where the irrelevant vertex will be found.
	      %\vspace{-.12cm}

	\item[$\circ$] To conclude, it just remains to ``lift'' the constructed embedding of the colored leveling of the topological minor to an embedding of the ``original'' minor in the flat wall  (\autoref{label_pretendientes}). For that, we exploit the fact that we have rerouted the model inside an $h$-homogeneous subwall not affected by the boundary, which allows to mimic the behavior of the original minor inside the flaps of the wall, using that all bricks have the same variety of $h$-folios.
\end{itemize}

The above arguments, incorporated in the proof of \autoref{label_pretendientes},
imply that if $R \in {\cal R}_{h}^{(t)}$ is a minimum-sized representative, then
its boundary affects all large enough $h$-homogeneous  flat walls, as otherwise we could remove an irrelevant vertex and find a smaller equivalent representative. In particular, it follows that, for every $R \in {\cal R}_{h}^{(t)},$ we have ${\bf p}_{h,r}(R) \leq t$ (\autoref{label_dispensaries}).

\medskip

\noindent$\red{\blacktriangleright}$ Combining that the parameter ${\bf p}_{h,r}$ is ``bidimensional’’ and separable along with the fact that ${\bf p}_{h,r}(R) \leq t$ for every $R \in {\cal R}_{h}^{(t)},$
we prove in \autoref{label_murmuradores}  that every representative $R \in {\cal R}_{h}^{(t)}$ has a vertex subset $S$ containing its boundary, with $|S| = \Ocal_h(t),$ whose removal leaves a graph of treewidth bounded by a function of $h$; such a set is called a \emph{treewidth modulator}. (In \autoref{label_reproductions}
we provide an improved version of \autoref{label_murmuradores}, namely \autoref{label_gleichsinnigen}, by adapting the proof of  \cite[Lemma~3.6]{F.V.Fomin:2010oq}.)

\medskip
\noindent$\red{\blacktriangleright}$ Once we have a treewidth modulator of size $\Ocal_h(t)$ of a representative $R,$ all that remains is to pipeline it with known techniques to compute an appropriate \emph{protrusion decomposition}~\cite{KimLPRRSS16line} (\autoref{label_unintentional}) and to \emph{reduce protrusions} to smaller equivalent ones of size bounded by a function of $h$ --we use the version given in~\cite{BasteST20-monster1} adapted to the \textsc{$\Fcal$-M-Deletion} problem-- (\autoref{label_intentionnel}), implying that
$|V(R)| = \Ocal_h(t)$ for every every $R \in {\cal R}_{h}^{(t)}$ and concluding the proof of \autoref{label_privateering}.
%\end{itemize}

\medskip

It should be noted that all the items above do {\sl not} need to be converted into an algorithm, they are just used in the analysis: the conclusion is that if $R \in {\cal R}_{h}^{(t)}$ is a minimum-sized representative, then $|V(R)| = \Ocal_h(t),$ as otherwise some reduction rule could be applied to it (either by removing an irrelevant vertex or by protrusion replacement), thus obtaining an equivalent representative of smaller size and contradicting its minimality.
Our main result can be formally stated as follows.

\begin{theorem}\label{label_instintivamente}
	Let $\Fcal$ be a finite non-empty collection of non-empty graphs. There exists a constant $c_{\Fcal}$ such that the $\Fcal$-\textsc{M-Deletion} problem is solvable in time $c_{\Fcal}^{\tw \cdot \log \tw}\cdot n$ on $n$-vertex graphs of treewidth at most $\tw.$%\sed{The $n$ is linear here!}
\end{theorem}

In \autoref{label_reproductions}
we provide an estimation of the constant $c_{\Fcal}$ in the above theorem based on the parametric dependencies of the Unique Linkage Theorem~\cite{KawarabayashiW2010asho,RobertsonSGM22}.

%
%\sed{Here we have to mention the lower bounds! \ig{this section is devoted only to the algorithm. I already wrote some stuff about the lower bounds in the intro}}
%Finally, it is important to point out that the sparsity of the representatives is necessary in order to prove that they have size $|V(R)| = \Ocal_h(t).$ Indeed, in \autoref{label_disillusioned} we provide \ig{TO BE DONE!} an example of a
%
%\ig{Speak about lower bounds for the size of the representatives, if we do not assume minor-freeness} We stress  that in~\autoref{label_soggiugnendo} the hypothesis that ${\bf G}$ is $K_{q}$-minor-free is necessary. Indeed, without this hypothesis there are collections ${\cal G}\subseteq {\cal B}_{h}^{(t)}$ and graphs in ${\sf intertwine}({\cal G})$ with \sed{UPDATE $t^{h}$}$2^{\Omega(t)}$ vertices. Moreover, this assumption is also necessary in~\autoref{label_thoughtlessness}: there are collections ${\cal G}\subseteq {\cal B}_{h}^{(t)}$ with $|{\sf intertwine}({\cal G})|=2^{t^h}.$
%\vspace{-3mm}

\section{Preliminaries}
\label{label_laterralavolpelcel}
%\vspace{-0mm}

%\ig{TO BE RE-ORGANIZED AND POLISHED}
%
%In this section we provide some definitions to be used throughout the article.
%In \autoref{sec_def_basic} we give the basic definitions, especially concerning boundaried graphs and operations on them, and in \autoref{label_meditatively} we formally define flat walls.

%\subsection{Basic definitions}
%\label{sec_def_basic}

\subsection{Basic definitions}
\paragraph{Sets and integers.}\label{label_circunspetta}
We denote by $\Bbb{N}$ the set of non-negative integers.
Given two integers $p,q,$ where $p\leq q,$  we denote by $[p,q]$ the set $\{p,\ldots,q\}.$
For an integer $p\geq 1,$ we set $[p]=[1,p]$ and $\Bbb{N}_{\geq p}=\mathbb{N}\setminus [0,p-1].$ Given a non-negative integer $p,$ we denote by ${\sf odd}(p)$ the minimum odd number that is not smaller than $p.$
For a set $S,$ we denote by $2^{S}$ the set of all subsets of $S$ and by ${S \choose 2}$ the set of all subsets of $S$ of size two.
If ${\cal S}$ is a collection of objects where the operation $\cup$ is defined, then we denote $\cupall {\cal S}=\bigcup_{X\in {\cal S}}X.$

\paragraph{Basic concepts on graphs.}\label{label_recommencera}
As a graph $G$ we denote any pair $(V,E)$ where $V$ is a finite set and $E\subseteq {V \choose 2},$ that is, all graphs of this paper are undirected, finite, and without loops or multiple edges.
We also define $V(G)=V$ and $E(G)=E.$
Unless stated otherwise, we denote by $n$ and $m$ the number of vertices and edges, respectively, of the graph under consideration.
%We use standard graph-theoretic notation and we refer the reader to \cite{Die10} for any undefined terminology.
We say that a pair $(L,R)\in 2^{V(G)}\times 2^{V(G)}$ is a {\em separation} of $G$ if $L\cup R=V(G)$ and there is no edge in $G$ between $L\setminus R$ and $R\setminus L.$ The \emph{order} of a separation $(L,R)$ is $|L \cap R|.$
Given a vertex $v\in V(G),$ we denote by $N_{G}(v)$ the set of vertices of $G$ that are adjacent to $v$ in $G.$
Also, given a set $S\subseteq V(G),$ we set $N_G(S) = \big(\bigcup_{v \in S}N_G(v)\big) \setminus S.$
% \ig{We need also this (used for instance in condition (iii) of tightness of a rendition): ``For a subgraph $H$ of $G$ and a set $S \subseteq V(G),$ we define $N_H(v) = N_G(v) \cap V(H)$ and $N_H(S) = \bigcup_{v \in S}N_H(v).$''}
A vertex $v \in V(G)$ is \emph{isolated} if $N_G(v) = \emptyset.$
For $S \subseteq V(G),$ we set $G[S]=(S,E\cap{S \choose 2} )$ and use  $G \setminus S$ to denote $G[V(G) \setminus S].$
Given an edge $e=\{u,v\}\in E(G),$ we define the {\em subdivision} of $e$ to be the operation of deleting $e,$ adding a new vertex $w,$ and making it adjacent to $u$ and $v.$
Given two graphs $H,G,$ we say that $H$ is a {\em subdivision} of $G$ if $H$ can be obtained from $G$ by subdividing edges.
The \emph{contraction} of an edge $e = \{u,v\}$ of a simple graph $G$ results in a simple graph $G'$ obtained from $G \setminus \{u,v\}$ by adding a new vertex $uv$ adjacent to all the vertices in the set $N_G(u) \cup N_G(v)\setminus \{u,v\}.$
A graph $G'$ is a \emph{minor} of a graph $G$ if $G'$ can be obtained from a subgraph of $G$ after a series of edge contractions. The \emph{distance} between two vertices $x$  and $y$ of a graph $G$ is the number of edges of a shortest path between $x$ and $y$ in $G.$
%Given two subgraphs $H_{1}$ and $H_{2}$
%of a graph $H$ we define the {\em distance} (in $H$) between $H_{1}$ and $H_{2}$ as the minimum distance
%between a vertex in $H_{1}$ and a vertex in $H_{2}.$ \ig{check}

\paragraph{Treewidth.}
Let $G=(V,E)$ be a graph. A {{\em tree decomposition}} of $G$ is a pair $(T,
	\mathcal{ X}=\{X_{t}\}_{t\in V(T)})$ where $T$ is a tree and ${\cal X}$ is a collection of subsets of $V$
such that
\begin{itemize}
	\setlength{\itemsep}{-0.2mm}
	\item   $\bigcup_{t\in V(T)}X_t=V,$
	\item $\forall {e=\{u,v\}\in E},~\exists {t\in V(T)} : \{u,v\}\subseteq X_{t},$   and
	\item $\forall {v\in V}, \ T[\{t\mid v\in X_{t}\}]$ is connected.
	      %\jul{With this ``non-empty'', bullet $1$ is useless.}
\end{itemize}
%\begin{itemize}
%\item $\cup_{u\in V_T}=V$
%\item $\forall_{e\in E}\ \exists_{t\in V_T} : e\subseteq X_{t}$
%\item $\forall_{v\in V}\ T[\{t\mid v\in X_{t}\}]$ is connected.
%\end{itemize}
%\vspace{-3mm}

\noindent We call the vertices of $T$ {{\em nodes}} and
the sets in $\mathcal{ X}$ {{\em bags}} of the tree decomposition $(T,{\cal X}).$
The {\em width} of $(T,{\cal X})$ is equal to $\max_{}\{|X_t|-1\mid {t\in V(T)}\}$ and the
	{{\em treewidth}} of $G$ is the
minimum width over all tree decompositions of $G.$  We denote  the treewidth of a graph $G$ by $\tw(G).$

For $t\in\Bbb{N},$ we say that a set $S\subseteq V(G)$ is a {\em $t$-treewidth modulator} of $G$ if $\tw(G\setminus S)\leq t.$

%The set ${\sf exc}_{\sf m}(\Fcal)$ is defined analogously.
%% For example $\ex(\{K_{5},K_{3,3}\})$ is the class of all planar graphs.
%
%\paragraph{Parameterized complexity.} We refer the reader to~\cite{DF13,CyganFKLMPPS15} for basic background on parameterized complexity, and we recall here only some very basic definitions.
% , with special emphasis on tools for polynomial kernelization.
%A \emph{parameterized problem} is a language $L \subseteq \Sigma^* \times \mathbb{N}.$  For an instance $I=(x,k) \in \Sigma^* \times \mathbb{N},$ $k$ is called the \emph{parameter}. %Given a classical (non-parameterized) decision problem $L_{c} \subseteq \Sigma^*$ and a function $\kappa: \Sigma^* \rightarrow \mathbb{N},$ we denote by
%% $L_{c}/\kappa = \{(x,\kappa(x)\} \mid x \in L_{c}\}$ the associated parameterized problem. \ig{really?}
%A parameterized problem is \emph{fixed-parameter tractable} ({\sf FPT}) if there exists an algorithm $\Acal,$ a computable function $f,$ and a constant $c$ such that given an instance $I=(x,k),$
%$\Acal$ (called an {\sf FPT} \emph{algorithm}) correctly decides whether $I \in L$ in time bounded by $f(k) \cdot |I|^c.$

\subsection{Formal definition of the problem}

Let $\Fcal$ be a finite non-empty collection of non-empty graphs; we call such a collection {\em proper}.
%For any proper collection ${\cal F},$ we  define
%${\sf size}({\cal F})=\max\{\{|V(H)|\mid H\in \cal F\}\cup\{|{\cal F}|\}\}.$
%A proper collection is called \red{\em connected} if all its graphs are connected.
We extend the minor relation to $\Fcal$ such that, given a graph $G,$
%\sed{Footnote about connectivity.}
$\Fcal \preceq_{\sf m} G$ if and only if there exists a graph $H \in \Fcal$ such that $H \preceq_{\sf m} G.$
We also denote ${\sf exc}_{\sf m}(\Fcal)=\{G\mid \Fcal\npreceq_{\sf m} G \},$ i.e.,
${\sf exc}_{\sf m}(\Fcal)$ is the class of graphs that do {\sl not} contain any graph in $\Fcal$ as a  minor.

Let ${\cal F}$ be a proper  collection.
We define the graph parameter ${\bf m}_{\cal F}$ as the function that maps graphs to non-negative integers as follows:
\begin{eqnarray*}
	\label{label_romantically}
	{\bf m}_{\cal F}(G) & = & \min\{|S|\mid  S\subseteq V(G)\wedge G\setminus S\in{\sf exc}_{\sf m}({\cal F})\}.
\end{eqnarray*}
%The parameter ${\bf m}_{\cal F}$ is defined analogously.
The main objective of this paper is to study the problem of computing the parameter ${\bf m}_{\cal F}$  for graphs of bounded treewidth. The corresponding decision problem is  formally defined as follows.

%\hspace{-1cm}
\begin{center}
	\begin{minipage}{.5cm}
	\end{minipage}
	\begin{minipage}{7cm}
		\paraprobl{\textsc{$\Fcal$-M-Deletion}}
		{A graph $G$ and an integer $k\in \Nbb.$}
		{Is ${\bf m}_\Fcal(G)\leq k$?}
		% {$w = \tw(G).$}
		{The treewidth of $G.$}{7}
	\end{minipage}
\end{center}

\subsection{Boundaried graphs, folios, and representatives}

\paragraph{Boundaried graphs.} Let $t\in\Bbb{N}.$
A \emph{$t$-boundaried graph} is a triple $\bound{G} = (G,B,\rho)$ where $G$ is a graph, $B \subseteq V(G),$ $|B| = t,$ and
$\rho : B \to [t]$ is a
bijection.
We  say that  two $t$-boundaried graphs ${\bf G}_1=(G_1,B_1,\rho_1)$ and ${\bf G}_{2}=(G_2,B_2,\rho_2)$
are {\em isomorphic} if there is an isomorphism from $G_{1}$ to $G_{2}$
that extends the bijection $\rho_{2}^{-1}\circ \rho_{1}.$
The triple $(G,B,\rho)$ is a {\em boundaried graph} if it is a $t$-boundaried graph for some $t\in\Bbb{N}.$ As in~\cite{RobertsonS95b}, we define the {\em detail} of a boundaried graph  $\bound{G} = (G,B,\rho)$ as  ${\sf detail}(G):=\max\{|E(G)|,|V(G)\setminus B|\}.$
%
%
%the minimum $d$ such that
%$G$ has at  most $d$ edges with at  most one endpoint in $B$ and $|V(G)\setminus B|\leq d.$
We denote by ${\cal B}^{(t)}$ the set of all (pairwise non-isomorphic)  $t$-boundaried graphs  and by ${\cal B}_{h}^{(t)}$ the set of all (pairwise non-isomorphic)  $t$-boundaried graphs with detail at most $h.$
We also set ${\cal B}=\bigcup_{t\in\Bbb{N}}{\cal B}^{(t)}.$
\vspace{-1mm}

\paragraph{Minors and topological minors of boundaried graphs.} We say that a $t$-boundaried graph ${\bf G}_{1}=(G_1,B_1,\rho_1)$ is a {\em minor} of a $t$-boundaried graph ${\bf G}_{2}=(G_2,B_2,\rho_2),$ denoted by ${\bf G} _{1}\prem{\bf G}_{2},$
if there is a sequence of  removals of non-boundary vertices, edge removals, and edge contractions in $G_2,$ disallowing  contractions of edges with both endpoints  in $B_{2},$ that  transforms ${\bf G}_{2}$ to a boundaried graph that is isomorphic to ${\bf G}_{1}$ (during edge contractions, boundary vertices prevail). Note that this extends the usual definition of minors in graphs without boundary.

We say that $(M,T)$   is a {\em {\sf tm}-pair} if $M$ is  a graph, $T\subseteq V(M),$ and  all vertices in
$V(M)\setminus T$ have degree two. We denote by ${\sf diss}(M,T)$ the graph obtained
from  $M$ by \emph{dissolving} all vertices  in $V(M)\setminus T,$ that is, for every vertex $v \in V(M)\setminus T,$ with neighbors $u$ and $w,$ we delete $v$ and, if $u$ and $w$ are not adjacent, we add the edge $\{u,w\}.$
A {\em {\sf tm}-pair} of a graph $G$  is a  {\em {\sf tm}-pair}  $(M,T)$ where
$M$ is a subgraph of $G.$

Given two graphs $H$ and $G,$ we say that a {\sf tm}-pair $(M,T)$ of $G$  is a {\em topological minor model of $H$ in $G$} if $H$ is isomorphic to ${\sf diss}(M,T).$ We denote this isomorphism by $\sigma_{M,T}: V(H)\to T.$
We call the vertices in $T$ {\em branch} vertices of $(M,T).$
We call each path in $M$ between two distinct branch vertices and with no internal branch vertices  a {\em subdivision path} of $(M,T)$ and the internal vertices of such paths, i.e., the vertices of $V(M)\setminus T,$ are the {\em subdivision} vertices of $(M,T).$
We also extend  $\sigma_{M,T}$ so to also map each $e=\{x,y\}\in E(H)$ to the subdivision
path of $M$ with endpoints $\sigma_{M,T}(x)$ and $\sigma_{M,T}(y).$
Furthermore, we extend $\sigma_{M,T}$ so to also map each subgraph $H'$ of $H$
to the subgraph of $M$ consisting of the vertices of $\sigma_{M,T}(T)$ and the paths
in $\sigma_{M,T}(e), e\in E(H').$

If ${\bf M}=(M,B,\rho)\in{\cal B}$ and   $T\subseteq V(M)$ with $B\subseteq T$ and such that all vertices in
$V(M)\setminus T$ have degree two, we  call  $({\bf M},T)$ a {\em {\sf btm}-pair}
and we  define  ${\sf diss}({\bf M},T)=({\sf diss}(M, T),B,\rho).$ Note that we do not permit dissolution of boundary vertices, as we consider all of them to be branch vertices. If ${\bf G}=(G,B,\rho)$ is a boundaried graph and $(M,T)$ is a  {\sf tm}-pair of $G$
where $B\subseteq T,$  then we say that
$({\bf M},T),$ where ${\bf M}=(M,B,\rho),$ is a   {\em {\sf btm}-pair of} ${\bf G}=(G,B,\rho).$
Let ${\bf G}_{i}=(G_{i},B_{i},\rho_{i}), i\in[2].$ We say that ${\bf G}_{1}$ is a {\em topological minor}
of ${\bf G}_{2},$ denoted by ${\bf G}_{1}\pretp{\bf G}_{2},$ if
${\bf G}_{2}$ has a {\sf btm}-pair $({\bf M},T)$
such that  ${\sf diss}({\bf M},T)$ is isomorphic to ${\bf G}_{1}.$

Given a ${\bf G}=(G,B,\rho)\in{\cal B},$ we define ${\sf ext}({\bf G})$ as the set
containing every topological-minor-minimal
boundaried graph ${\bf G'}=(G',B,\rho)$ among those that contain ${\bf G}$ as a minor.
Notice that we insist that $B$ and $\rho$ are the same for all graphs in  ${\sf ext}({\bf G}).$
{Moreover, we do not consider isomorphic boundaried graphs in ${\sf ext}({\bf G})$  as different boundaried  graphs.} %\sed{Is this correct?}
The set ${\sf ext}({\bf G})$ helps us to express the minor relation in terms
of the topological minor relation because of the following simple observation. Note that this definition extends naturally to graphs, seen as boundaried graphs with empty boundary.
%The definition of ${\sf ext}$ applies also to graphs, by treating them as boundaried graphs
%of empty boundary. Also, given a graph class ${\cal F}$ we define ${\sf ext}({\cal F})=\{{\sf ext}(H)\mid H\in {\cal F}\}.$

\begin{observation}\label{label_daskalojannes}
	%\label{label_interessiert}
	If ${\bf G}_{1},{\bf G}_{2}\in{\cal B},$ then
	${\bf G}_{1}\prem {\bf G}_{2}\iff \exists {\bf G}\in {\sf ext}({\bf G}_{2}): {\bf G}_{1}\pretp {\bf G}.$
	Moreover,
	if  ${\bf G}$ is  a boundaried graph with detail $h,$ then  every graph in ${\sf ext}({\bf G})$
	has detail at most $3h.$
\end{observation}

\vspace{-3mm}

\paragraph{Folios.}
We define the {\em $h$-folio} of ${\bf G}=(G,B,\rho)\in {\cal B}$ as
$${h}\mbox{\sf-folio}({\bf G})=\{{\bf G}'\in {\cal  B} \mid {\bf G}'\pretp {\bf G} \mbox{~and ${\bf G}'$ has detail at most $h$}\}.$$

Using the fact that an $h$-folio is a collection of $K_{h+1}$-minor-free boundaried graphs, it follows that the   $h$-folio
of a $t$-boundaried graph has at most  $2^{\Ocal((h+t)\cdot \log (h+t)}$ elements. Therefore, the number of distinct $h$-folios of $t$-boundaried graphs is given by the following lemma (also observed in~\cite{BasteST20-monster1}).

\begin{lemma}
	\label{label_systemizations}
	There exists a function $\newfun{label_distribuendo}:\Bbb{N}^{2}\to\Bbb{N}$
	such that for every $t,h\in\Bbb{N},$ $|\{{h}\mbox{\sf -folio}({\bf G})\mid {\bf G}\in {\cal B}_{h}^{(t)}\}|\leq \funref{label_distribuendo}(t,h).$ Moreover,  $\funref{label_distribuendo}(t,h)=2^{2^{\Ocal((h+t)\cdot \log (h+t))}}.$
\end{lemma}

\vspace{-4mm}

\paragraph{Equivalent boundaried graphs and representatives.}
We say that two boundaried graphs ${\bf G}_{1}=(G_1,B_1,\rho_1)$ and ${\bf G}_{2}=(G_2,B_2,\rho_2)$ are {\em compatible} if $\rho_{2}^{-1}\circ \rho_{1}$  is an isomorphism from $G_{1}[B_{1}]$ to $G_{2}[B_{2}].$
Given two compatible boundaried graphs ${\bf G}_{1}=(G_1,B_1,\rho_1)$
and  ${\bf G}_{2}=(G_2,B_2,\rho_2),$  we
define ${\bf G}_{1}\oplus{\bf G}_{2}$ as the graph obtained
if we take the disjoint union of $G_{1}$ and $G_{2}$
and, for every $i\in[|B_{1}|],$ we identify vertices $\rho_{1}^{-1}(i)$ and $\rho_{2}^{-1}(i).$

Given  $h\in \Bbb{N},$ we say that two  boundaried graphs ${\bf G}_{1}$ and ${\bf G}_{2}$ are {\em $h$-equivalent}, denoted by ${\bf G}_{1}\equiv_{h} {\bf G}_{2},$
if they are compatible and, for every graph $H$ on at most $h$ vertices and $h$ edges
%\ig{we don't need to say ``of detail at most $h$'', right?}
and every boundaried graph ${\bf F}$ that is compatible  with ${\bf G}_{1}$ (hence,  with ${\bf G}_{2}$ as well), it holds that\vspace{-1mm}
\begin{eqnarray}
	H\prem{\bf F}\oplus {\bf G}_{1}\iff H\prem{\bf F}\oplus {\bf G}_{2}.\label{label_unquenchable}
\end{eqnarray}\vspace{-6mm}

%Notice that, for every $t\in\Bbb{N},$  the relation $\equiv_{h}$ partitions  ${\cal B}_{t}$
%into a bounded number of equivalence classes.

Note that  $\equiv_{h}$ is an equivalence relation on ${\cal B}.$
A minimum-sized (in terms of number of vertices) element of an  equivalence class of  $\equiv_{h}$
is called {\em representative}
of $\equiv_{h}.$ For $t\in\Bbb{N},$ a \emph{set of $t$-representatives} for $\equiv_{h}$ is a collection
containing a minimum-sized  representative for each equivalence class of $\equiv_{h}$ restricted to ${\cal B}^{(t)}.$
% that consists of  $t$-boundaried graphs.
Given $t,h\in \Bbb{N},$ we denote by  ${\cal R}_{h}^{(t)}$  a set of $t$-representatives  for $\equiv_{h}.$

At this point, we wish to stress that the folio-equivalence defined in \autoref{label_unquenchable} is related
but is {\sl not} the same as the one defined by ``having the same $h$-folio’’.
Indeed, observe first that if ${\bf G}_{1}$ and ${\bf G}_{2}$ are compatible $t$-boundaried graphs
and ${h}\mbox{\sf-folio}({\bf G}_1)={h}\mbox{\sf-folio}({\bf G}_{2})$
then ${\bf G}_{1}\equiv_{h} {\bf G}_{2},$ therefore  the folio-equivalence is a refinement of $\equiv_{h}.$
In fact, a dynamic programming procedure for  solving \textsc{$\Fcal$-M-Deletion}  can also be based on the folio-equivalence, and this has already been done in the general algorithm in~\cite{BasteST20-monster1}, which  has a double-exponential
parametric dependence due to the bound of~\autoref{label_systemizations}.
In this paper we build our dynamic programming
on   the equivalence $\equiv_{h}$  and we essentially prove that   $\equiv_{h}$
is ``coarse enough'' so to reduce the double-exponential
parametric dependence of the dynamic programming  to a single-exponential one.
In fact, this has already been done in~\cite{BasteST20-monster1} for the case where ${\cal F}$
contains some planar graph, as this structural
restriction directly implies an upper  bound on the treewidth of the
representatives. To deal with the  general case, the only
structural restriction for the (non-trivial) representatives
is the exclusion of $H$ as a minor.  All the combinatorial machinery that we
introduce in the next two sections is  intended  to deal with the structure of this
general and (more entangled) setting.

\section{Flat walls}
\label{label_meditatively}

In this section we deal with flat walls. More precisely, in Subsections~\ref{label_traditionalists},~\ref{label_conveniences}, and~\ref{label_unenforceability} we give the definition of a {\sl flat wall} in the form of a {\sl flatness pair}.
In \autoref{label_comicotragical} we define the notion of {\sl regular} flatness pair
and we give a  version of the Flat Wall Theorem of Robertson and Seymour~\cite{RobertsonS95b} that has been recently proved in~\cite{accurate}.
This version  (\autoref{label_intercanvien})  incorporates the
regularity property and  is based on the recent results and the  terminology of Kawarabayashi et al.~\cite{KawarabayashiTW18}. In \autoref{label_rencontroient}
we define a notion of homogeneity of flat walls, also introduced in \cite{accurate}, that along with  \autoref{label_intercanvien}  will be the combinatorial framework for the proofs of \autoref{label_donnescamente}. We stress that the notion of  homogeneity
that we use is different from that defined by Robertson and Seymour in~\cite{RobertsonS95b} and can serve as an alternative for further applications
based on the technology of flat walls (see e.g.~\cite{ICALP-versions,SauST21kapiII,SauST21kapiI}).
In \autoref{label_youthfulness} we define a graph parameter related to flat walls and show that it enjoys a series of properties related to Bidimensionality
(as introduced in~\cite{DemaineFHT05sube} and further developed in~\cite{F.V.Fomin:2010oq}).

\subsection{Walls and subwalls}
\label{label_traditionalists}

We first introduce some basic concepts such as partially disk-embedded graphs, walls, subwalls, tilts, and layers (for an example of all the concepts defined in this subsection, see~\autoref{label_parasitarian}).

\paragraph{Partially disk-embedded graphs.} %\label{app_pde}
A {\em closed} (resp. {\em open}) {\em disk} is a set homeomorphic to the set $\{(x,y)\in \Bbb{R}^{2}\mid x^{2}+y^{2}\leq 1\}$ (resp. $\{(x,y)\in \Bbb{R}^{2}\mid x^{2}+y^{2}< 1\}$).
Let $\Delta$ be an open or  closed disk.
We use $\bd(\Delta)$ to denote the boundary of $\Delta$ and, %\sed{Boundary? $\overline{Δ}$}
if $\Delta$ is closed, we use $\inter(\Delta)$ to denote the open disk $\Delta\setminus \bd(\Delta).$ Also, if  $\Delta$ is an open disk, we use $\overline{Δ}=Δ\cup\bd(Δ)$ for the closure of $\Delta.$
When we embed a graph $G$ in the
plane or in a disk, we treat $G$ (both its vertex and edge sets) as a set of points. This permits us to make
set operations  between graphs and sets of points.

If $\Delta$ is a closed disk, we say that a graph $G$ is {\em $\Delta$-embedded} if $G$ is embedded in $\Delta$ without crossings such that the intersection of $\bd(\Delta)$ and $G$ (seen as a set of points of $\Delta$)  is a subset of $V(G).$
We say that a graph $G$ is {\em partially disk-embedded in some closed disk $\Delta$},
if there is some $\Delta$-embedded subgraph, say  $K,$ of $G$
%and contains a cycle $D$
such that  $G\cap \Delta=K$ and $(V(G)\cap \Delta,V(G)\setminus\inter(\Delta))$
is a separation of $G.$ From now on, we use the term {\em partially $\Delta$-embedded graph $G$}
to denote that a graph $G$ is  partially disk-embedded in some closed disk $\Delta.$
%We also call the graph $K$
%{\em compass}
%of the partially $\Delta$-embedded graph $G$ and we always assume that we accompany
%a partially $\Delta$-embedded graph $G$ together with an embedding of its compass in $\Delta$ that is the set $G\cap \Delta.$

A {\em circle} of $\Delta$ is any set  homeomorphic to
$\{(x,y)\in \Bbb{R}^{2}\mid x^{2}+y^{2}= 1\}.$
Given two distinct points $x,y\in \Delta,$ an {\em $(x,y)$-arc} of $\Delta$ is any subset of $\Delta$ that is homeomorphic to the closed interval $[0,1].$

\begin{figure}[h]
	\begin{center}
		\includegraphics[width=12cm]{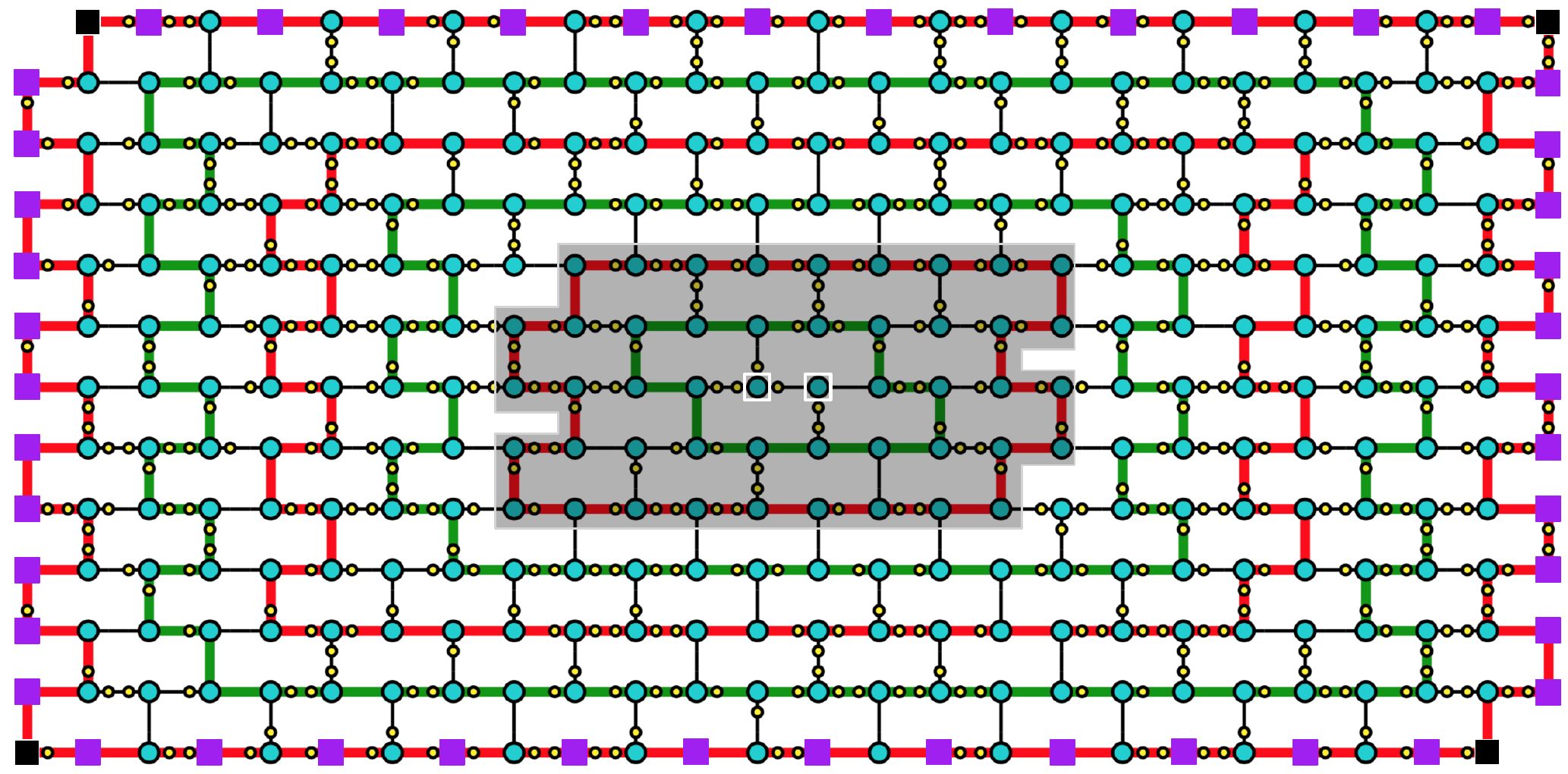}
	\end{center}
	\caption{An $13$-wall $W$ along with a choice of pegs and corners.
		The six layers of $W$ are colored alternatively in red and green and the central $5$-subwall of $W$ appears  in grey. The pegs are the squared vertices while, among them, those that are black are  the corners.
		The original vertices that are not pegs are turquoise circles while the subdivision vertices are the small yellow circles. the central vertices are the two 3-branch vertices that are surrounded by white squares. Notice that $W$ has 144 bricks and, among them, 100 are internal.
	}
	\label{label_parasitarian}
\end{figure}

\paragraph{Walls.}\label{label_totalitarians}
Let  $k,r\in\Bbb{N}.$ The
\emph{$(k\times r)$-grid} is the
graph whose vertex set is $[k]\times[r]$ and two vertices $(i,j)$ and $(i',j')$ are adjacent if $|i-i'|+|j-j'|=1.$
An  \emph{elementary $r$-wall}, for some odd integer $r\geq 3,$ is the graph obtained from a
$(2 r\times r)$-grid
% \ig{grids are defined only in \autoref{label_monosyllable}. We should either move the definition here, or to cite appropriately where the definition can be found}
with vertices $(x,y)
	\in[2r]\times[r],$
after the removal of the
``vertical'' edges $\{(x,y),(x,y+1)\}$ for odd $x+y,$ and then the removal of
all vertices of degree one.
Notice that, as $r\geq 3,$  an elementary $r$-wall is a planar graph
that has a unique (up to topological isomorphism) embedding in the plane $\Bbb{R}^{2}$
such that all its finite faces are incident to exactly six
%\ig{we write ``six'' here, but ``2'' above, we have to unify. I prefer to use letters for integers up to ten}
edges.
The {\em perimeter} of an elementary $r$-wall is the cycle bounding its infinite face, while the cycles bounding its finite faces are called {\em bricks}.
Also, the vertices
in the perimeter of an elementary $r$-wall that have degree two are called {\em pegs},
while the vertices $(1,1), (2,r), (2r-1,1),$ and $(2r,r)$ are called {\em corners} (notice that the corners are also pegs).

An {\em $r$-wall} is any graph $W$ obtained from an elementary $r$-wall $\bar{W}$
after subdividing edges. A graph $W$ is a {\em wall} if it is an $r$-wall for some odd $r\geq 3$
and we refer to $r$ as the {\em height} of $W.$ Given a graph $G,$
a {\em wall of} $G$ is a subgraph of $G$ that is a wall.
We
insist that, for every $r$-wall, the number $r$ is always odd.
%\red{: for this, whenever an $r$-wall appears with $r$ even, we agree to round it up to the next odd $r+1.$}

%\sed{Branch becomes 3-branch}
We call the vertices of degree three of a wall $W$ {\em 3-branch vertices}. The vertices that are created by subdivisions are called {\em subdivision vertices} while the rest are called {\em original} vertices of $W.$
A cycle of $W$ is a {\em brick} (resp. the {\em perimeter}) of $W$ if its 3-branch vertices are the vertices of a brick (resp. the perimeter) of $\bar{W}.$
We
%denote by ${\cal C}(W)$ the set of all cycles of $W,$ by ${\sf bricks}(W)$ the set of all the bricks of $W$ \ig{used?},
%and we
use $D(W)$ in order to denote the perimeter of the  wall $W.$
A brick of $W$ is {\em internal} if it is disjoint from $D(W).$ Note that every wall $W$ has a unique (up to homeomorphism) embedding in the plane whose infinite face is bounded by the perimeter $D(W)$ of the wall. Each time we consider a plane-embedded wall, we consider  this embedding.
%This allows to define, for each cycle $C$ of $W,$ $\Delta_C$ as the closed disk bounded by $C$ disjoint from the infinite face.

\medskip

Given two vertices $x$ and $y$ of a plane graph $G,$ we define their {\em face-distance} in $G$
as the smallest integer $i$ such that there exists an
arc of the plane  (i.e., a subset homeomorphic to the interval $[0,1]$) between $x$ and $y$ that   does not cross the infinite face of the embedding, crosses no vertices of $G,$ and {intersects at most $i$ faces of  $G.$} Note that two distinct vertices of a plane wall $W$ are within face-distance one if and only if they belong to the same brick. Given two vertex sets $X,Y$ of a plane graph $G,$ we define the face-distance between $X$ and $Y$ as the minimum face-distance between a vertex in $X$ and a vertex in $Y.$

\paragraph{Subwalls.} Given an elementary $r$-wall $\bar{W},$ some $i\in \{1,3,\ldots,2r-1\},$ and $i'=(i+1)/2,$
the {\em $i'$-th  vertical path} of $\bar{W}$  is the one whose
vertices, in order of appearance, are $(i,1),(i,2),(i+1,2),(i+1,3),
	(i,3),(i,4),(i+1,4),(i+1,5),
	(i,5),\ldots,(i,r-2),(i,r-1),(i+1,r-1),(i+1,r).$
Also, given some $j\in[2,r-1]$ the {\em $j$-th horizontal path} of $\bar{W}$
is the one whose
vertices, in order of appearance, are $(1,j),(2,j),\ldots,(2r,j).$

A \emph{vertical} (resp. \emph{horizontal}) path of $W$ is one
that is a subdivision of a  vertical (resp. horizontal) path of $\bar{W}.$
%whose 3-branch vertices are the vertices of a vertical (resp. horizontal) path of $\bar{W}.$
{Notice that the perimeter of an $r$-wall $W$ is uniquely defined regardless of the choice of the elementary $r$-wall $\bar{W}.$}
%\gstam{I think this observation is redundant. \ig{I think it is not completely redundant, I would keep it}}
A {\em subwall} of $W$ is any subgraph $W'$ of  $W$
that is an $r'$-wall, with $r' \leq r$ and such the vertical (resp. horizontal) paths of $W'$ are subpaths of the
	{vertical} (resp. {horizontal}) paths of $W.$

\paragraph{Tilts.}
The {\em interior} of a wall $W$ is the graph obtained
from $W$ if we remove from it all edges of $D(W)$ and all vertices of $D(W)$ that have degree two in $W.$ Given two walls $W$ and $\tilde{W}$ of a graph $G,$ we say that $\tilde{W}$ is a {\em tilt} of $W$ if $\tilde{W}$ and $W$ have identical interiors.

\paragraph{Layers.}
The {\em layers} of an $r$-wall $W$  are recursively defined as follows.
The first layer of $W$ is its perimeter. For $i=2,\ldots,(r-1)/2,$ the $i$-th layer of $W$ is the $(i-1)$-th layer of the subwall $W'$ obtained from $W$ after removing from $W$ its perimeter and removing recursively all occurring vertices of degree one. The {\em central vertices} of an $r$-wall are its two 3-branch vertices  that do not belong to any of its layers. See \autoref{label_parasitarian} for an illustration of the notions defined above.
Given an $r$-wall $W$ and an odd  integer $q,$ where $3\leq q\leq r,$ the {\em central $q$-subwall of $W$} is the subwall of $W$ of height $q$ whose central vertices are the central vertices of $W.$

\subsection{Paintings and renditions}
\label{label_conveniences}

Before defining flat walls, we need to introduce paintings and renditions. Here we closely follow the terminology of~\cite{KawarabayashiTW18}.

\paragraph{Paintings.}
Let $\Delta$ be a closed disk.
%, i.e., a set homeomorphic to the set $\{(x,y)\in \Bbb{R}^{2}\mid x^{2}+y^{2}\leq 1\}.$
%Given a subset $X$ of $\Delta,$ we
%denote its closure by $\overline{X}$ and its boundary by $\bd(X).$
A {\em {$\Delta$}-painting} is a pair $\Gamma=(U,N)$
%\rev{Please consider adding a picture explaining the concepts of paintings and cells. \ig{when we will include the definitions here, it will be more clear}}
where
\begin{itemize}
	\item  $N$ is a finite set of points of $\Delta,$
	\item $N \subseteq U \subseteq \Delta,$ and
	\item $U \setminus  N$ has finitely many arcwise-connected  components, called {\em cells}, where, for every cell $c,$
	      \begin{itemize}
		      \item[$\circ$] the closure $\bar{c}$ of $c$
		            is a closed disk
		            and
		      \item[$\circ$]  $|\tilde{c}|\leq 3,$ where $\tilde{c}:=\bd(c)\cap N.$
	      \end{itemize}
\end{itemize}
We use the  notation $U(\Gamma) := U,$
$N(\Gamma) := N$  and denote the set of cells of $\Gamma$
by $C(\Gamma).$
%Given a cell $c\in C(\Gamma)$  we will call the points in $\bd(c)\cap N$ {\em endpoints} of $c.$
For convenience, we may assume that each cell  of $\Gamma$ is an open disk of $\Delta.$ See \autoref{label_zahlensystem} for an example of a $\Delta$-painting.

\begin{figure}[ht]
	\begin{center}
		\includegraphics[width=11cm]{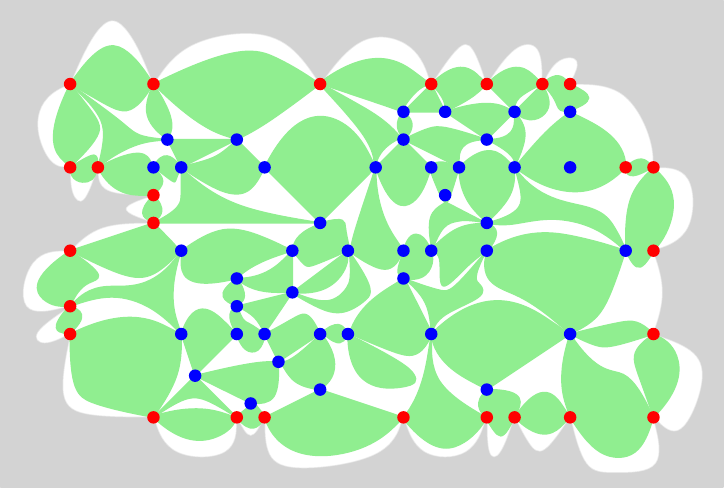}
	\end{center}
	\caption{A  $\Delta$-painting $\Gamma=(U,N).$ The red circles are the points of $N$ that are points of the boundary of $\Delta$ (whose complement is drawn in grey) and the blue circles are those that lie in the interior of $\Delta.$ The set $U \setminus  N$ is depicted in green.}
	\label{label_zahlensystem}
\end{figure}

Notice that, given a $\Delta$-painting $\Gamma,$
the pair $(N(\Gamma),\{\tilde{c}\mid c\in C(\Gamma)\})$  is a hypergraph whose hyperedges have cardinality at most three and  $\Gamma$ can be seen as a plane embedding of this hypergraph in $\Delta.$

\paragraph{Renditions.} Let $G$ be a graph, and let $\Omega$ be a cyclic permutation of a subset of $V(G)$ that we denote by $V(\Omega).$ By an {\em $\Omega$-rendition} of $G$ we mean a triple $(\Gamma, \sigma, \pi),$ where
%\ig{I have put these three properties in itemize environment}
\begin{itemize}
	\item[(a)] $\Gamma$ is a $\Delta$-painting for some closed disk $\Delta,$
	\item[(b)] $\pi: N(\Gamma)\to V(G)$ is an injection, and
	\item[(c)] $\sigma$ assigns to each cell $c \in  C(\Gamma)$ a subgraph $\sigma(c)$ of $G,$ such that
	      \begin{enumerate}
		      \item[(1)] $G=\bigcup_{c\in C(\Gamma)}\sigma(c),$
		            %\gstam{Should we use default LIPIcs enumeration here?}
		      \item[(2)]  for distinct $c, c' \in  C(\Gamma),$  $\sigma(c)$ and $\sigma(c')$  are edge-disjoint,
		      \item[(3)] for every cell $c \in  C(\Gamma),$ $\pi(\tilde{c}) \subseteq V (\sigma(c)),$
		      \item[(4)]  for every cell $c \in  C(\Gamma),$  $V(\sigma(c)) \cap \bigcup_{c' \in  C(\Gamma) \setminus  \{c\}}V(\sigma(c')) \subseteq \pi(\tilde{c}),$ and
		      \item[(5)]  $\pi(N(\Gamma)\cap \bd(\Delta))=V(\Omega),$ such that the points in $N(\Gamma)\cap \bd(\Delta)$ appear in $\bd(\Delta)$ in the same ordering as their images, via $\pi,$ in $\Omega.$
	      \end{enumerate}
\end{itemize}

Given an $\Omega$-rendition $(\Gamma, \sigma, \pi)$ of a graph $G,$ we call a cell $c$ of $\Gamma$ {\em trivial} if $\pi(\tilde{c})=V(\sigma(c)).$

\paragraph{Tight renditions.}
We say that an  {$\Omega$-rendition}  $(\Gamma, \sigma, \pi)$ of a graph $G$ is {\em tight} if the following conditions are satisfied:

\begin{enumerate}[label=(\roman*)]

	\item If there are two points $x,y$ of $N(\Gamma)$
	      such that $e=\{\pi(x),\pi(y)\}\in E(G),$ then
	      there is a cell $c\in C(\Gamma)$ such that $\sigma(c)$ is
	      the  two-vertex connected graph $(e,\{e\}),$

	\item\label{label_mantenimientos}  for every $c\in C(\Gamma),$ every two vertices in $\pi(\tilde{c})$ belong to some path of $\sigma(c),$

	\item\label{label_participates} for every $c \in  C(\Gamma)$ and every connected component $C$ of the graph
	      $\sigma(c)\setminus \pi(\tilde{c}),$  if $N_{\sigma(c)}(V(C))\neq\emptyset,$ then $N_{\sigma(c)}(V(C))=\pi(\tilde{c}),$

	\item  there are no two distinct non-trivial cells $c_{1}$ and $c_{2}$ such that  $\pi(\tilde{c_1})=\pi(\tilde{c_2}),$ and
	      %\sed{The last two have changed in relation to SODA 2020}

	\item\label{label_communication} for every $c \in  C(\Gamma)$ there are
	      $|\tilde{c}|$ vertex-disjoint paths in $G$ from $\pi(\tilde{c})$ to the set $V(\Omega).$
\end{enumerate}
%
%\ig{add appropriate labels}

As proved in~\cite{accurate}, it is possible to transform any  $\Omega$-rendition to a tight one.
For this reason, in this paper, we always assume that  $\Omega$-renditions are tight.

%
%In the rest of this paper we use only conditions (i)--(iii) of the tightness definition. However, we adopt the above, more strict, version of tightness as it will be useful in further applications.

\subsection{Flat walls and flatness pairs}
\label{label_unenforceability}

We are now in position to define the notion of a flat wall. We further
encode it into the concept of a {\sl flatness pair} of a graph.

\paragraph{Flat walls.} Let $G$ be a graph and let $W$ be an $r$-wall  of $G,$ for some odd integer $r\geq 3.$ We say that a pair $(P,C)\subseteq D(W)\times D(W)$ is a {\em choice
		of pegs and corners for $W$} if $W$ is the subdivision of an  elementary $r$-wall $\bar{W}$
where $P$ and
$C$ are the pegs and the corners of $\bar{W},$ respectively (clearly, $C\subseteq P$).
%{\gstam{I added this sentence here but maybe it is unnecessary. \ig{such sentences are really helpful}}
To get more intuition, notice that a wall $W$ can occur in several ways from the elementary wall $\bar{W},$ depending on the way the vertices in the perimeter of $\bar{W}$ are subdivided. Each of them %\sed{Height of a flatness pair}
gives a different selection $(P,C)$ of pegs and corners of $W$ (see~\autoref{label_parasitarian} for an example of a choice of pegs and conrers
$(P,C)$ in a 13-wall $W$).

We say that $W$ is a {\em flat $r$-wall}
of $G$ if there is a separation $(X,Y)$ of $G$ and a choice  $(P,C)$
of pegs and corners for $W$ such that:
\begin{itemize}
	\item $V(W)\subseteq Y,$
	\item  $P\subseteq X\cap Y\subseteq V(D(W)),$ and
	\item  if $\Omega$ is the cyclic ordering of the vertices $X\cap Y$ as they appear in $D(W),$ then there exists an $\Omega$-rendition $(\Gamma,\sigma,\pi)$ of  $G[Y].$
\end{itemize}

\begin{figure}[t]
	\begin{center}
		\includegraphics[width=11.3cm]{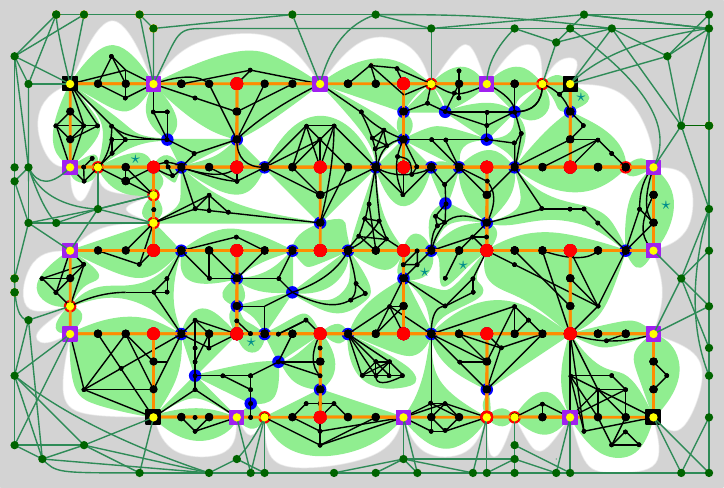}
	\end{center}
	\caption{A graph $G$ and a  flatness pair  $(W,\frak{R})$ of  $G$ where  $W$ is a 5-wall and $\frak{R}=(X,Y,P,C,\Gamma,\sigma,\pi)$ is a 7-tuple certifying the flatness of $W$ in $G.$ The edges of $W$
	are drawn in orange. In the corresponding separation $(X,Y),$
	the vertices of $X$ are green and yellow while the vertices in $Y$ are all the non-green vertices. Consequently, the yellow vertices are the vertices in $X\cap Y.$ The pegs and the corners are the squared vertices where the pegs that are not corners are purple and the the corners are black. The  $Δ$-painting of the $\Omega$-rendition $(\Gamma,\sigma,\pi)$ of $G’=G[Y]$ is the one depicted in \autoref{label_zahlensystem}.}
	\label{label_indubitadamente}
\end{figure}

\begin{figure}[ht]
	\begin{center}
		\includegraphics[width=11.3cm]{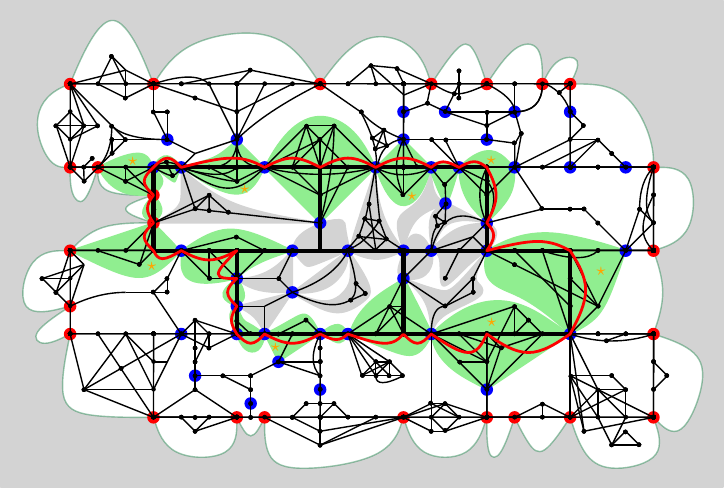}
	\end{center}
	\caption{The ${\frak{R}}$-compass of the $5$-wall $W$ for the flatness pair $(W,\frak{R})$ depicted in \autoref{label_indubitadamente}, and a subwall $W'$ of $W$ whose edges are depicted in bold. The red curve is the curve $K_{W'}.$ The $W'$-internal cells are depicted in grey while the $W'$-perimetric cells are depicted in green.  $W'$-marginal cells are marked with orange stars. The set ${\sf influence}_{\frak{R}}(W')$ contains all the flaps that are drawn inside the grey or the green cells (the $W'$-external cells are not depicted).}
	\label{label_disagguaglianza}
\end{figure}

\paragraph{Flatness pairs.}
Given the above, we  say that  the choice of the 7-tuple $\frak{R}=(X,Y,P,C,\Gamma,\sigma,\pi)$  {\em certifies 	that $W$ is a flat wall of $G$}. We call the pair $(W,\frak{R})$ a {\em flatness pair} of $G$ and define the {\em height} of the pair $(W,\frak{R})$ to be the height of $W.$
%	\red{If $(\Gamma,\sigma,\pi)$ is tight then we also say that $\frak{R}$ (and also $(W,\frak{R})$) is {\em tight}.}\gstam{This extra definition is not used anywhere.}
We use the term {\em cell of} $\frak{R}$ in order to refer to the cells of $\Gamma$ (see \autoref{label_indubitadamente} for an example of a  flatness pair  $(W,\frak{R})$ of  a graph).

We call the graph $G[Y]$ the {\em $\frak{R}$-compass} of $W$ in $G,$ denoted by ${\sf compass}_{\frak{R}}(W)$ (see \autoref{label_disagguaglianza} for the  $\frak{R}$-compass of $W,$ corresponding to the flatness pair $(W,\frak{R})$ of  \autoref{label_indubitadamente}). We define the  {\em flaps} of the wall $W$ in $\frak{R}$ as ${\sf flaps}_{\frak{R}}(W):=\{\sigma(c)\mid c\in C(\Gamma)\}.$  Given a flap $F\in {\sf flaps}_{\frak{R}}(W),$ we define its {\em base} as $\partial F:=V(F)\cap \pi(N(\Gamma)).$
A flap $F\in {\sf flaps}_{\frak{R}}(W)$ is {\em trivial}   if  $|\partial F|=2$ and $F$ consists of one edge between the two vertices in $\partial F.$
We call the edges of the trivial flaps {\em short edges of ${\sf compass}_{\frak{R}}(W)$}.
A  cell $c$ of ${\frR}$ is {\em untidy} if  $\pi(\tilde{c})$ contains a vertex
$x$ of ${W}$ such that two of the edges of ${W}$ that are incident to $x$ are edges of $\sigma(c).$ Notice that if $c$ is untidy then  $|\tilde{c}|=3.$
A cell is {\em tidy} if it is not untidy (in \autoref{label_indubitadamente} untidy cells are marked by green stars).

%
%In \red{\autoref{label_unreliability}} we depict a flat wall $W$ in a graph $G$ as well as the $\frak{R}$-compass of $W$ in $G,$  for some rendition $\frak{R}$ certifying its flatness.
%Notice that there is a unique subwall $W'$ of $W$
%that is disjoint from $D(W)$ and has height five.
%Interestingly, the subwall $W'$  is {\sl not} a flat wall of $G,$ however
%there is a tilt $\tilde{W}'$ of $W'$ that is a flat wall of $G.$
%$\tilde{W}'$ is \red{depicted in \autoref{label_unreliability}}  and the
%rendition certifying its flatness is \red{depicted in \autoref{label_frecuentemente}.}

\subsection{Influence, regularity,  and tilts of  flatness pairs}
\label{label_comicotragical}

We now introduce a classification of  the cells of a flatness pair $(W,\frak{R}).$ This classification will be used in order to define the concepts of regularity and $W'$-tilts of flatness pairs that will be important for our proofs.

\paragraph{Cell classification.}
Given a cycle $C$ of ${\sf compass}_{\frak{R}}(W),$ we say that
$C$ is {\em $\frak{R}$-normal} if it is not a subgraph of a flap $F\in {\sf flaps}_{\frak{R}}(W).$
Given an $\frak{R}$-normal cycle $C$ of ${\sf compass}_{\frak{R}}(W),$
we call a cell $c$ of $\frak{R}$ {\em $C$-perimetric} if   $\sigma(c)$ contains some edge of $C.$ Notice that if $c$ is $C$-perimetric, then $\pi(\tilde{c})$ contains two points $p,q\in N(\Gamma)$
such that  $\pi(p)$ and $\pi(q)$ are vertices of $C$ where one, say $P_{c}^{\rm in},$ of the two $(\pi(p),\pi(q))$-subpaths of $C$ is a subgraph of $\sigma(c)$ and the other, denoted by $P_{c}^{\rm out},$  $(\pi(p),\pi(q))$-subpath contains at most one internal vertex of $\sigma(c),$ which should be the (unique) vertex $z$ in $\partial\sigma(c)\setminus\{\pi(p),\pi(q)\}.$
We pick a $(p,q)$-arc $A_{c}$ in $\hat{c}:={c}\cup\tilde{c}$ such that  $\pi^{-1}(z)\in A_{c}$ if and only if $P_{c}^{\rm in}$ contains
the vertex $z$ as an internal vertex.

We consider the circle  $K_{C}=\cupall\{A_{c}\mid \mbox{$c$ is a $C$-perimetric cell of $\frak{R}$}\}$
and we denote by $\Delta_{C}$ the closed disk bounded by $K_{C}$  that is contained in  $\Delta.$
A cell $c$ of $\frak{R}$ is called {\em $C$-internal} if $c\subseteq \Delta_{C}$
and is called {\em $C$-external} if $\Delta_{C}\cap c=\emptyset.$
Notice that  the cells of $\frak{R}$ are partitioned into  $C$-internal,  $C$-perimetric, and  $C$-external cells.

Let $c$ be a tidy $C$-perimetric cell of $\frak{R}$ where $|\tilde{c}|=3.$ Notice that $c\setminus A_{c}$ has two arcwise-connected components and one of them is an open disk $D_{c}$ that is a subset of $\Delta_{C}.$
If the closure $\overline{D}_{c}$  of $D_{c}$ contains only two points of $\tilde{c}$ then we call the cell $c$ {\em $C$-marginal}.

%A cell of $\frak{R}$ is called {\em internal/marginal/external} if it is   $W$-internal/marginal/external.
\paragraph{Influence.}
For every $\frak{R}$-normal cycle $C$ of ${\sf compass}_{\frak{R}}(W)$ we define the set
$${\sf influence}_{\frak{R}}(C)=\{\sigma(c)\mid \mbox{$c$ is a cell of $\frak{R}$ that is not $C$-external}\}.$$
%
%
%graph $${\sf influence}_{\frak{R}}(W')=\cupall\{\pi(c)\mid \mbox{$c$ is a cell of $\frak{R}$ that is not $W'$-external}\}.$$

A wall $W'$  of ${\sf compass}_{\frak{R}}(W)$  is \emph{$\frak{R}$-normal} if $D(W')$ is  $\frak{R}$-normal.
Notice that every wall of $W$ (and hence every subwall of $W$) is an $\frak{R}$-normal wall of ${\sf compass}_{\frak{R}}(W).$ We denote by ${\cal S}_{\frak{R}}(W)$ the set of all $\frak{R}$-normal walls of ${\sf compass}_{\frak{R}}(W).$ Given a $W'\in {\cal S}_{\frak{R}}(W)$ and a cell $c$ of $\frak{R}$
we say that $c$ is {\em $W'$-perimetric/internal/external/marginal} if $c$ is  $D(W')$-perimetric/internal/external/marginal (see  \autoref{label_disagguaglianza} for an example). We also use $K_{W'},$ $\Delta_{W'},$ ${\sf influence}_{\frak{R}}(W')$ as shortcuts
for $K_{D(W')},$ $\Delta_{D(W')},$ ${\sf influence}_{\frak{R}}(D(W')).$

\paragraph{Regularity.}
Let $(W,\frak{R})$ be a flatness pair of a graph $G.$
We call a  flatness pair $(W,\frak{R})$ of a graph $G$ {\em regular}
if none of its cells is $W$-external, $W$-marginal, or untidy.
Notice that the flatness pair  of \autoref{label_indubitadamente}
is not regular (for an example of a regular flatness pair of a graph that is a  modification of the one in \autoref{label_indubitadamente}, see \autoref{label_grandiloquents}).
The notion of regularity has been defined in \cite{accurate}
and will be useful in \autoref{label_disintegrated}. In fact, regularity  permits the definition of
a ``well-allinged''  $\Delta$-embedded representation of the $\frak{R}$-compass
that will be valuable in the   proofs
of \autoref{label_preoccupation}.  The precise definition of the notion of well-allinged flatness pairs is  given in \autoref{label_disintegrated}.
\medskip

The next result has been proved in~\cite{accurate}. It can be seen as a version  of the Flat Wall Theorem incorporating the concept of regularity, which is necessary for our proofs.

%In \autoref{label_disagguaglianza} we depict the $W''$-marginal and untidy
%as examples, see the cells in that are marked with stars).

\begin{theorem}\label{label_intercanvien}
	%\sed{Make it non-algoritrhmic -- DONE!}
	There exist two functions $\newfun{label_congiugnersi},\newfun{label_scheinbeziehungen}:\Bbb{N}\to \Bbb{N}$  such that for every
	graph $G,$ every odd integer $r\geq 3,$ and every  $q\in\Bbb{N}_{\geq 1},$ one of the following is true:
	\begin{itemize}
		\item $K_{q}$ is a minor of $G,$
		\item $\tw(G) \leq \funref{label_congiugnersi}(q)\cdot r,$ or
		\item there exist a set $A\subseteq V(G),$  where $|A|\leq \funref{label_scheinbeziehungen}(q),$ and a regular flatness pair $(W,\frak{R})$ of $G\setminus A$ of height $r.$
		      %		      and a tree decomposition of the $\frak{R}$-compass of $W$ of width at most $\funref{label_congiugnersi}(q)\cdot r.$  Moreover, $\funref{label_congiugnersi}(q)=2^{\Ocal(q^2 \log q)}$ and $\funref{label_scheinbeziehungen}(q)= \Ocal(q^{24}).$
	\end{itemize}
	Moreover, $\funref{label_congiugnersi}(q)=2^{\Ocal(q^2 \log q)}$ and $\funref{label_scheinbeziehungen}(q)= \Ocal(q^{24}).$
\end{theorem}

\paragraph{Tilts of flatness pairs.} Let $(W,\frak{R})$ and $(\tilde{W}',\tilde{\frak{R}}')$  be two flatness pairs of a graph $G$ and let $W'\in {\cal S}_{\frak{R}}(W).$ We also assume that ${\frak{R}}=(X,Y,P,C,\Gamma,\sigma,\pi)$ and $\tilde{\frak{R}}'=(X',Y',P',C',\Gamma',\sigma',\pi').$
We say that   $(\tilde{W}',\tilde{\frak{R}}')$   is a {\em $W'$-tilt}
of $(W,\frak{R})$ if \begin{itemize}
	\item $\tilde{\frak{R}}'$ does not have $\tilde{W}'$-external cells,
	\item  $\tilde{W}'$ is a tilt of $W',$
	\item  the set of $\tilde{W}'$-internal  cells of  $\tilde{\frak{R}}'$ is the same as the set of $W'$-internal cells of ${\frak{R}}$ and their images via $\sigma'$ and ${\sigma}$ are also the same,
	\item ${\sf compass}_{\tilde{\frak{R}}'}(\tilde{W}')$ is a subgraph of $\cupall{\sf influence}_{{\frak{R}}}(W'),$ and
	\item if $c$ is a cell in $C(\Gamma') \setminus C(\Gamma),$ then $|\tilde{c}| \leq 2.$
\end{itemize}

The next observation follows from the definitions of regular flatness pairs and tilts.
\begin{observation}\label{label_satisfacerse}
	If $(W,\frak{R})$ is a regular flatness pair, then for every $W'\in {\cal S}_{\frak{R}}(W),$ every $W'$-tilt of $(W,\frak{R})$ is also regular.
\end{observation}

We need one more observation, which follows from the third item above and the fact that the cells corresponding to flaps containing a central vertex of $W'$ are all internal (recall that the height of a wall is always at least three).

\begin{observation}\label{label_tranquillity}
	The central vertices of $W'$ belong to every $W'$-tilt of $(W,\frak{R}).$
\end{observation}

The need to define $W'$-tilts of flatness pairs emerges from the fact that
not every subwall $W'$ of a flat wall $W$ is necessarily flat, recently observed in \cite{accurate}.
The next proposition, proved in \cite{accurate}, suggests that
there is always a {\sl slight modification} of $W'$ in the $\frak{R}$-compass
of $W$ that is indeed a flat wall. This ``tilt'' preserves the internal cells, and therefore the ``essential'' part of the influence of $W'.$ That way, it permits us to define
a notion of compass relative to a subwall of a flat wall.

\begin{proposition}
	\label{label_contradizion}
	For every   flatness pair $({W},{\frak{R}})$ of a graph $G$ and every $W'\in {\cal S}_{\frak{R}}(W),$ there exists a  flatness pair  $(\tilde{W}',\tilde{\frak{R}}')$ of $G$  that is a $W'$-tilt of $({W},{\frak{R}}).$
\end{proposition}

\subsection{Homogeneous walls}\label{label_rencontroient}

Homogeneous walls were a basic ingredient of the seminal  algorithm of Robertson and Seymour for the {\sc Disjoint Paths} problem in\cite{RobertsonS95b}. This algorithm introduced the {\sl Irrelevant Vertex Technique} that consisted in the identification of a vertex in an instance of  the {\sc Disjoint Paths} problem that is {\em irrelevant} in the sense that
its removal does not change the
	{\sc Yes}/{\sc No}-status of the instance.
The notion of wall homogeneity was given in
\cite{RobertsonS95b} and was based on the concept of the {\sl vision} of an ``internal'' flap of a flat wall.
It was proved in \cite{RobertsonSGM22}
that the central vertices of a sufficiently big homogenous flat wall are indeed {\sl irrelevant} with respect to the {\sc Disjoint Paths}  problem and therefore they could safely discarded.
Our results are following the same technique. However, we need an alternative notion of homogeneity that we introduce in this subsection.
\medskip

Let $G$ be a graph, let $A\subseteq V(G),$ and let $(W,\frak{R})$ be a  flatness pair of  $G\setminus A,$
% Let $G$ be a graph and let $(W,\frak{R})$ be a flatness pair of $G,$
where $\frak{R}=(X,Y,P,C,\Gamma,\sigma,\pi)$ and $(\Gamma,\sigma,\pi)$ is an $\Omega$-rendition of $G[Y].$ Recall that $\Gamma=(U,N)$ is a $\Delta$-painting of the closed disk $\Delta.$

\paragraph{Augmented flaps.}
For each flap $F\in {\sf flaps}_{\frak{R}}(W)$ we consider a labeling  $\ell_{F}: \partial F\rightarrow\{1,2,3\}$ such that
the set of labels assigned by $\ell_{F}$ to $\partial F$ is  one of $\{1\},$ $\{1,2\},$ $\{1,2,3\}.$
We also consider a bijection  $\rho_{A}: A\to [a],$ where $a = |A|.$
The labelings in ${\cal L}=\{\ell_{F} \mid F\in  {\sf flaps}_{\frak{R}}(W)\}$ and the labeling $\rho_{A}$ will be useful for defining a set of boundaried graphs that we will call augmented flaps. We first need some more definitions.

Given a flap $F\in{\sf flaps}_{\frak{R}}(W),$ we define an ordering
$\Omega(F)=(x_{1},\ldots,x_{q}),$ with $q\leq 3,$ of the vertices of $\partial{F}$
so that
\begin{itemize}
	\item $(x_{1},\ldots,x_{q})$ is a  counter-clockwise cyclic ordering of the vertices of $\partial F$ as they appear in the corresponding cell of $C(\Gamma).$ Notice that this cyclic ordering is significant  only when $|\partial F|=3,$
	      in the sense that $(x_{1},x_{2},x_{3})$ remains invariant under shifting, i.e., $(x_{1},x_{2},x_{3})$ is the same as $ (x_{2},x_{3},x_{1})$ but not  under inversion, i.e.,   $(x_{1},x_{2},x_{3})$ is not the same as $(x_{3},x_{2},x_{1}),$ and
	\item   for $i\in[q],$ $\ell_{F}(x_{i})=i.$
\end{itemize}
Notice that the second condition is necessary for completing the definition of the ordering $\Omega(F),$
and this is the reason why we set up the labelings in ${\cal L}.$\medskip\medskip

For each $F\in {\sf flaps}_{\frak{R}}(W)$ with $t_{F}=|\partial F|,$
we fix $\rho_{F}: \partial F\to [a+1,a+t_F]$ such that
$(\rho^{-1}_{F}(a+1),\ldots,\rho^{-1}_{F}(a+t_F))= \Omega(F).$
%Moreover, for each flap $F\in {\sf flaps}_{\frak{R}}(W),$
Also, we define the boundaried graph
\begin{equation}\label{label_disconvenevole}
	\textbf{F}^{A}=(G[A\cup F],A\cup \partial F,\rho_A\cup \rho_F)
\end{equation}
and we denote by $F^{A}$ the underlying graph of $\textbf{F}^{A}.$ We call $\textbf{F}^{A}$ an {\em augmented flap} of the flatness pair $(W,\frak{R})$ of $G\setminus A$
in $G.$

\paragraph{Palettes and homogeneity.}
For each cycle $C$ of $W,$ we define $(A,\ell)\mbox{\sf -palette}(C)=\{{\ell}\mbox{\sf -folio}({\bf F}^{A})\mid F\in {\sf  influence}_{\frak{R}}(C)\}.$
We say that the  flatness pair   $(W,\frak{R})$ of $G\setminus A$
is  {\em $\ell$-homogeneous  with respect to the pair $(G,A)$}  if  every  {\sl internal} brick $B$ of ${W}$   (seen as a cycle of ${W}$) has the {\sl same} $(A,\ell)$\mbox{\sf -palette}.

\paragraph{Apex-wall triples.}
Let $G$ be a graph, let $A \subseteq V(G)$ with $|A| \leq a,$ and let  $(W,\frak{R})$
be a regular flatness pair of $G \setminus A$ such that $W$ has height $r$ and is $\ell$-homogenous with respect to $(G,A)$ for some $\ell \in \Bbb{N}.$ We call such a triple $(A,W,\frak{R})$ an \emph{$(a,r,\ell)$-apex-wall triple of $G$}.
%\sed{Correct the use  of homogenous in triple}
\medskip

The next proposition, proved in~\cite{accurate}, implies  that
it is possible to find an $\ell$-homogeneous flat wall inside the compass of a sufficiently big flat wall.

\begin{proposition}
	\label{label_incominciando}
	There exists a function $\newfun{label_vulgairement}: \Bbb{N}^{2}\to\Bbb{N}$
	such that if $r\in \Bbb{N}_{\geq 3},$
	$G$ is a graph,  $A\subseteq V(G),$ and $(W,\frak{R})$  is a flatness pair  of  $G\setminus A$ of height $\funref{label_vulgairement}(r,w),$ where $w=\funref{label_distribuendo}(|A|+3,\ell),$ then $W$ contains some subwall $W'$ of height $r$ such that every $W'$-tilt of $(W,\frak{R})$ is $\ell$-homogeneous  with respect to $(G,A).$
	Moreover, $\funref{label_vulgairement}(r,w)={\cal O}(r^w).$
	%and ${\sf compass}_{\tilde{\frak{R}}'}(\tilde{W}')$ is a subgraph of ${\sf compass}_{\frak{R}}(W).$
\end{proposition}

\subsection{A parameter for affecting  flat walls}
\label{label_youthfulness}

We proceed to define a graph parameter that will be useful for our proofs. We prove that it satisfies
some properties related to Bidimensionality theory~\cite{F.V.Fomin:2010oq,DemaineFHT05sube,FominDHT16}
that will be used later in \autoref{label_sechseckigen}.\medskip

Let $G$ be a graph and let $(A,W,\frak{R})$ be an $(a,r,\ell)$-apex-wall triple of $G.$ We say that $S$ {\em affects}  $(A,W,\frak{R})$
if $N_{G}[V({\sf compass}_{\frak{R}}(W))]\cap (S \setminus A) \neq \emptyset.$ For $a,r,\ell\in\Bbb{N},$ we define
\begin{eqnarray*}
	{\bf p}_{a,r,\ell}(G) & = & \min\{k\mid \exists S\subseteq V(G): |S|\leq k \ \wedge\ \mbox{$S$ affects every $(a,r,\ell)$-apex-wall triple  of $G$}\}.
\end{eqnarray*}

Using \autoref{label_intercanvien}, \autoref{label_contradizion}, and
\autoref{label_incominciando}, we prove that the above parameter grows {\sl quadratically} with its treewidth.

\begin{lemma}%[{\sf BIDIMENSIONALITY}]
	\label{label_substantiality}
	There is a function $\newfun{label_superintendent}:\Bbb{N}^{3}\to\Bbb{N}$ such that  if $q,r,\ell \in \Bbb{N}_{>0},$ and  $G$ is a $K_{q}$-minor-free  graph,
	then  $\tw(G)\leq \funref{label_superintendent}(q,r,\ell)\cdot\max\big\{1,\sqrt{{\bf p}_{\funref{label_scheinbeziehungen}(q),r,\ell}(G)}\big\}.$
	In particular, one may choose $\funref{label_superintendent}(q,r,\ell)=\Ocal(\funref{label_congiugnersi}(q)\cdot r^w),$ where $w=\funref{label_distribuendo}(\funref{label_scheinbeziehungen}(q)+3,\ell)=2^{2^{\Ocal((q^{24}+\ell)\cdot \log (q^{24}+\ell))}}.$
	%\sed{Creo que esto es linear en $r,$ pero lo vermos despues de escribir la demo \ig{Sí!}}
	%\sed{Attention to the $\max$}
\end{lemma}

%\red{First \autoref{label_intercanvien}, then  \autoref{label_contradizion} to each subwall, and then \autoref{label_incominciando} to find an homogenous one}

%There exists an algorithm that given a graph $G,$ a (regular) flatness pair $({W},{\frak{R}})$ of $G,$ and a wall $W'\in {\cal S}_{\frak{R}}(W),$ outputs  a (regular) $W'$-tilt of $({W},{\frak{R}})$ in  ${\cal O}(n+m)$ time.
\begin{proof}
	The lemma follows easily if we prove that, for every positive integer $p,$   $\tw(G)> \funref{label_congiugnersi}(q)\cdot \big((r^w+2) \cdot \lceil\sqrt{p+1}\rceil+2\big)$ implies that  ${\bf p}_{\funref{label_scheinbeziehungen}(q),r,\ell}(G) > p.$

	By  \autoref{label_intercanvien}, it follows that, for any $q,r,\ell \in \Bbb{N}_{>0},$ if $G$ is $K_q$-minor-free and
	$\tw(G) > \funref{label_congiugnersi}(q)\cdot \big((r^w+2) \cdot \lceil\sqrt{p+1}\rceil+2\big),$ then	$G$ contains some vertex set $A$ such that  $|A|\leq \funref{label_scheinbeziehungen}(q)$ and $G\setminus A$ has a regular flatness pair $(W,\frak{R})$ of  height $(r^w+2) \cdot \lceil\sqrt{p+1}\rceil+2.$ Let $\hat{W}_{1},\ldots,\hat{W}_{p+1}'\in {\cal S}_{\frak{R}}(W)$ be a collection of $p+1$ pairwise disjoint subwalls of $W,$ each of height $r^w,$ such that there are no two vertices $w_{1},w_{2}$  in $W$ of face-distance at most one such that $w_{1}$ belongs to some $\hat{W}_{i}$ and $w_{2}$ belongs either in $D(W)$ or in some other $\hat{W}_{i'},$ $i'\neq i.$

	By \autoref{label_contradizion}, for every $i\in[p+1],$ there is a $\hat{W}_{i}$-tilt of $(W,\frak{R})$ that we denote by    $(W_{i},\frak{R}_{i}).$ Since $(W,\frak{R})$ is regular,
	\autoref{label_satisfacerse} implies that $(W_{i},\frak{R}_{i})$ is also regular for every $i\in[p+1].$  Moreover,  by the definition of a $\hat{W}_{i}$-tilt, it follows that,
	for every two distinct $i,i'\in[p+1],$\ $N_{G}[V({\sf compass}_{\frak{R}_{i}}(W_i))] \cap N_{G}[V({\sf compass}_{\frak{R}_{i}}(W_{i'}))] \subseteq A.$
	Note that this is correct because the face-distance demand  leaves a ``buffer'' among the
	flat walls and the perimeter of $W$ to guarantee that the neighborhoods of their compasses do not intersect, except possibly at apex vertices. Let $w=\funref{label_distribuendo}(|A|+3,\ell).$
	From \autoref{label_incominciando}, for each regular flatness pair $(W_{i},\frak{R}_{i})$ of $G\setminus A$ there is a subwall $\hat{W}_{i}'$ of height $r$ such
	that every regular $\hat{W}_{i}'$-tilt of $(W,\frak{R})$ is $\ell$-homogeneous  with respect to $(G,A).$ We denote this $\hat{W}_{i}'$-tilt by $(W_i',\frak{R}_i')$ and we conclude that  $(A,W_i',\frak{R}_i')$ is an $(a,r,\ell)$-apex-wall triple of $G,$ for every $i\in[p+1].$
	As before,  $N_{G}[V({\sf compass}_{\frak{R}_{i}'}(W_i'))] \cap N_{G}[V({\sf compass}_{\frak{R}_{i}'}(W_{i'}'))] \subseteq A$ for every $i,i' \in [p+1], i \neq i'.$ Therefore, every set $S \subseteq V(G)$ affecting  every $(\funref{label_scheinbeziehungen}(q),r,\ell)$-apex-wall triple of $G$ needs to contain at least one vertex from each of the sets $\{N_{G}[V({\sf compass}_{\frak{R}}(W_i))] \mid i \in [p]\},$ implying that ${\bf p}_{\funref{label_scheinbeziehungen}(q),r,\ell}(G) > p.$
	%
	%
	%Summarizing, we have proved that if $\tw(G)\geq \funref{label_quiriniennes}(q)\cdot \conref{sadfsdfsdfafdsf} \cdot (q \cdot (2r \cdot \sqrt{p} + q)),$ then ${\bf p}_{q-5,r}(G) \geq p.$
	%Therefore, the lemma follows, in particular, with $\funref{label_superintendent}(q,r)=2r \cdot \funref{label_quiriniennes}(q)\cdot \conref{sadfsdfsdfafdsf} \cdot q^2 = r\cdot 2^{\Ocal(q^2\log q)}.$ Note that we have used ``$\max\big\{1,\sqrt{{\bf p}_{q-5,r}(G)}\big\}$'' in the statement of the lemma in order to cover the case where ${\bf p}_{q-5,r}(G) = 0.$
\end{proof}

We now prove that the parameter ${\bf p}_{a,r,\ell}$ is \emph{separable}, that is, that when considering a separation of a graph, the value of the parameter is ``evenly'' split along both sides of the separation, possibly with an offset bounded by the order of the separation.

\begin{lemma}%[{\sf SEPARATION}]
	\label{label_concerniente}
	Let  $a,r,\ell\in \Bbb{N},$ let $G$ be a graph,  and let $S\subseteq V(G)$ such that $S$ affects  every $(a,r,\ell)$-apex-wall triple of $G.$
	Then, for every separation $(L,R)$ of $S$ in  $G,$ the set
	$L\cap (R \cup S)$ affects  every $(a,r,\ell)$-apex-wall triple of $G[L].$
	%\sed{Dangerous!!!! \ig{why?}}
	%\sed{Is homogeneity and regularity an issue here?}

\end{lemma}
\begin{proof}
	Suppose for contradiction that $(A,W,\frak{R})$ is an $(a,r,\ell)$-apex-wall triple of $G[L]$ that is not affected by $L\cap (R \cup S).$ In particular, it holds that $V({\sf compass}_{\frak{R}}(W)) \subseteq L \setminus R.$ Since by assumption $(A,W,\frak{R})$ is affected by $S$ but not by $L\cap (R \cup S),$ there should exist a vertex $v \in S \cap (R \setminus L)$ with a neighbor in $V({\sf compass}_{\frak{R}}(W)) \subseteq L \setminus R,$ contracting the hypothesis that $(L,R)$ is a separation of $G.$
\end{proof}

\section{Finding an irrelevant vertex}
\label{label_donnescamente}

%\ig{Dimitrios, please check if you like the following introduction to this section}
In this section we show how to find inside a sufficiently large flat wall of a boundaried graph ${\bf G}=(G,B,ρ)$
a flat subwall whose compass
is ``irrelevant'' with respect to the presence of a graph $H$ in ${\cal F}$ as a minor.
Here the term ``irrelevant'' is not only related to $G$ but  to {\sl every} graph ${\bf K}\oplus {\bf G}$ that can be obtained
by gluing ${\bf G}$ with another boundaried graph ${\bf K}.$
For this we need a stronger notion of irrelevancy, defined in \autoref{label_preoccupation}, that takes into account
only the ``essential part'' of a topological minor model $(M,T)$ of  $H$ that is ``invading'' ${\bf G}.$

We start in  \autoref{label_generalmente}, by detecting in every wall  a {\sl railed annulus}. This
structure, introduced  in~\cite{KaminskiT12} and reused later in~\cite{GolovachKMT17,GolovachST20-SODA},
turns out to be quite handy in order to guarantee a ``taming property'' of topological minor models (cf. \autoref{label_interessiert}). In \autoref{label_biologically}
we first use graph drawing tools to prove that we can assume that our model is embedded ``nicely'' inside a railed annulus, in the sense that certain vertices are sufficiently pairwise far apart (cf.~\autoref{label_dispenseront}); this will be helpful in order to reroute the model of every ``invading'' topological minor model $(M,T)$ of $H.$ With the help of \autoref{label_interessiert}, we prove (cf.~\autoref{label_constitutivos}) that,
given a partially disk-embedded graph that contains a railed annulus, the topological minor model $(M,T)$ of a graph $H$ can be rerouted so to obtain  another topological minor model that can be contracted back to $H$ and such that a ``large enough'' central region of the railed annulus is avoided.
The rerouting of $(M,T)$ will be done so
that a prescribed subset of degree-$3$ vertices of the original model will not be affected by contractions.

Once we have all the above  ingredients, we consider in \autoref{label_preoccupation}
a boundaried graph ${\bf G}=(G,B,ρ)$ and an apex wall triple $(A,W,\frak{R})$ that is not affected by $B$
and we show, in \autoref{label_pretendientes}, how  every topological minor model $(M,T)$ of a graph $H$ in $G$ can be rerouted
away from the compass of the central
subwall $W'$ of $W.$ This will permit us later to declare the whole compass of $W'$ {\sl irrelevant}
and rule out  the possibility that $W$ has size exceeding some function depending on the ``intrusion'' of $H$ in ${\bf G}.$
The proof of \autoref{label_pretendientes} is the most technical part of this paper.
For this, we define an appropriate ``flat'' representation of the $\frak{R}$-compass of $W,$
called its \emph{leveling}, and a representation of the wall $W$ in the leveling that is ``well-aligned''. This well-alignment property, defined in \autoref{label_disintegrated}, emerges from the regularity of the flatness pair $(W,\frak{R})$ and
permits the representation of $(M,T)$ by a topological minor model $(\tilde{M},\tilde{T})$  of the leveling, accompanied with a suitable encoding of the parts of  $(M,T)$ that have been suppressed by the leveling.
This will permit us to obtain, using  \autoref{label_constitutivos}, a rerouting $(\hat{M},\hat{T})$ of $(\tilde{M},\tilde{T})$ inside the leveling. Finally, using the   homogeneity property,  we will translate back $(\hat{M},\hat{T})$ to a rerouting of $(M,T)$ that will avoid the compass of the central subwall $W'.$

\subsection{A lemma for model taming}
\label{label_generalmente}

We introduce the concept of a railed annulus and   present the main combinatorial result of~\cite{GolovachST20-SODA}.

%\removed{
\paragraph{Railed annuli.}
Let $G$ be a  partially $\Delta$-embedded graph and let ${\cal C}=[C_{1},\ldots,C_{r}],$ $r\geq 2,$
be a collection of  vertex-disjoint cycles of the compass of $G.$ We say that the sequence ${\cal C}$ is a {\em $\Delta$-nested sequence of  cycles} of $G$
if every $C_{i}$ is the boundary of an open  disk $D_{i}$ of $\Delta$ such that  $\Delta\supseteq D_{1}\supseteq\cdots \supseteq D_{r}.$ From now on,
each $\Delta$-nested sequence ${\cal C}=[C_{1},\ldots,C_{r}]$ will be accompanied
with the sequence $[D_{1},\ldots,D_{r}]$ of the corresponding open disks
as well as
the sequence $[\overline{D}_{1},\ldots,\overline{D}_{r}]$ of their closures.
Given $x,y\in[r]$ with $x\leq y,$
we call the set  $\overline{D}_{x}\setminus D_{y}$ {\em $(x,y)$-annulus} of ${\cal C}$
and we denote it by $\ann({\cal C},x,y).$ Finally we say that $\ann({\cal C},1,r)$
is the {\em annulus} of ${\cal C}$ and we denote it by $\ann({\cal C}).$
%\marg{$\ann({\cal C},x,y),$ $\ann({\cal C})$}

Let $r\in\Bbb{N}_{\geq 3}$ and $q\in \Bbb{N}_{\geq 3}$ with $r$ odd.
An {\em $(r,q)$-railed annulus} of a $\Delta$-partially-embedded graph\ $G$ is a pair ${\cal A}=({\cal C},{\cal P})$ where
${\cal C}=[C_{1},\ldots,C_{r}]$  is a $\Delta$-nested collection of cycles of $G$ and ${\cal P}=[P_{1},\ldots,P_{q}]$ is a  collection of pairwise vertex-disjoint
paths in $G,$ called {\em rails},
%SED> rails are not used...
% \ig{say that these paths are called \emph{rails}, as in~\cite{KaminskiT12}?}
such that
\begin{itemize}
	\item for every $j\in[q],$\ $P_{j}\subseteq \ann({\cal C}),$ and

	\item for every $(i,j)\in[r]\times[q],$   $C_{i}\cap P_{j}$ is  a non-empty path that we denote by $P_{i,j}.$%\marg{ $P_{i,j}.$}
\end{itemize}
%
%\ig{Dimitrios, please put some figure of a railed annulus}
%

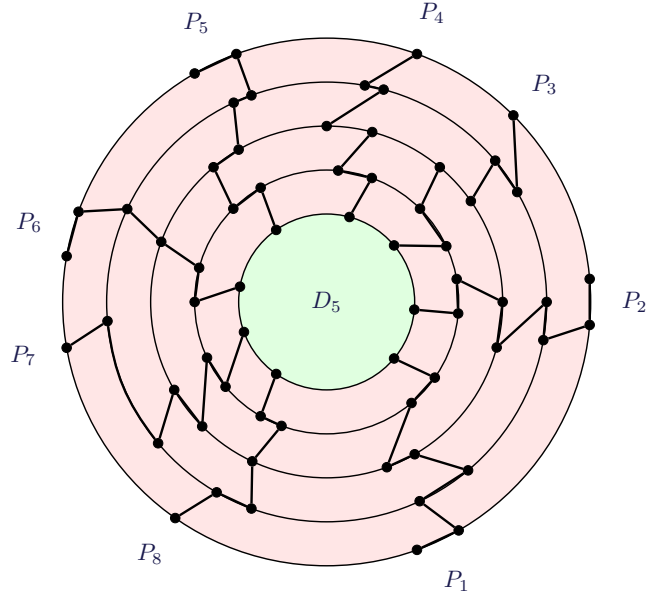
\begin{figure}
	\begin{center}
		\scalebox{.9}{
			\begin{tikzpicture}[scale=.65]
				\foreach \x in {2,...,6}{
						\draw[line width =0.6pt] (0,0) circle (\x cm);
					}
				\begin{scope}[on background layer]
					\fill[red!10!white] (0,0) circle (6 cm);
					\fill[green!12!white] (0,0) circle (2 cm);
				\end{scope}

				\node (P3) at (45:7) {$P_{3}$};
				\node[black node] (P11) at (45:6) {};
				\node[black node] (P21a) at (30:5) {};
				\node[black node] (P21b) at (40:5) {};
				\node[black node] (P31a) at (35:4) {};
				\node[black node] (P31b) at (50:4) {};
				\node[black node] (P41a) at (45:3) {};
				\node[black node] (P41b) at (25:3) {};
				\node[black node] (P51) at (40:2) {};
				\draw[line width=1pt] (P11) -- (P21a) -- (P21b) -- (P31a)  (P31b) -- (P41a) -- (P41b) -- (P51);

				\node (P4) at (70:7) {$P_{4}$};
				\node[black node] (P12) at (70:6) {};
				\node[black node] (P22a) at (80:5) {};
				\node[black node] (P22b) at (75:5) {};
				\node[black node] (P32a) at (90:4) {};
				\node[black node] (P32b) at (75:4) {};
				\node[black node] (P42a) at (85:3) {};
				\node[black node] (P42b) at (70:3) {};
				\node[black node] (P52) at (75:2) {};
				\draw[line width=1pt] (P12) -- (P22a) -- (P22b) -- (P32a) (P32b) -- (P42a) -- (P42b) -- (P52);

				\node (P5) at (115:7) {$P_{5}$};
				\node[black node] (P13a) at (120:6) {};
				\node[black node] (P13b) at (110:6) {};
				\node[black node] (P23a) at (110:5) {};
				\node[black node] (P23b) at (115:5) {};
				\node[black node] (P33a) at (120:4) {};
				\node[black node] (P33b) at (130:4) {};
				\node[black node] (P43a) at (135:3) {};
				\node[black node] (P43b) at (120:3) {};
				\node[black node] (P53) at (125:2) {};
				\draw[line width=1pt] (P13a) -- (P13b) -- (P23a) -- (P23b) -- (P33a) -- (P33b) -- (P43a) -- (P43b) -- (P53);

				\node (P6) at (165:7) {$P_{6}$};
				\node[black node] (P14a) at (170:6) {};
				\node[black node] (P14b) at (160:6) {};
				\node[black node] (P24) at (155:5) {};
				\node[black node] (P34) at (160:4) {};
				\node[black node] (P44a) at (165:3) {};
				\node[black node] (P44b) at (180:3) {};
				\node[black node] (P54) at (170:2) {};
				\draw[line width=1pt] (P14a) -- (P14b) -- (P24) -- (P34) -- (P44a) -- (P44b) -- (P54);

				\node (P7) at (190:7) {$P_{7}$};
				\node[black node] (P18a) at (190:6) {};
				\node[black node] (P28a) at (185:5) {};
				\node[black node] (P28b) at (220:5) {};
				\node[black node] (P38a) at (210:4) {};
				\node[black node] (P38b) at (225:4) {};
				\node[black node] (P48a) at (205:3) {};
				\node[black node] (P48b) at (220:3) {};
				\node[black node] (P58) at (200:2) {};
				\draw[line width=1pt] (P18a) -- (P28a) to [bend right=15]  (P28b);
				\draw[line width=1pt] (P28b) --  (P38a) -- (P38b) -- (P48a) -- (P48b) -- (P58);

				\node (P8) at (235:7) {$P_{8}$};
				\node[black node] (P15) at (235:6) {};
				\node[black node] (P25a) at (240:5) {};
				\node[black node] (P25b) at (250:5) {};
				\node[black node] (P35b) at (245:4) {};
				\node[black node] (P45a) at (250:3) {};
				\node[black node] (P45b) at (240:3) {};
				\node[black node] (P55) at (235:2) {};
				\draw[line width=1pt] (P15) -- (P25a)  -- (P25b) --  (P35b) -- (P45a) -- (P45b) -- (P55);

				\node (P1) at (295:7) {$P_{1}$};
				\node[black node] (P16a) at (290:6) {};
				\node[black node] (P16b) at (300:6) {};
				\node[black node] (P26a) at (295:5) {};
				\node[black node] (P26b) at (310:5) {};
				\node[black node] (P36a) at (300:4) {};
				\node[black node] (P36b) at (290:4) {};
				\node[black node] (P46a) at (310:3) {};
				\node[black node] (P46b) at (325:3) {};
				\node[black node] (P56) at (320:2) {};
				\draw[line width=1pt] (P16a) -- (P16b) -- (P26a) -- (P26b) -- (P36a) -- (P36b) -- (P46a) -- (P46b) -- (P56);

				\node (P2) at (0:7) {$P_{2}$};
				\node[black node] (P17a) at (5:6) {};
				\node[black node] (P17b) at (-5:6) {};
				\node[black node] (P27a) at (0:5) {};
				\node[black node] (P27b) at (-10:5) {};
				\node[black node] (P37a) at (-15:4) {};
				\node[black node] (P37b) at (0:4) {};
				\node[black node] (P47a) at (10:3) {};
				\node[black node] (P47b) at (-5:3) {};
				\node[black node] (P57) at (-5:2) {};
				\draw[line width=1pt] (P17a) -- (P17b) -- (P27b)  -- (P27a) -- (P37a) -- (P37b) -- (P47a) -- (P47b) -- (P57);

				\node (C5) at (0:0) {$D_{5}$};

			\end{tikzpicture}
		}

	\end{center}

	\caption{An example of a $(5,8)$-railed annulus and its inner disk $D_{5}.$}
	\label{label_expectatives}
\end{figure}

See \autoref{label_expectatives} for an example of a $(5,8)$-railed annulus. The following proposition states that large railed annuli can be found inside a modestly larger wall  and will be used in the next section. A similar (but less precise) statement can be found in~\cite{KaminskiT12}.

\begin{figure}[h]
	\begin{center}%\vspace{-.1cm}
		\includegraphics[width=.57\textwidth]{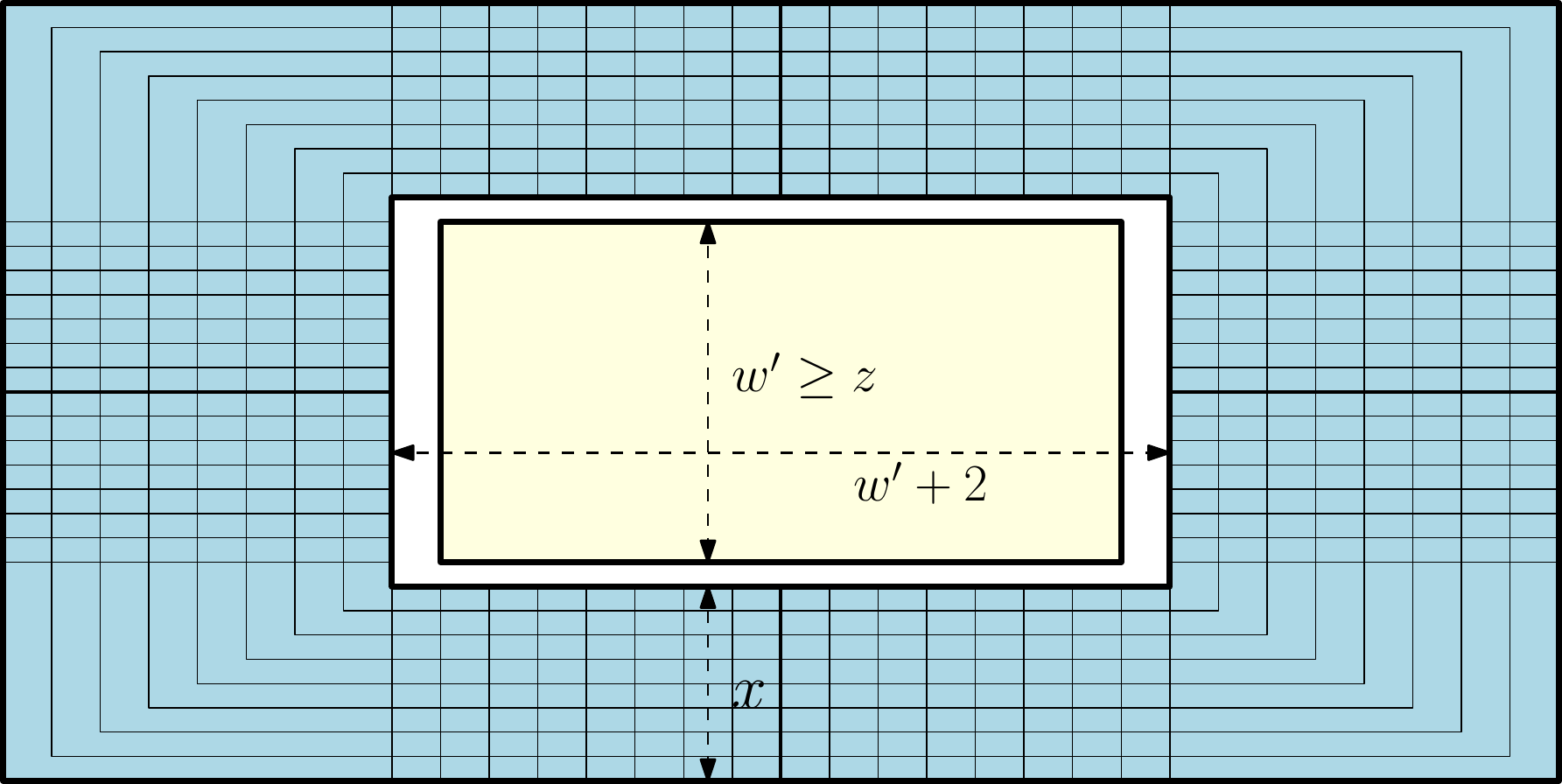}
	\end{center}\vspace{-.25cm}
	\caption{A visualization of the proof of \autoref{label_simultaneously} in the $33$-wall
		where $x=9,$ $y=64,$ and $z= 15.$ For simplicity, the paths in $\{P_{i,j}\mid  (i,j)\in[33]^2\}$
		are depicted as vertices.}
	\label{label_importunarla}
	\vspace{.35cm}
\end{figure}

\begin{proposition}
	\label{label_simultaneously}
	If $x,z \geq 3$ are odd integers, $y\geq 1,$ and $W$  is an   ${\sf odd}(2x+\max\{z,\frac{y}{4}-1\})$-wall,  then
	\begin{itemize}
		\item there is a collection ${\cal P}$ of $y$ paths in $W$ such that if ${\cal C}$ is
		      the collection of the first $x$ layers of $W,$ then $({\cal C},{\cal P})$
		      is an $(x,y)$-railed annulus of $W$ where the first cycle of ${\cal C}$ is the perimeter of $W,$ and
		\item the open disk defined by the $x$-th cycle of ${\cal C}$ contains the vertices of the compass of the central $z$-subwall of $W.$
	\end{itemize}
	%\sed{Alert: Discs $D_{i}$ are OPEN!}
\end{proposition}

\begin{proof}

	Let ${\cal C}=\{C_1, \ldots, C_x\}$ be the collection of the first $x$ layers of $W.$ Notice that  the open disk defined by $C_x$ contains the vertices of the compass of the central $w'$-subwall, where $w':=w-2x={\sf odd}(\max\{z,\frac{y}{4}-1\})\geq z.$ On the other hand, $W$ contains a collection ${\cal P}$ of at least $2w' + 2(w'+2)$ pairwise vertex-disjoint paths from $C_1$ to $C_x.$ Since $2w' + 2(w'+2)=4w'+4 \geq y,$
	the pair $({\cal C},{\cal P})$	is an $(x,y)$-railed annulus of $W$ where the first cycle of ${\cal C}$ is the perimeter of $W$ (see \autoref{label_importunarla}).
\end{proof}

We define  the {\em annulus} of ${\cal A}=({\cal C}, {\cal P})$ as the annulus of ${\cal C}.$ We call $C_{1}$ and $C_{r}$ the \emph{outer} and the \emph{inner} cycle of ${\cal A},$ respectively.
Also, if $(i,i')\in[r]^2$ with $i<i'$ then we define ${\cal A}_{i,i'}=([C_{i},\cdots,C_{i'}],{\cal P}\cap \ann({\cal C},i,i')).$

The {\em union-graph} of an $(r,q)$-railed annulus ${\cal A}=({\cal C}, {\cal P})$ is defined as $G({\cal A}):=(\bigcup_{i\in[r]}C_{i})\cup(\bigcup_{i\in[q]}P_{i}).$ Clearly, $G({\cal A})$ is a planar graph and we always assume that its infinite face is the one whose boundary is the first cycle of ${\cal C}.$

Let ${\cal A}$ be a $(r,q)$-railed annulus of a partially $\Delta$-embedded graph $G.$
Let  $r=2t+1,$ for some $t\geq 0.$ Let also $s\in[r]$ where $s=2t'+1,$ for some $0\leq t'\leq t.$
Given some $I\subseteq [q],$ we say that a subgraph  $M$  of $G$ is {\em $(s,I)$-confined in ${\cal A}$} if
$$M\cap \ann({\cal C},t-t',t+t')\subseteq \bigcup_{i\in I}P_{i}.$$

The following proposition has been recently proved by Golovach et al.~\cite[Theorem 2.1]{GolovachST20-SODA}, where it has been dubbed as the ``Model Taming Lemma''.

\begin{proposition}%[{\sf MODEL TAMING}]
	\label{label_interessiert}
	There exist two functions   $\newfun{label_inevitability}, \newfun{label_enlightening}:\Bbb{N}_{\geq  0}\to\Bbb{N}_{\geq  0}$ {such that the images of $\funref{label_enlightening}$ are even} and such that
	if
	\begin{itemize}
		\item $s$ is a positive odd integer,
		\item $H$ is a graph on at most $\ell$ edges,
		      %\ig{what about isolated vertices?}\sed{Es correcto porque los isolated sone terminales, pues están fuera del annulus},
		\item $G$ is a $\Delta$-partially-embedded graph,
		\item  ${\cal A}=({\cal C},{\cal P})$ is an  $(r,q)$-railed
		      annulus of $G,$ where  $r =  \funref{label_enlightening}(\ell)+2+s$ and $q\geq  5/2 \cdot \funref{label_inevitability}(\ell),$
		\item $({M},{T})$ is a  topological minor model of $H$ in $G$ such that ${T}\cap {\sf ann}({\cal A})=\emptyset,$ and
		\item $I\subseteq [q]$ where $|I|> \funref{label_inevitability}(\ell),$
	\end{itemize} then
	$G$ contains a topological minor model $(\tilde{M},\tilde{T})$ of $H$ in $G$ such that
	%\sed{Fix (actually simplify when this is possible) topological minor model terminology}
	\begin{enumerate}
		\item $\tilde{T}=T,$
		      %	\item ${T'}\cap {\sf ann}({\cal A})=\emptyset$
		\item   $\tilde{M}$ is $(s,I)$-confined in ${\cal A},$ and
		\item $\tilde{M}\setminus \ann({\cal A})\subseteq {M}\setminus \ann({\cal A}).$
	\end{enumerate}
	Moreover $\funref{label_enlightening}(\ell)=\Ocal((\funref{label_inevitability}(\ell))^2).$
\end{proposition}

\begin{remark}\label{label_constitueraient}
	It is worth mentioning here that the function $\funref{label_inevitability}(\ell)$ of 	\autoref{label_interessiert}
	depends on the constants involved in the Unique Linkage Theorem~\cite{KawarabayashiW2010asho,RobertsonSGM22}; see
	\autoref{label_reconstituirse} for a more detailed discussion. At this point we just remark that, according to the results of Adler and Krause~\cite{AdlerK19}, we have that $\funref{label_inevitability}(\ell) = 2^{\Omega(\ell)}.$ This permits us to henceforth make the (generous) assumption that $\ell = \Ocal(\funref{label_inevitability}(\ell)).$
\end{remark}

%We say that a linkage $L$ of a graph $G$ is {\em vital} if $V(L)=V(G)$ and there is no other linkage of $G$ that is equivalent to $L.$

%}%REMOVE

\subsection{Model rerouting in partially disk-embedded graphs}
\label{label_biologically}

%\ig{write small intro for this subsection, or not?}

%
%Let $G$ be a plane graph and $x\in V(G).$  We define ${\sf clf}(x)$ as the closure of the finite faces of $G$ that are incident to $x.$ Given a $A\subseteq V(G),$ we say that $A$ is {\em well-scattered} in $G$
%if for every two distinct vertices $x,y\in A,$ ${\sf clf}(x)\cap {\sf clf}(y)=\emptyset.$
%
%
%

%\remove{
Using classic  results on how to optimally draw planar graphs
of maximum degree three into grids (see e.g.,~\cite{Kant93PhDThesis}) one may easily derive the following.

\begin{proposition}
	\label{label_unrecognized}
	There exists a function  $\newfun{label_slaughtering}: \Bbb{N}\to\Bbb{N}$
	such that for every $\ell$-vertex  planar graph $H$ with maximum degree three
	there is a {\sf tm}-pair $(M,T)$  of the $(\funref{label_slaughtering}(\ell)\times \funref{label_slaughtering}(\ell))$-grid, denoted by $\Gamma,$ that is a topological minor model of $H$ in $\Gamma.$ Moreover, it holds that  $\funref{label_slaughtering}(\ell)=\Ocal(\ell).$
\end{proposition}

Let $\Gamma$ be an $(r\times r)$-grid for some $r\geq 3.$ We  see a  $\Gamma$-grid as the union of $r$ horizontal paths and $r$ vertical paths.
% Clearly we can embedd $G$ in the plane $\Bbb{R}^{2}$ such that the  the infinite face is the unique
%face, called {\em external face} of $\Gamma,$ denoted by ${\sf ext}(\Gamma)$ that contains more than 4 vertices in its boundary. We also call this boundary {\em perimeter} of $\Gamma.$
Given an $i\in\lfloor\frac{r}{2}\rfloor,$ we define the {\em $i$-th layer} of $\Gamma$
recursively as follows: the first layer of $\Gamma$ is its perimeter, while, if $i\geq 2,$ the
$i$-th layer of $\Gamma$ is the perimeter of the $(r-2(i-1)\times r-2(i-1))$-grid created if we remove from $\Gamma$ its $i-1$ first layers. When we deal with a  $(r\times r)$-grid $\Gamma,$ we always consider its embedding where the infinite face is bounded by the first layer of $\Gamma.$

\paragraph{Safely arranged  models.}

Let $G$ be a plane graph.
Given two subgraphs of $G,$ we define their {\em face-distance} as the minimum face-distance
between two of their vertices.
We denote by ${\bf F}^{(i)}_{G}(x)$ the set of all vertices of $G$ that are within face-distance at most $i$ from vertex $x.$

Given a $c\geq 0$ and   a {\sf tm}-pair $(M,T)$ of $G,$  we say that $(M,T)$ is {\em safely $c$-dispersed in $G$} if

\begin{itemize}
	\item every two distinct vertices $t,t'\in T$ are within face-distance at least $2c+1$ in $G,$ and
	\item for every $t\in T$ of degree $d$ in $M,$  the graph $M[{\bf F}_{G}^{(c)}(t)\cap V(M)]$
	      consists of $d$ paths with $t$ as a unique common endpoint.
\end{itemize}

\medskip

With \autoref{label_unrecognized} at hand, we can prove the following useful lemma.

\begin{lemma}%[{\sf GRAPH DRAWING}]
	\label{label_dispenseront}
	There exists a  function $\newfun{label_endeavouring}: \Bbb{N}^3\to\Bbb{N}$ such that
	the following holds. Let $c,r,r',\ell\in \Bbb{N},$  $r'\leq r,$
	$H$ be a $D$-embedded $(\ell+r')$-vertex
	graph, and $Z:=\{z_1,\ldots,z_{r'}\}\subseteq V(H)$ such that
	\begin{itemize}
		\item the vertices of $H$ have degree at most three,
		\item $Z$ is an independent set of $H,$
		\item all vertices of $Z$ have degree one in $H,$
		      %\item all vertices of $Z$ are points of the boundary of $D$\sed{Exactly the points of the boundary},
		\item $\bor(D)\cap H=Z,$ and
		\item  $(z_1,\ldots,z_{r'})$
		      is the cyclic ordering of the vertices of $Z$ as they appear in the boundary of $D.$
		      %\sed{difference ${D}$ and $\Delta$}
	\end{itemize}
	Let also $G$ be a $\Delta$-embedded  graph,  ${\cal A}=({\cal C},{\cal P})$ be an  $(x,y)$-railed annulus of $G,$ where $x,y$ are integers such that $\min\{x,y\} \geq  \funref{label_endeavouring}(c,r,\ell) \geq r,$ and where  ${\cal C}=[C_{1},\ldots,C_{x}]$ and ${\cal P}=[P_{1},\ldots, P_y],$ $w_{i}$ be the    endpoint of $P_{i}$ that is contained  in $C_{1},$ for $i\in[r],$ and $I:=\{i_{1},\ldots,i_{r'}\}\subseteq [r].$
	Then the union-graph $G({\cal A})$ of ${\cal A}$  contains a {\sf tm}-pair $(M,T)$ that is a topological minor model of $H$ in $G({\cal A})$ such that
	\begin{itemize}

		\item
		      for each $j\in[r'],$  $\sigma_{M,T}(z_j)=w_{i_{j}},$
		      %\item $C_{\funref{label_endeavouring}(r,\ell)}\cap M$ is null \ig{null? do you mean empty?}, and
		\item the {\sf tm}-pair $(M,T)$ is safely $c$-dispersed in the union graph $G({\cal A}),$ and
		\item none of the vertices of $T\setminus\{w_{i_1},\ldots,w_{i_{r'}}\}$ is within face-distance less than $c$
		      from some vertex in $C_{1}$ or in $C_{r}.$
	\end{itemize}
	Moreover, it holds that  $\funref{label_endeavouring}(c,r,\ell)=\Ocal(cr(\ell+r)).$
\end{lemma}

\begin{proof}
	Using $H,$ we construct
	a new graph $H'$ as follows: consider a copy $\tilde{H}$ of $H$
	where the copy of $z_{i}$ in $\tilde{H}$ is denoted by $\tilde{z}_{i},$
	for each  $i\in [r'].$ We take the disjoint union of $H$ and $\tilde{H},$
	add the edges
	$\{z_{1},z_2\},\ldots,\{z_{r'},z_{1}\},$ forming a cycle $C,$ subdivide the edges in $C,$
	add  the edges  $\{\tilde{z}_{1},\tilde{z}_{2}\},\ldots,\{\tilde{z}_{r},\tilde{z}_{1}\},$ forming a cycle $\tilde{C},$
	subdivide the edges in $\tilde{C},$ and, given that, for $i\in[r'],$  $x_{i}$ (resp. $\tilde{x}_{i}$)
	is the vertex created after the subdivision of $\{z_{i},z_{i+1}\}$ (resp. $\{\tilde{z}_{i},\tilde{z}_{i+1}\}$) (here $r'+1$ is interpreted as $1$), add the edges
	$\{x_{1},\tilde{x}_1\},\ldots,\{x_{r},\tilde{x}_{r'}\}.$
	The resulting graph $H'$ has $2(\ell+r')\leq 2(\ell+r)$ vertices, is planar, and has maximum degree three. Let $s=\funref{label_slaughtering}(2(\ell+r)).$
	By \autoref{label_unrecognized}, there is a  ${\sf tm}$-pair $(M',T')$ of the $(s\times s)$-grid $\Gamma$ that is a topological minor
	model of $H'$ in $\Gamma.$
	We now subdivide $\overline{r}=2(c+1)r$ times  each of the edges of $\Gamma$  and see the resulting graph $\Gamma'$
	as a subgraph of a $((\overline{r}+1)(s+2)\times (\overline{r}+1)(s+2))$-grid $\Gamma''$
	in a way that none of the $\overline{r}$ first layers of $\Gamma''$ intersects $\Gamma'.$ By  also subdividing $\overline{r}$ times each of the edges of $M'$
	we construct a  ${\sf tm}$-pair $(M'',T')$ of $\Gamma''$
	that is a  a topological minor  model of $H'$ in $\Gamma''.$

	Let $w_{1},\ldots,w_{r}$ be the first $r$ vertices of the lower path of $\Gamma''.$
	Recall that $H$ is a subgraph of $H',$ therefore we can define  $M=\sigma_{M',T'}(H).$
	Let also $T=\sigma_{M',T'}(V(H)).$
	Notice that $(M,T)$
	is a {\sf tm}-pair of $\Gamma''$  that is a topological minor model
	of $H.$ Let $\hat{z}_{i}=\sigma_{M,T}(z_{i}), i\in[r'].$
	We make two observations about the position of these vertices in $\Gamma''.$
	The first is that, because of the construction of $H',$ $\hat{z}_{1},\ldots,\hat{z}_{r'}$ appear, in this ordering, on a cycle of $\Gamma'$ bounding a closed disk, say $\Delta,$ that  contains the whole $M.$
	The second is that, as $M$ is a subgraph of $\Gamma',$
	each pair $\hat{z}_{i},\hat{z}_{j}, i\neq j,$
	is at distance at least $\overline{r}+1$ in the graph $\Gamma''':=\Gamma''\setminus \inter(\Delta).$
	It is now easy to observe that the two previous
	observations permit to
	find in $\Gamma'''$ pairwise disjoint paths joining $\hat{z}_{j}$ with $w_{i_{j}},$ for $j\in[r'].$
	%SED> we need a figure for this
	By adding these paths  in $M$ and including in  $T$  the set $\{w_{i_1},\ldots,w_{i_{r'}}\},$
	we construct a {\sf tm}-pair $(M,T)$
	that is a topological minor model of $H$ in $\Gamma'$ such that
	for each $j\in[r'],$ the function $\sigma_{M,T}$
	maps the vertex $z_{{j}}$ to $w_{i_{j}}'$ and the intersection of $V(M)$ and the upper path of $\Gamma''$ is empty.  Moreover, as we applied at least $\overline{r}=(c+2)r$ subdivisions, it also holds that
	the set $T$ is safely  $2c$-dispersed in $\Gamma''.$
	Moreover, it is easy to observe that none of the vertices of $T\setminus\{w_{i_1},\ldots,w_{i_{r'}}\}$ is within face-distance less than $c$
	from some vertex in the perimeter of $\Gamma''.$

	Consider now a  $(x,y)$-railed annulus $({\cal C},{\cal P})$ of some $\Delta$-embedded graph $G,$ with $\min\{x,y\} \geq q,$ where $q=(\overline{r}+1)(s+2).$ Let ${\cal C}=[C_{1},\ldots,C_{x}]$ and ${\cal P}=[P_{1},\ldots, P_y]$ as in the statement of the lemma. 	Let also $\tilde{\Gamma}=G({\cal A}).$
	For every $i\in[q],$ we define  $F^{(i)}_{\cal A}$ as the edge set of the unique
	path in $C_{i}$  with one endpoint in $P_{i,q}$ and the other in $P_{i,1},$ that does not contain
	internal vertices of the paths $P_{i,q}$ or $P_{i,1},$ and does not contain any vertex from $P_{2}.$ We denote by $F^{\sf e}_{\cal A}$ (resp. $F^{\sf v}_{\cal A}$) the set of
	all edges (resp. internal vertices) of the paths $F^{(i)}_{\cal A}, i\in[q].$
	Notice  that  the grid $\Gamma''$ occurs from $(\tilde{\Gamma}\setminus F_{\cal A}^{\sf e})\setminus F_{\cal A}^{\sf v}$ if, for every $(i,j)\in [q]^2,$
	we contract the path $P_{i,j}$ defined by the intersection of the $i$-th horizontal path and the $j$-th vertical path of $W.$ It is easy to see that
	if in $\tilde{\Gamma}$ we uncontract each vertex, say  $(i,j)$ of $M$ to the path $P_{i,j},$
	one can transform $(M,T)$ to a ${\sf tm}$-pair of $W$ that is a topological minor model of $H$
	in $W$ and additionally,
	for each $i\in[r'],$ the function $\sigma_{M,T}$
	maps the vertex $z_{{j}}$ to $w_{i_{j}}.$ This implies the first condition of the lemma.
	%The  second condition follows  by the fact that the intersection of $V(M)$ and the upper path of $\Gamma''$ is empty.
	The second condition follows directly from the fact that    the pair $(M,T)$ was already safely %\ig{just ``safely'', right?}
	$2c$-dispersed  before applying the uncontractions and  such uncontractions  cannot reduce the distance to more than half of it. The third condition is also an obvious consequence of the uncontraction procedure.
	Therefore, the lemma holds
	if  we set  $\funref{label_endeavouring}(c,r,\ell):=q=\Ocal(cr(\ell+r)).$
\end{proof}

Let $G$ be a partially $\Delta$-embedded graph and let ${\cal C}=[C_{1},\ldots,C_{r}$]  be a {\em $\Delta$-nested sequence of  cycles} of $G$  and let $[D_{1},\ldots,D_{r}]$
(resp. $[\overline{D}_{1},\ldots,\overline{D}_{r}]$)
be the sequences of the corresponding open (resp. closed) disks.\medskip

Let also $(M,T)$ be a {\sf tm}-pair of $G$ and $p\in [r].$
We define the {\em $p$-crop} of $(M,T)$ in ${\cal C},$ denoted by $(M,T)\Cap \overline{D}_{p},$
as the {\sf tm}-pair $(M',T')$ where $M'=M\cap \overline{D}_{p}$ and $T'=(T\cap \overline{D}_{p})\cup(V(C_{p}\cap M)).$

Given a graph $H$ a set $Q\subseteq V(H)$ and a graph $G,$ we say that $\phi: V(H)\to 2^{V(G)}$
is a {\em $Q$-respecting contraction-mapping of $H$ to $G$} if
\begin{itemize}
	\item $\bigcup_{x\in V(H)}\phi(x)=V(G),$
	\item $\forall x,y\in V(H),$ if $x\neq y$ then $\phi(x)\cap \phi(y)=\emptyset,$
	\item $\forall x\in V(H),$ $G[\phi(x)]$ is connected,
	\item $\forall \{x,y\}\in E(H),\ $ $G[\phi(x)\cup \phi(y)]$ is connected, and
	\item $\forall x\in Q, |\phi(x)|=1.$
\end{itemize}

The critical point in the above definitions is that
vertices in $Q$ are not ``uncontracted’’
when transforming $H$ to $G.$
%Given a non-negative integer $x,$ we denote by ${\sf odd}(x)$ the minimum odd number that is not smaller than $x.$

\paragraph{Intrusion of a topological minor model.}
Let $G$ be a graph, let $S\subseteq V(G),$ and let  $(M,T)$ be a {\sf tm}-pair
of $G.$  We define the {\em $S$-intrusion} of $(M,T)$ in $G$
as the maximum value between $|S\cap T|$
and  the number of subdivision paths of $(M,T)$ that contain vertices
of $S.$ It is important to notice that $S$ can intersect many times a subdivision path of $(M,T),$ however
the value of the $S$-intrusion counts each such a path only once.
\medskip

Using \autoref{label_simultaneously}, \autoref{label_interessiert}, and \autoref{label_dispenseront}  we prove the following.

\begin{theorem}%[{\sf COLLAPSE}]
	\label{label_constitutivos}
	There exist three functions $\newfun{label_tourterelles}:\Bbb{N}^2\to\Bbb{N},$ $\newfun{label_quitablement}:\Bbb{N}^2\to\Bbb{N},$ and $\newfun{label_presuntuosos}:\Bbb{N}^3\to\Bbb{N}$
	such that the following holds.
	Let  $c,\ell\in\Bbb{N},$ $z \geq 3$ be an odd integer, and $G$ be a partially $\Delta$-embedded graph,
	whose compass contains an $\funref{label_presuntuosos}(c,z,\ell)$-wall $W$
	with $\bor(\Delta)$ as perimeter. Let also ${\cal C} = [C_{1},\ldots,C_{\funref{label_tourterelles}(c,\ell)}]$ be the first $\funref{label_tourterelles}(c,\ell)$-layers of $W$ and $D_{1},\ldots,D_{\funref{label_tourterelles}(c,\ell)}$ be the open disks of $\Delta$ that they define. If  $(M,T)$
	is a {\sf tm}-pair of $G$ whose $\Delta\cap V(G)$-intrusion in $G$  is at most $\ell$ and $Q$ is a subset of $T$
	containing vertices of degree at most three in $M,$ then there is  a  {\sf tm}-pair $(\hat{M},\hat{T})$  of $G$  and an integer $b\in [\funref{label_tourterelles}(c,\ell)]$ such that

	\begin{enumerate}
		\item    $\hat{M}\setminus D_{b}$ is a subgraph of ${M}\setminus D_{b},$
		\item ${\sf ann}({\cal C}, b,b+\funref{label_quitablement}(c,\ell)-1)\cap (T \cup \hat{T})=\emptyset,$
		\item   $(\hat{M},\hat{T})\Cap \overline{D}_{b+\funref{label_quitablement}(c,\ell)}$ is a {\sf tm}-pair of $W$ that is  safely $c$-dispersed in $W$
		      and none of the vertices of $\hat{T}\cap \overline{D}_{b+\funref{label_quitablement}(c,\ell)}$ is within face-distance less than $c$  in $W$
		      from some vertex of $C_{b+\funref{label_quitablement}(c,\ell)}\cup C_{\funref{label_tourterelles}(c,\ell)},$
		\item $\hat{M}\cap D_{\funref{label_tourterelles}(c,\ell)}=\emptyset,$
		      %\item $\hat{M}$ does not intersect the vertices of the compass of the central $z$-subwall of $W.$
		\item the compass of the central $z$-subwall of $W$ is a subset of $D_{\funref{label_tourterelles}(c,\ell)},$ and
		\item there is a $Q$-respecting contraction-mapping of ${\sf diss}(M,T)$ to ${\sf diss}(\hat{M},\hat{T}).$
	\end{enumerate}
	Moreover, it holds that
	$\funref{label_tourterelles}(c,\ell) = \Ocal( c \cdot (\funref{label_inevitability}(\ell))^3),$
	$\funref{label_quitablement}(c,\ell) =\Ocal( c \cdot (\funref{label_inevitability}(\ell))^2),$ and
	$\funref{label_presuntuosos}(c,z,\ell) =  \Ocal( c \cdot (\funref{label_inevitability}(\ell))^3   + z).$
\end{theorem}
%$\funref{label_inevitability}(\ell)$

%$\funref{label_enlightening}(\ell)$

%$\funref{label_endeavouring}(c,\funref{label_inevitability}(\ell)+1,3\ell+\funref{label_inevitability}(\ell)+1)$

See \autoref{label_palsanquienne} for an illustration of the conditions guaranteed by \autoref{label_constitutivos}.

\begin{figure}[h!]
	\begin{center}%\vspace{-.1cm}
		\includegraphics[width=.63\textwidth]{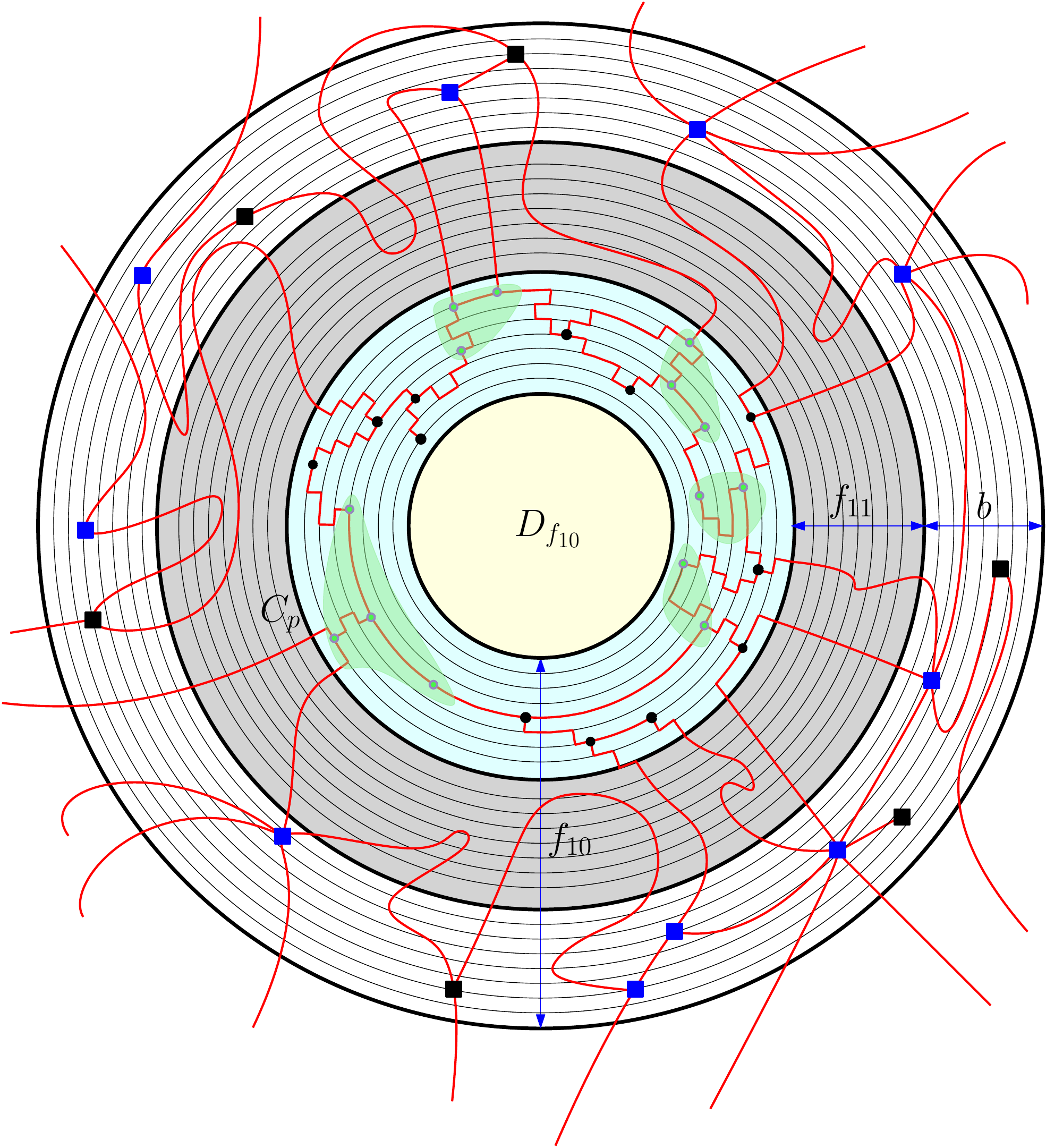}
	\end{center}\vspace{-.25cm}
	\caption{A visualization of how a  {\sf tm}-pair $(M,T)$  is rearranged to a new   {\sf tm}-pair $(\hat{M},\hat{T})$ as in \autoref{label_constitutivos}. The figure depicts in red the part of the {\sf tm}-pair $(\hat{M},\hat{T})$ that intersects the disk $\Delta.$ The cycles
		correspond to the first $\funref{label_tourterelles}(c,\ell)$
		layers of $W.$ The black vertices are the vertices in $Q,$ while the circled vertices inside the turquoise  area are the ``new'' branch vertices of $\tilde{T}$ that are vertices of $W.$ The ``green clouds'' are the non-singleton images of the $Q$-respecting contraction-mapping of ${\sf diss}(M,T)$ to ${\sf diss}(\hat{M},\hat{T}).$ We stress that in this picture, the way the model enters the turquoise area does not reflect the fact $\tilde{M}$ that is $(s,I)$-confined in ${\cal A}',$ as it is argued in the proof. We opted not to reflect this fact in the figure as we prioritized the visualization of other, more important, aspects of the proof.}
	\label{label_palsanquienne}
	\vspace{.35cm}
\end{figure}

\begin{proof}% \ig{in the whole proof, replace (carefully) ``$l$'' by ``$\ell$''}
	%\sed{But these two are different variables!}
	Let  $r=\funref{label_inevitability}({\ell})+1,$  $s={\sf odd}(\funref{label_endeavouring}(c,r,{3\ell}+r)),$
	%\blue{OLD: $x={\sf odd}(\max\{\funref{label_enlightening}(\ell)+2+s, \lceil 5/2\cdot  \funref{label_inevitability}(h)\rceil\}),$}
	$x'={\sf odd}(\funref{label_enlightening}(\ell)+2+s),$ $y=\max\{s, \lceil 5/2\cdot  \funref{label_inevitability}(\ell)\rceil\},$ $x={\sf odd}((\ell+1)\cdot x'),$ and  $
		w={\sf odd}(2x+\max\{z,\frac{y}{4}-1\}).$ We will prove the theorem for $\funref{label_tourterelles}(c,\ell)=x,$
	$\funref{label_quitablement}(\ell)=\frac{x'-s}{2}, $ and $\funref{label_presuntuosos}(c,z,\ell)=w.$
	Let  $G$ be a partially $\Delta$-embedded graph,
	whose compass contains a $w$-wall $W$
	with $\bor(\Delta)$ as perimeter. Let also ${\cal C}=[C_{1},\ldots,C_{x}]$ be the first $x$ layers of $W$ and  let $[D_{1},\ldots,D_{x}]$
	(resp. $[\overline{D}_{1},\ldots,\overline{D}_{x}]$)
	be the sequences of the corresponding open (resp. closed) disks  of $\Delta$ bounded by the cycles in ${\cal C}.$
	From \autoref{label_simultaneously}
	there is a collection ${\cal P}=\{P_{1},\ldots,P_{y}\}$ of paths in $W$ such that ${\cal A}=({\cal C},{\cal P})$
	is an $(x,y)$-railed annulus of $W$ where the outer cycle of $C$ is the perimeter of $W$
	and such that the vertices of the compass of the central $z$-subwall of $W$ belong to $D_{x},$ and Property~\emph{5} follows.

	Let $\breve{M}$ be the union of all subdivision paths of $(M,T)$ that intersect $\Delta\cap V(G)$
	and let $\breve{T}$ be the endpoints of these paths. Moreover, we denote $\breve{H}={\sf diss}(\breve{M},\breve{T})$ and observe that $\breve{H}$ is a subgraph of $H.$ Intuitively,
	$\breve{H}$ is the subgraph of $H$
	whose topological minor
	model $(\breve{M},\breve{T})$ is the part
	of $(M,T)$ that intersects the closed disk $\Delta.$
	As the $\Delta\cap V(G)$-intrusion of $(M,T)$
	in $G$  is at most $\ell,$ the same bound
	applies to the $\Delta\cap V(G)$-intrusion of $(\breve{M},\breve{T})$ in $G.$
	This in turn implies that $|\breve{T}\cap \Delta|\leq \ell$ and that $|E(\breve{H})|\leq \ell.$

	Since, $x= (\ell+1)\cdot x',$ there is a $b\leq \ell\cdot x'+1\leq x$ such that $\Bbb{A}:={\sf ann}({\cal C}, b,b+x'-1)$ does not contain any vertex of $T.$
	We define $T^{\rm out}=\breve{T}\setminus \overline{D}_{b}$ and $T^{\rm in}=\breve{T}\cap D_{b+x'-1}.$
	Clearly, $\{T^{\rm out},T^{\rm in}\}$ is a partition of $\breve{T}.$

	We set ${\cal A'}=([C_{b}, \ldots, C_{b+x'-1}],{\cal P}\cap \Bbb{A}).$
	By applying \autoref{label_interessiert} on  $s,$ $\breve{H},$ $g:=\ell,$  the $\Delta$-boundaried graph $G,$  the $(x',y)$-railed annulus ${\cal A}',$ the {\sf tm}-pair $(\breve{M},\breve{T}),$   and the set $I=[r],$ we have that
	$G$ contains a topological minor model $(\tilde{M},\breve{T})$ of $\breve{H}$ in $G$ such that
	$\tilde{M}$ is $(s,I)$-confined in ${\cal A}'$ and  $\tilde{M}\setminus \ann({\cal A}')\subseteq {\breve{M}}\setminus \ann({\cal A}').$
	We enhance $\tilde{M}$ by adding to it
	all subdivision paths of $(M,T)$ that are not intersecting $\Delta.$ That way, we have that
	$(\tilde{M},T)$ is a topological minor model of $H$ in $G$ 	such that
	$\tilde{M}$ is $(s,I)$-confined in ${\cal A}'$ and  $\tilde{M}\setminus \ann({\cal A}')\subseteq {{M}}\setminus \ann({\cal A}').$

	Let $p=b+\frac{x'-s}{2}$ and $q=b+\frac{x'+s}{2}-1$ and notice that $q\leq x.$
	We set $\Bbb{A}':={\sf ann}({\cal C}, p,q)$
	and we define ${\cal A}'':=([C_{p},\ldots C_q],{\cal P}')$  where ${\cal P}'={\cal P}\cap \Bbb{A}'.$
	Let ${\cal P}’=\{P_{1}’,\ldots,P_{y}'\}.$
	Observe that, from the second property of \autoref{label_interessiert},  the connected components of $\tilde{M} \cap \Bbb{A}'$ are some of the first $r$ paths in ${\cal P}'.$ This means that there is a subset of indices $\{i_{1},\ldots,i_{r'}\}\subseteq I$ such that $\tilde{M}\cap \Bbb{A}'=P'_{i_{1}}\cup\cdots\cup P_{i_{r'}}'.$
	Let $Z=\{z_{i_{1}},\ldots,z_{i_{r'}}\}$ be the set of endpoints of the paths $P'_{i_{1}},\ldots, P_{i_{r'}}'$ that are contained in  $C_{p}.$

	Let $\tilde{M}^{\rm in}=\tilde{M}\cap \overline{D}_{p},$ $\tilde{M}^{\rm out}=(\tilde{M}\setminus D_{p})\setminus E(C_{p}),$
	and  observe that $\tilde{M}=\tilde{M}^{\rm in}\cup \tilde{M}^{\rm out}$
	and that $Z=V(\tilde{M}^{\rm in})\cap V(\tilde{M}^{\rm out}).$
	%Notice also that $Z':=\{z_{i_{1}},\ldots,z_{i_{r'}}\}$
	Moreover, all vertices of $Z$  have  degree one in both $\tilde{M}^{\rm in}$ and $\tilde{M}^{\rm out}.$
	Let $\tilde{H}^{\rm in}$ (resp. $\tilde{H}^{\rm out}$) be the
	graph obtained from $\tilde{M}^{\rm in}$  (resp. $\tilde{M}^{\rm out}$) by
	dissolving all vertices except from those in $T^{\rm in} \cup Z$  (resp.~$T^{\rm out} \cup Z$).
	Note that $(\tilde{M}^{\rm in},T^{\rm in}\cup Z)$  (resp. $(\tilde{M}^{\rm out},T^{\rm out}\cup Z)$) is a topological minor model of $\tilde{H}^{\rm in}$ (resp.~$\tilde{H}^{\rm out}$).

	Notice that $\tilde{H}^{\rm in}$ has vertex set $T^{\rm in}\cup Z$ and
	can
	be seen as a $D$-embedded graph, for some closed disk $D,$ on at most $\ell+r$ edges where $\bor(D)\cap H=Z$ and $(z_{i_{1}},\ldots,z_{i_{r'}})$ is the ordering of the
	vertices of $Z$ as they appear in $C_{p}.$
	Observe now that
	$\tilde{H}^{\rm in}$ can be seen as the contraction of another $D$-embedded graph $\hat{H}^{\rm in}$ with detail at most $3\ell+r$ that has maximum degree at most three. Moreover, we can assume
	that the vertices of $\tilde{H}^{\rm in}$ that have degree at most three are also vertices of $\hat{H}^{\rm in}$ that are not affected by the contractions while transforming $\hat{H}^{\rm out}$ to $\tilde{H}^{\rm out}.$ This implies that there is a $Q$-respecting contraction-mapping  of $\tilde{H}^{\rm out}$ to $\hat{H}^{\rm out}.$
	Again, in the embedding of $\hat{H}^{\rm in}$ in $D,$ $(z_{i_{1}}\ldots,z_{i_{r'}})$ is the ordering of the
	vertices of $Z$ as they appear in $\bor(D).$

	Keep in mind that $\tilde{H}^+:=\tilde{H}^{\rm out}\cup \tilde{H}^{\rm in}$ is a minor of
	$\hat{H}^{+}:=\tilde{H}^{\rm out}\cup \hat{H}^{\rm in}$
	and that, if we dissolve in $\tilde{H}^+$ all the vertices in $Z,$ we obtain $H.$
	Also let $\hat{H}$ be the graph obtained if we dissolve in $\hat{H}^{+}$ all the vertices in $Z.$ Clearly, $\hat{H}$ is a minor of $H.$

	We now apply \autoref{label_dispenseront} for $c,r,r',3\ell,$ the $D$-embedded  graph $\hat{H}^{\rm in},$ the set $Z,$
	and the $(s,y)$-railed annulus ${\cal A}''$ of the $\overline{D}_{p}$-disk embedded graph $G\cap \overline{D}_{p}$ and obtain
	a {\sf tm}-pair $(\hat{M}^{\rm in},\hat{T}^{\rm in})$ of $G({\cal A}'')$ that is a topological minor model of $\hat{H}^{\rm in}$ and such that
	for each $j\in[r'],$ the function $\sigma_{\hat{M}^{\rm in},\hat{T}^{\rm in}}$
	maps  vertex $z_{i_{j}}$ to itself.
	%and  $\hat{T}^{\rm in}$ is well-scattered in $G({\cal A}'').$
	Notice that $G({\cal A}'')$ is a subgraph of $W\cap {\sf ann}({\cal C},p,q).$
	%therefore, $\hat{M}^{\rm in}\subseteq W\cap {\sf ann}({\cal C}_{b+\funref{label_quitablement}(c,\ell),l})$ (recall that ${\sf ann}({\cal C}_{p,q})\subseteq {\sf ann}({\cal C}_{b+\funref{label_quitablement}(c,\ell),l})$).
	%and $(\hat{M},\hat{T})\Cap D_{b+\funref{label_quitablement}(c,\ell)}$ is a {\sf tm}-pair of $W$
	From the second property of  \autoref{label_dispenseront},  $(\hat{M}^{\rm in},\hat{T}^{\rm in})$ is  safely $c$-dispersed in $W\cap {\sf ann}({\cal C},p,q).$
	From the third property of \autoref{label_dispenseront}, it follows that
	none of the vertices of $\hat{T}^{\rm in}\setminus\{w_{i_1},\ldots,w_{i_{r'}}\}$ is within face-distance less than $c$
	from some vertex of $C_{p}\cup C_{q}$ in $W\cap {\sf ann}({\cal C},p,q).$

	We now consider the graph $\hat{M}=\hat{M}^{\rm in}\cup \tilde{M}^{\rm out}.$ Property {\it 3} follows by the conclusions of the previous paragraph.
	Moreover,
	$\hat{M}$ does not intersect $D_{q}$ and, as $q\leq x,$  it neither intersects $D_{x},$
	%Therefore $\hat{M}$ does not intersect the vertices of the compass of the central $z$-wall of $W,$
	hence  Property {\it 4} holds.
	Notice also that  $\tilde{M}\setminus \ann({\cal A}')\subseteq {M}\setminus \ann({\cal A}')$ implies $\tilde{M}\setminus \overline{D}_{b}\subseteq {M}\setminus \overline{D}_{b}.$ This along with the fact that  $\hat{M}\setminus \overline{D}_{b}=\tilde{M}\setminus \overline{D}_{b},$ yield Property {\it 1}.
	% and, moreover,  $\hat{T}$ is well-scattered in $G({\cal A}'').$

	Observe that $(\hat{M}, \hat{T}^{\rm in}\cup T^{\rm out}\cup Z)$ is a topological minor model
	of $\hat{H}^+,$ which in turn implies that  $(\hat{M}, \hat{T}^{\rm in}\cup T^{\rm out})$ is a topological minor model
	of $\hat{H}.$ We now set $\hat{T}=\hat{T}^{\rm in}\cup T^{\rm out}.$ As  there is a $Q$-respecting contraction-mapping  of $\tilde{H}^{\rm out}$ to $\hat{H}^{\rm out},$ we also have that
	there is a $Q$-respecting contraction-mapping  of $H={\sf diss}(M,T)$ to $\hat{H}={\sf diss}(\hat{M},\hat{T})$ and Property {\it 6} holds.
	As  $\hat{T}^{\rm in}\subseteq \inter({\sf ann}({\cal A}'')\subseteq D_{p} =D_{b+\funref{label_quitablement}(\ell)}$ and $T^{\rm out}\subseteq G\setminus \overline{D}_{b},$ we deduce that $\hat{T}\in G\setminus{\sf ann}({\cal C},b,b+\funref{label_quitablement}(\ell) -1),$ which together with the fact that ${\sf ann}({\cal C},b,b+x'-1)$) does not contain any vertex of $T,$ yield Property {\it 2}.

	To conclude the proof, let us provide upper bounds on the claimed functions. By definition, it holds that
	\begin{eqnarray*}
		\funref{label_tourterelles}(c,\ell) & =& \Ocal( \ell \cdot (\funref{label_enlightening}(\ell)+\funref{label_endeavouring}(c,\funref{label_inevitability}(\ell)+1,3\ell+\funref{label_inevitability}(\ell)+1))),\\
		\funref{label_quitablement}(c,\ell) &=& \Ocal(\funref{label_enlightening}(\ell)
		+\funref{label_endeavouring}(c,\funref{label_inevitability}(\ell)+1,3\ell+\funref{label_inevitability}(\ell)+1)), \text{ and}\\
		\funref{label_presuntuosos}(c,z,\ell) &= &\Ocal(\ell \cdot \big(\funref{label_enlightening}(\ell) + \funref{label_endeavouring}(c,\funref{label_inevitability}(\ell)+1,3\ell+\funref{label_inevitability}(\ell)+1)\big) +\\
		& & z + \funref{label_inevitability}(\ell) + \funref{label_endeavouring}(c,\funref{label_inevitability}(\ell)+1,3\ell+\funref{label_inevitability}(\ell)+1)).
	\end{eqnarray*}
	%\red{
	%\item $\funref{label_presuntuosos}(c,z,\ell)= \Ocal(\ell \cdot \big(\funref{label_enlightening}(\ell) + \funref{label_endeavouring}(c,\funref{label_enlightening}(\ell)+1,3\ell+\funref{label_enlightening}(\ell)+1)\big) +z ).$
	%}
	Since by \autoref{label_interessiert} we have that $\funref{label_enlightening}(\ell)=\Ocal((\funref{label_inevitability}(\ell))^2),$ by \autoref{label_dispenseront} we have that $\funref{label_endeavouring}(c,r,\ell)=\Ocal(cr(\ell+r)),$ and by \autoref{label_constitueraient}
	we may assume that $\ell = \Ocal(\funref{label_inevitability}(\ell)),$ the above can be simplified to
	\begin{eqnarray*}
		\funref{label_tourterelles}(c,\ell) &=&\Ocal( \ell \cdot \big((\funref{label_inevitability}(\ell))^2 + c  \cdot  (\funref{label_inevitability}(\ell))^2\big))\\
		&=& \Ocal( c \cdot (\funref{label_inevitability}(\ell))^3),\\
		\funref{label_quitablement}(c,\ell) &=&\Ocal( c \cdot (\funref{label_inevitability}(\ell))^2), \text{ and}\\
		\funref{label_presuntuosos}(c,z,\ell) &=&  \Ocal(\ell \cdot \big((\funref{label_inevitability}(\ell))^2 +  c \cdot   (\funref{label_inevitability}(\ell))^2 \big)+ z)\\
		&=& \Ocal( c \cdot (\funref{label_inevitability}(\ell))^3   + z),
	\end{eqnarray*}
	and the theorem follows.
\end{proof}

\subsection{Levelings and well-aligned flatness pairs}
\label{label_disintegrated}

Let $G$ be a graph and let $(W,\frak{R})$ be a flatness pair of $G.$
Let also  $\frak{R}=(X,Y,P,C,\Gamma,\sigma,\pi),$  where $(\Gamma,\sigma,\pi)$ is an $\Omega$-rendition of $G[Y]$ and $\Gamma=(U,N)$ is a $\Delta$-painting.
%We call $(X,Y)$ the {\em separation certifying the flat wall $W$} and $X\cap Y$ is called the  {\em frontier} of $W.$
The {\em ground set} of $W$ in ${\frak{R}}$ is ${\sf ground}_{\frak{R}}(W):=\pi(N(\Gamma))$ and we refer to the vertices of this set as the {\em ground vertices} of the $\frak{R}$-compass of $W$ in $G.$
Notice  that ${\sf ground}_{\frak{R}}(W)$ may  contain vertices
of ${\sf compass}_{\frak{R}}(W)$ that are not necessarily vertices of $W.$

In the flatness pairs of  \autoref{label_grandiloquents}   and
\autoref{label_indubitadamente} the ground vertices are the vertices on the boundaries of
the green cells. (Notice also that the  flatness pair in \autoref{label_grandiloquents} is regular, while the one in \autoref{label_indubitadamente}  is not.)

\begin{figure}[t]
	\begin{center}
		\includegraphics[width=11cm]{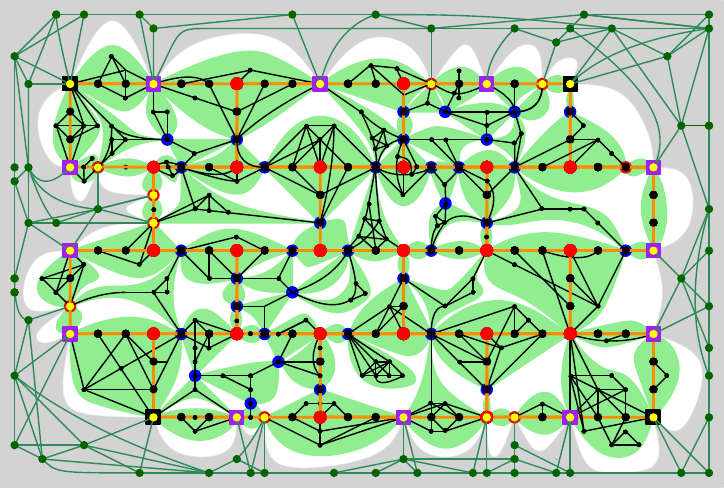}
	\end{center}
	\caption{A graph $G$ and a  regular flatness pair  $(W,\frak{R})$ of  $G.$}
	\label{label_grandiloquents}
\end{figure}

\paragraph{Levelings.}
We define  the $\frak{R}$-{\em leveling}  of $W$ in $G,$
denoted by ${W}_{\frak{R}},$ as the bipartite graph
where  one part is the ground set of $W$ in $\frak{R},$ the  other part is a set ${\sf vflaps}_{\frak{R}}(W)=\{v_{F}\mid F\in {\sf flaps}_{\frak{R}}(W)\}$ containing one new vertex $v_{F}$ for each flap  $F$ of $W$ in $\frak{R},$
and, given  a pair $(x,F)\in {\sf ground}_{\frak{R}}(W)\times {\sf flaps}_{\frak{R}}(W),$   the set $\{x,v_F\}$ is an edge of ${W}_{\frak{R}}$ if and only if
$x\in \partial F.$ We call the vertices of ${\sf ground}_{\frak{R}}(W)$ (resp. ${\sf vflaps}_{\frak{R}}(W)$) {\em ground-vertices} (resp. {\em flap-vertices}) of ${W}_{\frak{R}}.$
%Again, keep in mind that ${W}_{\frak{R}}$ may contain (many) vertices that are not in $W.$
Notice that the incidence graph of the plane hypergraph $(N(\Gamma),\{\tilde{c}\mid c\in C(\Gamma)\})$ is isomorphic to ${W}_{\frak{R}}$
via an isomorphism that extends  $\pi$ and, moreover, bijectively corresponds cells to flap-vertices.
This permits us to treat ${W}_{\frak{R}}$ as a $\Delta$-embedded graph where  $\bd(\Delta)\cap {W}_{\frak{R}}$ is the set $X\cap Y.$
As an example,  see \autoref{label_outbuildings} for the $\frak{R}$-leveling corresponding to the flatness pair $(W,\frak{R})$ in \autoref{label_grandiloquents}.
%The graph $W^\bullet$ can be seen as a ``slight variant'' of $W$  that will be important at the end of this subsection for representing
%$W$ by a wall in $W_{\frak{R}}.$

The following observation is a consequence of the definition of leveling and condition
\ref{label_communication}
of the tightness property of a rendition.
\begin{observation}\label{label_subdivisiones}
	Let $G$ be a graph, let $(W,\frak{R})$ be a flatness pair of $G,$ and let ${W}_{\frak{R}}$ be the leveling of $W$ in $G.$ For every $v_F \in {\sf vflaps}_{\frak{R}}(W)$ of degree $r$ in ${W}_{\frak{R}},$ there exist $r$ internally vertex-disjoint paths in ${W}_{\frak{R}}$ from $v_F$ to $r$ distinct ground-vertices of ${W}_{\frak{R}}$ that belong to the perimeter of $W.$
\end{observation}

\begin{figure}[t]
	\begin{center}
		\includegraphics[width=11cm]{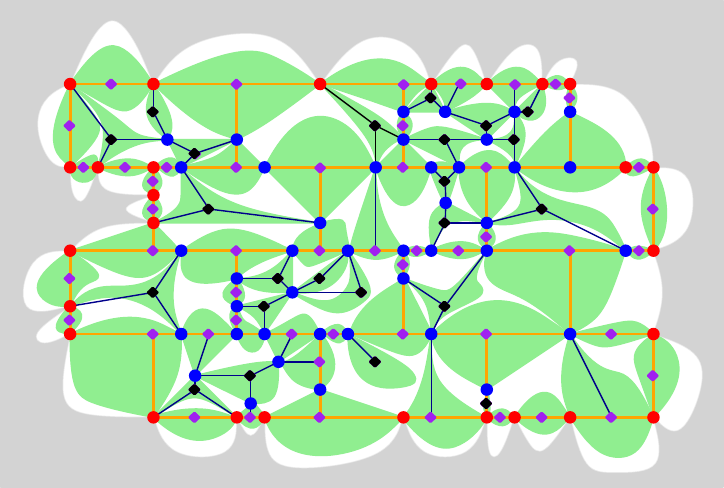}
	\end{center}
	\caption{The leveling $W_{\frak{R}}$ corresponding to the  regular flatness pair $(W,\frak{R})$ in \autoref{label_grandiloquents}. The ground-vertices of $W_{\frak{R}}$
	are the circled vertices while the flap-vertices are the rhombic vertices. The representation $R_{W}$ of $W$ in $W_{\frak{R}}$
	is obtained from $W_{\frak{R}}$ after removing the black squared vertices.
	The ground vertices in $\bd(\Delta)\cap {W}_{\frak{R}}=\bd(\Delta)\cap {W}_{\frak{R}}$ are depicted in red.}
	\label{label_outbuildings}
\end{figure}

\paragraph{Well-aligned flatness pairs.}
{We denote by $W^{\bullet}$ the  graph obtained from $W$ if we subdivide {\em once} every
	edge of $W$ that is {short} in ${\sf compass}_{\frak{R}}(W).$}
The graph $W^\bullet$  is  a ``slightly richer variant'' of $W$  that is necessary for our definitions and  proofs, namely to be able to associate  every flap-vertex of  an appropriate subgraph of $W_{\frak{R}}$ (that we will denote by $R_{W}$) with  a non-empty path of $W^\bullet,$ as we proceed to formalize.
We say that $(W,\frak{R})$ is {\em well-aligned} if the following holds:
\begin{quote}
	$W_{\frak{R}}$ contains as a subgraph an $r$-wall $R_{W}$
	where ${D(R_{W})}=D({W}_{\frak{R}})$ and $W^{\bullet}$  is isomorphic to some subdivision of  $R_{W}$
	via an isomorphism that maps each ground vertex to itself.
\end{quote}
Suppose now that the flatness pair $(W,\frak{R})$ is well-aligned.
We call the wall  $R_{W}$ in the above condition  a {\em representation} of $W$ in $W_{\frak{R}}.$ Note that, as $R_{W}$ is a subgraph of $W_{\frak{R}},$ it is bipartite as well.  The above property gives us a way to represent a flat wall by a wall of its leveling
in a way that ground vertices are not altered.
The following proposition, proved in~\cite{accurate}, indicates that such
a representation is yielded by regularity.

\begin{proposition}
	\label{label_esaminadores}
	Every  regular  flatness pair $(W,\frak{R})$ of a graph $G$  is  well-aligned.
	%Moreover, there is an {$\Ocal(n)$} time algorithm that, given $G$ and such a $({W},\frak{{R}}),$ outputs a representation $R_{W}$
	%of $W$ in $W_{\frak{R}}.$
\end{proposition}

Notice that both $W_{\frak{R}}$ and its subgraph $R_{W}$ can be seen as $\Delta$-embedded graphs where $\bd(\Delta)\cap {W}_{\frak{R}}=\bd(\Delta)\cap R_{W}\subseteq V(D(W_{\frak{R}}))=V(D(R_{W})).$
This establishes a bijection from the set of cycles of $W$  to the set of cycles of $R_W,$
%${\sf rep}_{\frak{R}}(C)$ of $R_{W},$ and observe the following.
%\red{ and from each subwall ${W}'$ of $W$ to a subwall ${\sf rep}_{\frak{R}}({W}')$ of $R_{W}.$}
which allows to reinterpret the homogeneity property of a regular flatness pair in terms of its representation, as stated in the following observation. This translation will be used in the proof of \autoref{label_pretendientes}. Given the $\Delta$-embedded graph $R_{W},$ we define,  for every brick $B$ of $R_W,$  ${\sf vflaps}_{R_W}(B)$ as the flap-vertices of the leveling ${W}_{\frak{R}}$ that belong to the closed disk of the plane bounded by $B$ disjoint from the infinite face. (Recall \autoref{label_disconvenevole} for the definition of the augmented flap ${\bf F}^{A}$ corresponding to a flap-vertex $v_F$ of the leveling ${W}_{\frak{R}}.$)

\begin{observation}\label{label_comprometida}
	If $(A,W,\frak{R})$ is an $(a,r,\ell)$-apex-wall triple of a graph $G$ and $R_W$ is the representation of $W$ in $W_{\frak{R}},$ then for any two internal bricks $B,B'$ of $R_W,$ it holds that
	\begin{equation*}\label{label_commercially}
		\{{\ell}\mbox{\sf -folio}({\bf F}^{A})\mid v_F\in {\sf vflaps}_{R_W}(B) \} = \{{\ell}\mbox{\sf -folio}({\bf F}^{A})\mid v_F\in {\sf vflaps}_{R_W}(B')\}.
	\end{equation*}
\end{observation}
Note that, in the above equation, a flap $F$ is notationally associated with both ${\bf F}^{A}$ and $v_F.$

\subsection{Rerouting minors of small intrusion}
\label{label_preoccupation}

Let $W$ be a plane-embedded $r$-wall and $c\geq 1.$
We call a cycle $C$ of $W$ {\em $c$-internal} if $V(C)$ and $V(D(W))$ are within face-distance at least $c.$
Given a $1$-internal cycle $C$ of $W,$ we define its {\em internal pegs} (resp.  {\em external pegs})
as its vertices that are incident to edges of $W$ that belong to the interior (resp. exterior) of $C$ with respect to the embedding of the wall (we see edges as open sets).
Notice that each vertex of $C$ is either an internal or an external peg.

\begin{observation}
	\label{label_deliberadament}
	%There exists a constant $\newcon{sadfdsfdsf}$ such that the following holds.
	Let $W$ be an $r$-wall and let $C_{1}$ and $C_{2}$ be four cycles of $W$ within face-distance at least \red{four} and such that
	$C_{2}$ is a subset of the closed disk bounded by $C_{1}.$ Let $y \in [3],$ let $p_1,\ldots,p_y$  be internal pegs of $C_{1}$ and $\bar{p}_1,\ldots,\bar{p}_y$ be
	external pegs of $C_{2},$ assuming that both these sets of vertices are ordered as they appear in their
	corresponding cycles
	in counter-clockwise order.
	%If the face-distance, in $W,$  between $C_{1}$ and $C_{2}$  is at least $\conref{sadfdsfdsf},$
	Then there are $y$ pairwise vertex-disjoint paths $\hat{P}_1, \ldots,\hat{P}_y$ such that, for $i \in [y],$ $\hat{P}_i$ joins $p_{i}$ with $\bar{p}_{i},$ $V(\hat{P}_i) \cap V(C_1) = \{ p_i \},$ and $V(\hat{P}_i) \cap V(C_2) = \{\bar{p}_i\}.$
\end{observation}

%\ig{the constant  ``sadfdsfdsf'' is always 2!!!}
%\begin{proof}
%\ig{MISSING PROOF -- Menger?}
%\end{proof}

Given a $1$-internal brick $B$ of $W,$  one can see the union of all bricks of $W$ that have
a common vertex with $B,$ as a subdivision of the graph in the left part of  \autoref{label_incominciommi}.
We call this subgraph $X$ of $W$ the {\em brick-neighborhood} of $B$ in $W.$ The \emph{perimeter}  of a brick-neighborhood
is defined in the obvious way.% (see \autoref{label_incominciommi}).

\begin{figure}[h!]
	\begin{center}%\vspace{-.1cm}
		\includegraphics[width=.85\textwidth]{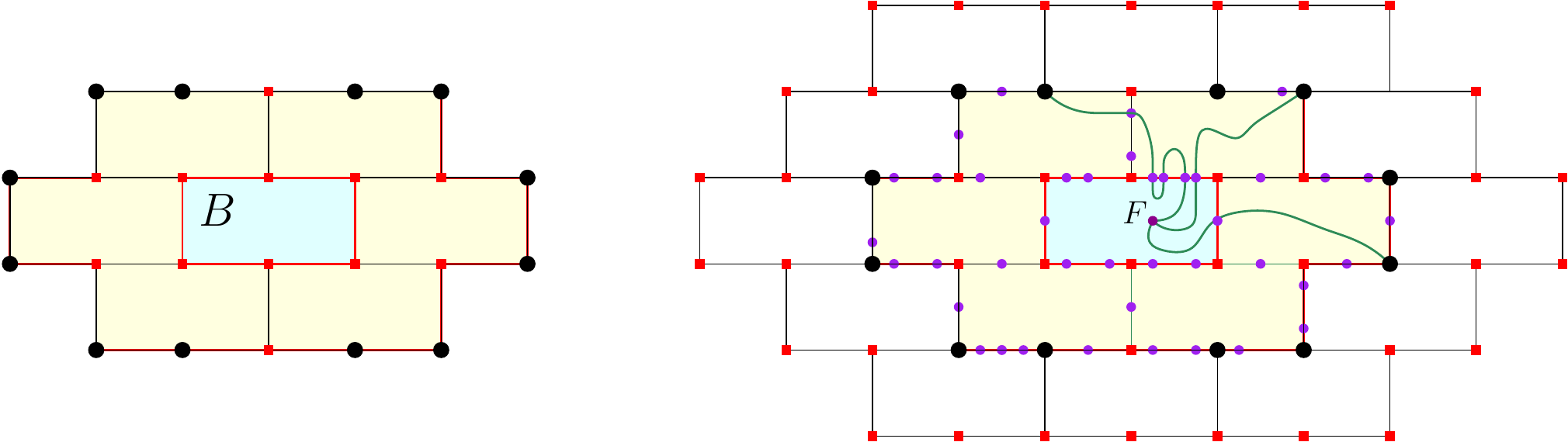}
	\end{center}\vspace{-.25cm}
	\caption{On the left: the base graph for the definition of a brick-neighborhood -- the external pegs of the perimeter  of $X$ are the black round vertices. On the right we depict  a $2$-internal brick $B$ of an $r$-wall $W,$ $r\geq 6,$
		contained as a subgraph in a plane graph $G$ along with  three  internally vertex-disjoint paths from a vertex $F$ of $G$  to the external pegs of the perimeter of the  brick-neighborhood of $B.$}
	\label{label_incominciommi}
	%  \vspace{.35cm}
\end{figure}

The next lemma is based on~\autoref{label_subdivisiones}.
\begin{lemma}
	\label{label_ofrecimiento}
	Let $(W,\frak{R})$ be a well-aligned flatness pair of a graph $G$ and let $R_{W}$  be its representation in the leveling ${W}_{\frak{R}}$ of $W.$
	For every 2-internal brick $B$ of $R_{W}$ and every flap vertex $v_F\in {\sf vflaps}_{\frak{R}}(B),$ $R_{W}$ contains  $|\partial F|$ internally vertex-disjoint paths from $v_F$ to the external pegs of the perimeter of the brick-neighborhood of $B$ in $R_{W}.$ Moreover, these paths belong to the closed
	disk  bounded by the perimeter of the brick-neighborhood of $B$ in $R_W.$
\end{lemma}

\begin{proof}
	%\ig{MISSING PROOF}
	Let $r=|\partial F|,$ $X$ be the brick-neighborhood of $B$ in $R_{W},$  and $P$ be the perimeter of $X$ in $R_{W}.$
	%, and $D$ be the closed 	disk of the plane bounded by $P.$
	We call {\em frontier-path} of $X$ a subpath of $P$ that joins two external pegs and does not contain any other external peg. Notice that $P$ is the union of the frontier-paths of $X$ and that there are exactly 12 such paths.
	Notice that for every  frontier-path $Q$ of $X$ there is a path $\hat{Q}$ of $X$ such that
	\begin{itemize}
		%\item $\hat{Q}$ is a subset of $D,$
		\item its endpoints are in $P$ but not in $Q,$
		\item it does not contain any internal vertex in $P,$ and
		\item every path in ${W}_{\frak{R}}$ from $v_F$ to a vertex of $Q$ intersects some vertex of $\hat{Q}.$
	\end{itemize}
	See \autoref{label_correspondingly} for two indicative examples of the above correspondence.

	\begin{figure}[h!]
		\begin{center}%\vspace{-.1cm}
			\includegraphics[width=.55\textwidth]{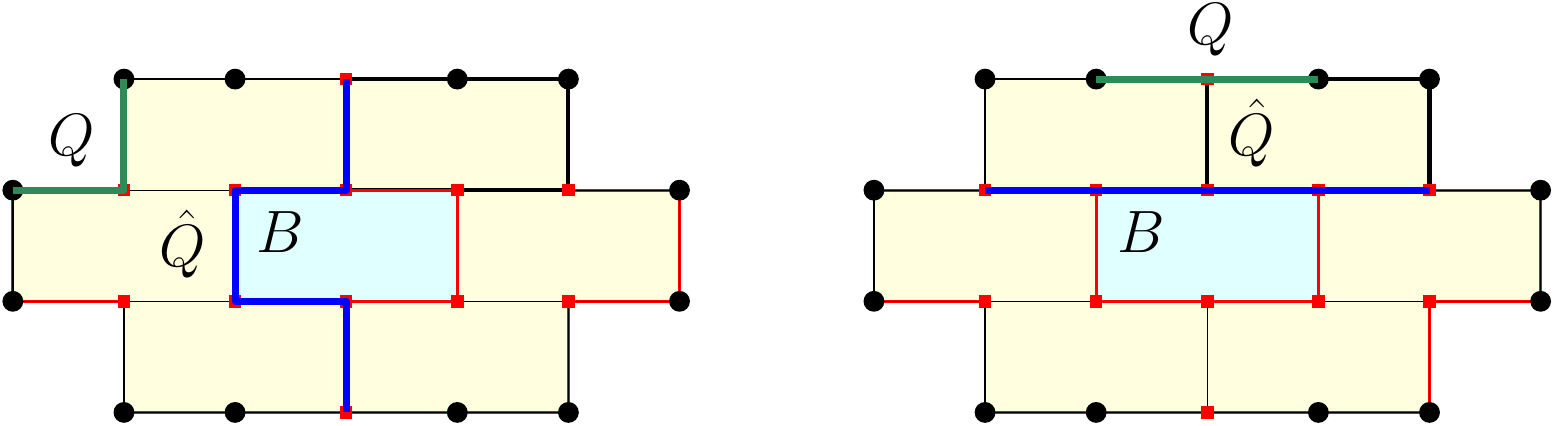}
		\end{center}\vspace{-.25cm}
		\caption{Two indicative examples for the choice of $\hat{Q}$ (in blue), depending on the choice of $Q$ (in green).}
		\label{label_correspondingly}
		%  \vspace{.35cm}
	\end{figure}

	From \autoref{label_subdivisiones}, there are $r$ internally vertex-disjoint paths $P_{1},\ldots,P_{r}$ in ${W}_{\frak{R}}$ starting from $v_F$
	and finishing in vertices $v_{1},\ldots,v_{r}$ of $P.$ Moreover, we can assume that
	$(P_{1}\cup\cdots\cup P_{r})\cap P=\{v_{1},\ldots,v_{r}\}.$ The lemma is trivial in case $r<3$ as the paths  $P_{1},\ldots,P_{r}$ can easily be extended so to finish in external pegs of $P.$
	Moreover, the lemma also follows easily if $r=3$ and $v_{1},v_{2},v_{3}$ do not all belong to some of the frontier-paths of $X.$
	It remains to examine the case where $v_{1},v_{2},v_{3}$ are vertices of some frontier-path $Q$ of $X.$
	Let $\hat{Q}$ be as defined above and let $\hat{q}$ be one of its endpoints.  Let $z_{Q}$ be the first vertex of $P_{1}\cup P_2 \cup P_3$ that is met while following $\hat{Q}$ starting from $\hat{q}$ and moving towards its other
	endpoint. The vertex  $z_{Q}$ exists because of the definition of $\hat{Q}.$
	W.l.o.g., we assume that $z_{Q}$ is an internal vertex of $P_{1}.$
	We now define $P_{1}'$ by first removing from $P_1$ all vertices of its subpath
	from $z_{Q}$ to $v_{i},$ except $z_{Q}$ and then taking the union of the
	resulting path with the subpath of $\hat{Q}$ between $\hat{q}$ and $z_{Q}.$
	It is now easy to see that $P_{1}$ can be extended to some external peg that is
	different from the endpoints of $P_{2}$ and $P_{3},$ while $P_{2}$ and $P_{3}$
	can be extended towards the external pegs that are endpoints of $\hat{Q}.$ This completes the proof of the lemma.
\end{proof}

\paragraph{Irrelevant vertices in boundaried graphs.}

Let $G$ be a graph, $H$ be a minor of $G,$ and $S\subseteq V(G).$
We define the {\em $S$-minor-intrusion of $H$} in $G$ as the minimum  $S$-intrusion in $G$
over all  {\sf tm}-pairs $(M,T)$ of $G$ such that $(M,T)$ is a topological minor model
in $G$ where ${\sf diss}(M,T)\in {\sf ext}{(H)}.$

Let ${\bf Z}=(Z,B,\rho)$ be a $t$-boundaried graph and let $\ell\in \Bbb{N}.$ We say that a vertex set $S \subseteq V(Z)\setminus B$ is $\ell$-{\em irrelevant}
if for  every boundaried graph ${\bf K}=(K,B,\rho)$
that is compatible with ${\bf Z},$ every
minor   of ${\bf K}\oplus {\bf Z}$ with $(V(Z)\setminus B)$-minor-intrusion
at most $\ell,$ is also a minor of  ${\bf K}\oplus (Z\setminus S, B,\rho).$ Informally, an $\ell$-irrelevant set of vertices can be removed without affecting the occurrences of any minor of  minor-intrusion at most $\ell,$ where the intrusion is defined without taking into account the terminal vertices in the boundary.

%graph  $H$ and every boundaried graph ${\bf C}=(C,B,\rho)$ that
%is compatible with ${\bf Z},$  if $(T,M)$ is a {\sf tm}-pair in ${\bf C}\oplus {\bf Z}$ that is a topological minor model of a graph in ${\sf ext}(H)$ where $|T\setminus V(C)|\leq \ell,$ then  $H$ is a minor of ${\bf C}\oplus (Z\setminus v, B,\rho).$
\medskip

Using~\autoref{label_constitutivos}, we can finally prove the main result of this section.

%\ig{why was Giannos bothered about the ``$a$'' in the next theorem? $\funref{label_encompassing}(a,s,\ell)$ or $\funref{label_encompassing}(s,\ell)$?}

%\ig{I changed the following lemma to a theorem}

\begin{theorem}\label{label_pretendientes}
	There exist two functions $\newfun{label_encompassing}: \Bbb{N}^3\to\Bbb{N}$ and $\newfun{label_intervention}: \Bbb{N}^2\to\Bbb{N}$ such that,
	for every $a,z,\ell\in\Bbb{N}$ and  every boundaried graph ${\bf Z}=(Z,B,
		\rho),$ if  $(A,W,\frak{R})$ is  an %$(\ell,a)$-homogeneous
	$(a,\funref{label_encompassing}(a,z,\ell),\funref{label_intervention}(a,\ell))$-apex-wall triple of $Z$ that is not affected by $B,$
	%\sed{I removed the text: «where $W$ is a wall that is $\ell$-homogeneous with respect to $(G,A)$»}
	then the vertex set of the compass of every $W'$-tilt of $(W,\frak{R}),$ where $W'$ is the
	central $z$-subwall of $W,$ is $\ell$-irrelevant.
	Moreover, it holds that $\funref{label_encompassing}(a,z,\ell)=\Ocal( (\funref{label_inevitability}(16a+12\ell))^3   + z)$
	and $\funref{label_intervention}(a,\ell)=a+3+\ell.$
\end{theorem}
%\ig{forward to Giannos}
\begin{proof}
	Let $\tilde{\ell}=16a +12\ell,$ $\hat{\ell} = a +3 + \ell ,$ $r = \funref{label_presuntuosos}({6},z,\tilde{\ell}),$ and $\tilde{r} = \funref{label_tourterelles}({6},\tilde{\ell}).$ We prove the theorem for $\funref{label_encompassing}(a,z,\ell)=r$ and $\funref{label_intervention}(a,\ell)=\hat{\ell}.$ Note that, by \autoref{label_constitutivos}, $\funref{label_encompassing}(a,z,\ell) = \Ocal( (\funref{label_inevitability}(\tilde{\ell}))^3   + z).$
	%\ig{CAREFUL! the function $\funref{label_quitablement}$ from  \autoref{label_constitutivos} is NOT used in this proof!}
	Let ${\bf Z}=(Z,B,\rho)$ and let $(A,W,\frak{R})$ be an  $(a,r,\hat{\ell})$-apex-wall triple  of the graph $Z$ that is not affected by $B.$
	%where $a=|A|$ and $W$ is an $\ell$-homogenous with respect to $(G,A)$
	%$(\ell,a)$-homogeneous
	%wall of height $r.$
	Let $W'$ be the central $z$-subwall of $W,$ and let $(\tilde{W}',\tilde{\frak{R}}')$  be a flatness pair of $Z \setminus A$ that is a $W'$-tilt of $(W,\frak{R}).$ Our objective is to prove that $V({\sf compass}_{\tilde{\frak{R}}'}(\tilde{W}'))$ is $\ell$-irrelevant.

	Let ${\bf K}=(K,B,\rho)$ be a boundaried graph  compatible with ${\bf Z}.$
	As  $(A,W,\frak{R})$ is not affected by $B,$ we have that $(A,W,\frak{R})$ is an  $(a,r,\hat{\ell})$-apex-wall triple  of the graph $G:={\bf K}\oplus {\bf Z}$ as well.
	Let now $H$ be a minor of $G$
	whose $(V(Z)\setminus B)$-minor-intrusion in $G$ is at most $\ell.$ Let also
	$(M,T)$ be a {\sf tm}-pair in $G$
	that is a topological minor model of a graph $H'\in {\sf ext}(H),$ such that $|T\setminus V(K)|=|T\cap (V(Z)\setminus B)|\leq \ell$
	and with at most $\ell$ subdivision paths intersecting $V(Z)\setminus B.$
	%(in \autoref{label_disadvantages}, $H'$ can be seen a the graph consisting of the green and the red edges).
	Our purpose is to find a  {\sf tm}-pair $(\overline{M},\overline{T})$ of
	$G$ where $V({\sf compass}_{\tilde{\frak{R}}'}(\tilde{W}'))\cap  V(\overline{M})=\emptyset$ and such that $H' = {\sf diss}(M,T)$ is a minor of ${\sf diss}(\overline{M},\overline{T}),$ which in turn implies  that $H$ is a minor of $G\setminus V({\sf compass}_{\tilde{\frak{R}}'}(\tilde{W}'))={\bf K}\oplus (Z\setminus V({\sf compass}_{\tilde{\frak{R}}'}(\tilde{W}')), B,\rho),$ as required.\medskip

	\begin{figure}[h!]
		\begin{center}%\vspace{-.1cm}
			\includegraphics[width=.67\textwidth]{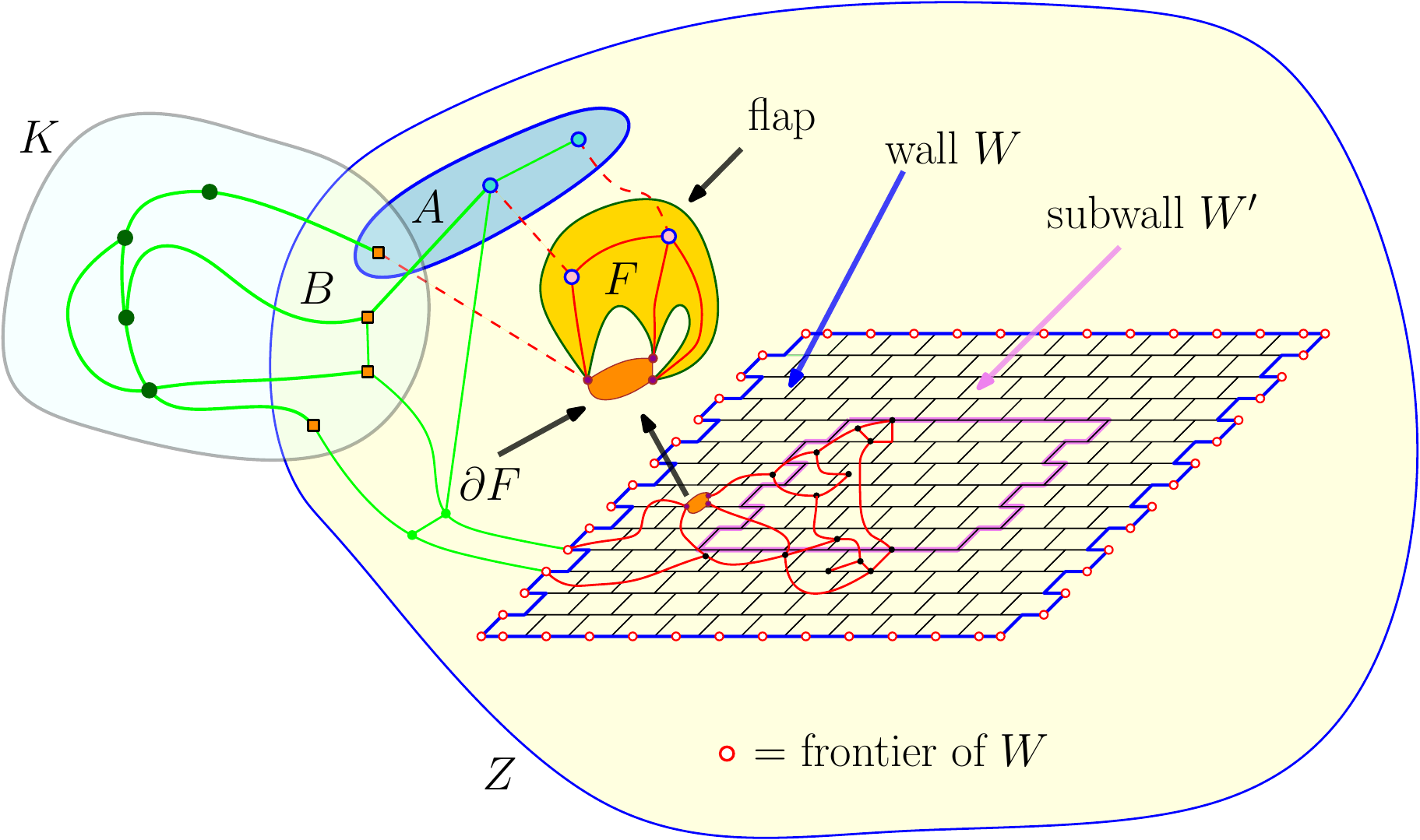}
		\end{center}\vspace{-.25cm}
		\caption{An illustration of the graph $G={\bf K}\oplus {\bf Z}$ and a {\sf tm}-model $(M,T)$ in it.
		The apex set $A$ is  cyan. We also draw a single flap $F$ and its magnification.
		%\ig{change $C$ to $K$} \ig{check all names}
		}
		%\ig{I really like this figure!}}
		\label{label_disadvantages}
		\vspace{.35cm}
	\end{figure}

	We proceed with the definition of a series of auxiliary graphs that will permit
	to work on a partially planarized version of $G.$ Suppose that $\frak{R}=(X,Y,P,C,\Gamma,\sigma,\pi).$
	We  first define  the graph ${G}_{\rm out}$
	%\ig{it should be ${G}_{\rm out},$ right?}
	as the graph obtained from $G$ if we remove
	all the vertices of ${\sf compass}_{\frak{R}}(W),$ except from those in $X \cap Y.$ %\red{and then
	%remove all edges between $A$ and ${\sf compass}_{\frak{R}}(W)$} \ig{shall we remove this sentence?}.
	We then define
	$\tilde{G}:=G_{\rm out}\cup {W}_{\frak{R}}$ where ${W}_{\frak{R}}$ is the
	leveling of $W.$
	%Notice that $\tilde{G}$ is a partially $\Delta$-embedded graph  whose compass is $\tilde{W}$ and where $\bor(\Delta)$   is the perimeter of $W.$
	%
	Since $(W,\frak{R})$ is regular, it is also well-aligned by \autoref{label_esaminadores}, hence ${W}_{\frak{R}}$ contains a representation $R_W$ of $W,$ which is
	a subgraph  of ${W}_{\frak{R}}$ that is isomorphic to some subdivision of $W^{\bullet}$
	via an isomorphism that maps each ground vertex to itself (recall that $W^{\bullet}$ is the  graph obtained from $W$ if we subdivide once every short edge in $W$).
	%	Similarly, the subwall $\breve{W}$ of $W$ corresponds to a subwall $\breve{W}’$ of the  $d$-wall $W’$ and we call $\breve{W}’$ {\em representation} of $\breve{W}$ in $\tilde{W}.$
	Notice that $\tilde{G}$ is a partially $\Delta$-embedded graph whose compass is ${W}_{\frak{R}}$ that, in turn,
	is the compass in $\tilde{G}$ of  the  $d$-wall $R_W.$
	%	Let now $\breve{\Delta}$ be the closed disk in $\Delta$ that is
	%	bounded by the perimeter of $\breve{W}.
	%	$ The graph  $\tilde{G}$ can also be seen as  a partially $\breve{\Delta}$-embedded
	%	graph whose compass is $\tilde{W}\cap \breve{\Delta}.$
	Recall that each flap $F$ of $W$ corresponds to a flap-vertex $v_F$ of ${W}_{\frak{R}}.$
	%	 and therefore of $\tilde{W}\cap \breve{\Delta}$ as well.
	\medskip

	We  enhance $(M,T)$ by defining another {\sf tm}-pair  $(M_A,T_A),$ where  $M_A=(V(M)\cup A,E(M))$ and $T_A=T\cup A,$ i.e., $(M_A,T_A)$ is obtained from $(M,T)$ by including all the apices in $A.$ We set $H_A={\sf diss}(M_A,T_A)$ and
	observe that $H$ is a (topological) minor of $H_A.$
	We set $T_{\rm in}=A\cup(T_A\cap {\sf compass}_{\frak{R}}(W))$ and $T_{\rm out}=A\cup(T_A\setminus  ({\sf compass}_{\frak{R}}(W)\setminus (X \cap Y))).$
	%	(in \autoref{label_disadvantages}, the vertices in $T_{\rm in}$ are those incident to the red edges, while the vertices in $T_{\rm out}$ are those incident to the green edges).
	As $(A,W,\frak{R})$  is not affected by $B,$ it follows that
	$T_{\rm in}\cap {\sf compass}_{\frak{R}}(W)\subseteq T\setminus (V(K) \cup B),$  hence  $|T_{\rm in}\cap {\sf compass}_{\frak{R}}(W)|\leq \ell.$
	%This in turn implies that  the $V({\sf compass}_{\frak{R}}(W))$-intrusion of the new $(M,T)$ in $G$ is at most $\ell.$

	%We proceed with a series of definitions.
	% $G:={\bf C}\oplus{\bf Z}.$
	%Let also $(\Gamma,\sigma,\pi)$ be a {rendition}  of the compass of $W$ in $G.$
	%We denote by  $\partial A$
	%all the vertices of $A$ that have neighbors outside $A\cup {\sf compass}_{\frak{R}}(W).$

	%Let now $(M,T)$ be a {\sf tm}-pair of  $G$ with $A\subseteq T.$\sed{We include $A$ in $T$ for reasons as this simplifies the description later.}  We also set $T_{\rm in}=A\cup(T\cap {\sf compass}_{\frak{R}}(W)).$
	We call a vertex $v\in {\sf compass}_{\frak{R}}(W)$ an {\em apex-jump} vertex if there exists an edge $\{v,w\}$ of $M_A$ with $w\in A.$
	%	 (in~\autoref{label_disadvantages} and \autoref{label_lestrygonians}, the lower endpoints of the three dashed red edges are apex-jumps).
	%Observe that there are at most   apex-jump-vertices in ${\sf compass}_{\frak{R}}(W).$
	Notice that there are at most {$a+\ell$} apex-jump vertices.
	%p30 line 4: Maybe we can develop why there is at most $a+\ell$ apex-jump vertices.
	We define a set of {\em $(M_A,T_A)$-dirty} flaps of the flatness pair $(W,\frak{R})$  by applying the following definition: first we declare as  $(M_A,T_A)$-dirty every
	flap $F$ such that $V(F)\setminus \partial F$ contains an apex-jump vertex or a vertex in $T_{\rm in}.$ Second, for every
	apex-jump vertex $v\in {\sf ground}(W)$ that is not in the boundary of some $(M_A,T_A)$-dirty
	flap,  we arbitrarily pick a flap $F$ with $v\in \partial F$ and we declare it  {\em $(M_A,T_A)$-dirty}.
	% (note that $F$ may have been already declared dirty in the first step).
	Observe that $(W,\frak{R})$ has at most $a + \ell + |T_{\rm in}\cap {\sf compass}_{\frak{R}}(W)|\leq {a + 2\ell}$
	$(M_A,T_A)$-dirty flaps.

	\begin{figure}[htb]
		\begin{center}%\vspace{-.1cm}
			\includegraphics[width=.72\textwidth]{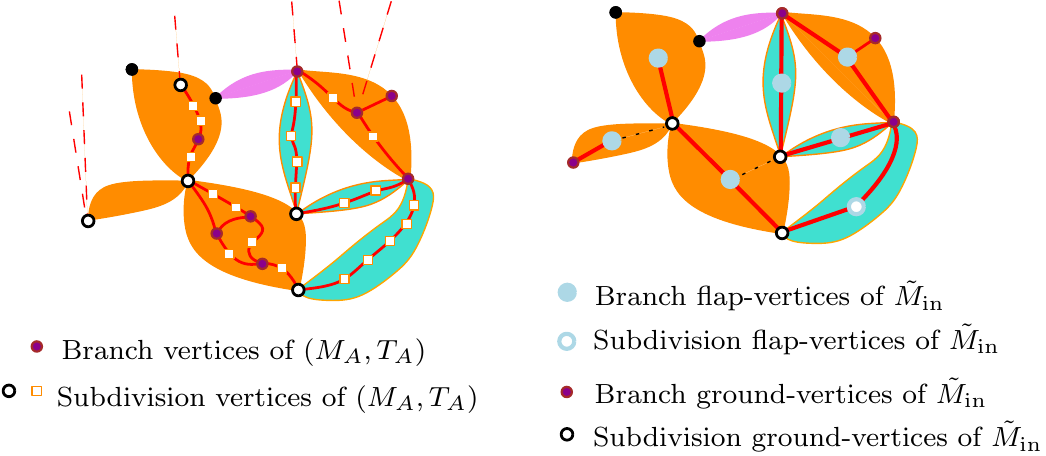}
		\end{center}\vspace{-.25cm}
		\caption{An example of how part of the  model $(M_A,T_A)$ (in red edges) traverses the flaps of ${W}$ along with the model $(\tilde{M}_{\rm in},\tilde{T}_{\rm in})$ of $\tilde{W}$ (i.e., the leveling of $(M_A,T_A)$). The dashed red edges are edges pointing to apices. The orange flaps are the $(M_A,T_A)$-dirty flaps. The black vertices are the ground vertices that do not belong to $M.$}
		\label{label_lestrygonians}
		\vspace{.35cm}
	\end{figure}

	We define   $\tilde{M}_{\rm in}$ as the subgraph of the leveling ${W}_{\frak{R}}$ induced
	%\sed{Explain some sugar...}
	by the vertices in ${\sf ground}(W)\cap V(M_{A})$ and all the flap-vertices $v_F$ of ${W}_{\frak{R}}$ such that $F$
	is either $(M_A,T_A)$-dirty or  contains an edge of $M$ (in \autoref{label_lestrygonians}, these latter flaps are turquoise).
	%If some of the induced edges in $\tilde{M}_{\rm in}$ has one endpoint in $V(M)\setminus T$ and the other
	%is an $(M_A,T_A)$-dirty flap, then we remove this edge from $\tilde{M}_{w}.$%\sed{I prefer it is beter to discard here!}
	Note that the induced edges may increase the degree of some, say $w,$ of the degree-2 ground vertices in $V(M_A)\setminus T_A$ because of some edge $e=\{v_F,w\}$ between
	a flap-vertex $v_F$ of ${W}_{\frak{R}}$ and $w$; we then remove from $\tilde{M}_{\rm in}$ all such edges
	(in \autoref{label_lestrygonians},  these edges are the black dotted edges). Actually this last modification could be avoided, however it facilitates  the presentation of the last
	part of the proof.

	Notice that if a flap $F$ is not $(M_A,T_A)$-dirty and contains an edge $e$ of $M,$ then this edge should belong
	to  a subpath $P$ of a subdivision path of $M$ such that the  endpoints of $P$ are two of the vertices of $\partial F.$
	% and, moreover, these endpoints do not belong to  $\tilde{M}_{\rm in}.$
	We call such flap-vertices of $\tilde{M}_{\rm in}$ {\em subdivision flap-vertices} of $\tilde{M}_{\rm in}.$
	If a flap-vertex of $\tilde{M}_{\rm in}$ is not a  subdivision flap-vertex, then we call it   {\em branch  flap-vertex} of $\tilde{M}_{\rm in}.$ We denote by $Q$ the set of all the branch  flap-vertices of $\tilde{M}_{\rm in},$ and as each  branch flap-vertex corresponds to a $(M_A,T_A)$-dirty flap, we have that  $|Q|\leq {a+2\ell}.$
	Notice that all vertices of $Q$ have degree at most three in ${W}_{\frak{R}}$ (and therefore in $\tilde{M}_{\rm in}$ as well).

	%This means that if in $G$ we replace ${\bf F}^{A}$ by ${\bf F}'^A,$ then $H$ should also be a minor of the resulting graph.

	%\medskip

	The ground-vertices of $\tilde{M}_{\rm in}$ that are apex-jump vertices or belong to $T_{\rm in} \setminus A$ are called {\em branch ground-vertices} of $\tilde{M}_{\rm in},$
	while the rest of the ground vertices  of $\tilde{M}_{\rm in}$ are called  {\em subdivision ground-vertices} of $\tilde{M}_{\rm in}.$ Notice that there are at most ${a+2\ell}$ branch ground-vertices in  $\tilde{M}_{\rm in}.$
	% which, from now on, denote by $Q.$
	Let   $M_{\rm out}=M_A\cap G_{\rm out}$ and $\tilde{M}=M_{\rm out}\cup \tilde{M}_{\rm in}.$
	%	(in~\autoref{label_disadvantages}, $M_{\rm out}$ consists of the green edges -- however only the edges in $T_{\rm out}$ are depicted in this figure).
	We  define
	$\tilde{T}_{\rm in}$ as the union  of the set of branch ground-vertices of $\tilde{M}_{\rm in}$ and the
	set $Q$ of the branch  flap-vertices of  $\tilde{M}_{\rm in}.$ Observe that  $|\tilde{T}_{\rm in}|\leq {2a +4\ell}.$ Moreover,   we set $T_{\rm out}=T_A\cap V(M_{\rm out})$
	and  $\tilde{T}=T_{\rm out}\cup \tilde{T}_{\rm in}.$ %\ig{should we try to unify ``$=$'' and ``$:=$''?}

	Notice that
	$(\tilde{M},\tilde{T})$ is a ${\sf tm}$-pair   of $\tilde{G}$; we refer to  it
	%\sed{Is this verbal definition necessary?}
	as the
		{\em leveling of the ${\sf tm}$-pair  $(M_A,T_A)$ of $G$ with respect to $(A,W,\frak{R})$}.
	Clearly, $\tilde{M}$ is a partially $\Delta$-embedded graph whose compass is $\tilde{M}_{\rm in},$ i.e., $\tilde{M}_{\rm in}=\Delta\cap \tilde{M}.$ Also, as $|\tilde{T}_{\rm in}|\leq {2a +4\ell},$ the $\Delta\cap V(G)$-intrusion of $(\tilde{M},\tilde{T})$ in $\tilde{G}$ can be bounded, using the fact that $\tilde{M}_{\rm in}$ is a subgraph of the planar graph  $\tilde{G}_{\rm in},$  by {$3(2a +4\ell) = \tilde{\ell}$}.

	\medskip

	Let us now give some intuition on the definition of $(\tilde{M},\tilde{T}).$
	We see $\tilde{G}$ and its {\sf tm}-pair $(\tilde{M},\tilde{T})$ as a
	``projection'' of the graph $G$ and its  {\sf tm}-pair $({M},{T}),$ respectively,
	in what concerns the leveling ${W}_{\frak{R}}$ of $W.$ This projection
	is loosing some of the information of  $(\tilde{M},\tilde{T}),$ however it will be
	valuable as now both $\tilde{G}$ and  $\tilde{M}$ are partially $\Delta$-embedded graphs.
	The lost information is encoded, for every branch flap-vertex $v_F$ in $Q,$
	by ${\ell}\mbox{\sf -folio}({\bf F}^{A}).$ Recall that ${\bf F}^{A}:=(G[A\cup F],A\cup \partial F,\rho_A\cup \rho_F).$
	%the equivalence class {\sf $(\ell,a)$-color}$(F),$ that is, the equivalence class of $\sim_{A,\ell}$ to which ${\bf F}^{A}$ belongs (recall that ${\bf F}^{A}:=(G[A\cup F],A\cup \partial F,\rho_A\cup \rho_F)$ and $F$ is the underlying graph of ${\bf F}^{A}$).\sed{$\sim$ now died!}
	In what follows, we
	will use \autoref{label_constitutivos} in order to  draw in $\tilde{G}$ a modification $(\hat{M},\hat{T})$ of   $(\tilde{M},\tilde{T})$ so that $(\hat{M},\hat{T})$ does not go ``too deeply'' in the representation $R_W$ of $W$ in $W_{\frak{R}},$ and the parts of  $(\hat{M},\hat{T})$
	that are different from $(\tilde{M},\tilde{T})$  are routed through $R_W$  in a ``dispersed enough'' way. Moreover, again from \autoref{label_constitutivos},
	$(\tilde{M},\tilde{T})$ can be seen as a contraction of $(\hat{M},\hat{T})$
	that does not identify any of the vertices of $Q.$
	Maintaining the vertices in
	$Q$ intact, while performing contractions in  $(\hat{M},\hat{T}),$ is important
	as we need to keep the information given by $\hat{\ell}\mbox{\sf -folio}({\bf F}^{A})$ for each branch flap-vertex $v_F\in Q.$
	%for some  suitable choice of $\hat{\ell}$ that we explain below.

	\medskip

	For each branch  flap-vertex $v_F$ of $Q$ we define ${\bf M}_{F}=(G[A\cup V(F)]\cap  M_A,A\cup \partial F,\rho_{A}\cup\rho_{F}),$ where the functions $\rho_{A}$ and $\rho_{F}$ are defined as explained in the beginning of \autoref{label_rencontroient},  and
	$T_{F}=A\cup \partial F\cup (T_A\cap F).$ Notice that  $({\bf M}_{F},T_{F})$ is a {\sf btm}-pair of ${\bf F}^A$ such
	{that ${\sf diss}({\bf M}_{F},T_{F})\in {\hat{\ell}}\mbox{-}{\sf folio}({\bf F}^{A}),$ given that $\hat{\ell} =  a +3 +\ell.$} This means that if $v_{\bar{F}}$
	is another  flap-vertex  of $W_{\frak{R}}$ with ${\hat{\ell}}\mbox{\sf -folio}({\bf F}^{A}) = {\hat{\ell}}\mbox{\sf -folio}({\bf \bar{F}}^{A}),$
	then there is a  {\sf btm}-pair $({\bf M}_{\bar{F}},T_{\bar{F}})$ of ${\bf \bar{F}}^A$  such that ${\sf diss}({\bf M}_{F},T_{F})={\sf diss}({\bf M}_{\bar{F}},T_{\bar{F}}).$

	\medskip

	Recall  that  ${\tilde{G}}$ is a partially $\Delta$-embedded
	graph whose compass is $\tilde{G}\cap {\Delta} = {W}_{\frak{R}}$ and that  $\tilde{M}$ can be seen as  a  partially ${\Delta}$-embedded graph whose compass is $ \tilde{M} \cap {\Delta}.$ Recall also that the $\Delta\cap V(G)$-intrusion of $(\tilde{M},\tilde{T})$  in $\tilde{G}$ is at most {$\tilde{\ell}$}.

	Let $C_{1},\ldots,C_{\tilde{r}}$ be the first $\tilde{r}$ layers of $R_W$ and  let $[D_{1},\ldots,D_{\tilde{r}}]$
	(resp. $[\overline{D}_{1},\ldots,\overline{D}_{\tilde{r}}]$)
	be the sequences of the corresponding open (resp. closed) disks.
	We can now apply~\autoref{label_constitutivos} with input ${6},$ $\tilde{\ell},$  $z,$ $\tilde{G},$ the  $r$-wall $R_W,$ the {\sf tm}-pair $(\tilde{M},\tilde{T})$ of $\tilde{G},$ and the set $Q$ of branch flap-vertices defined above. Let $d=\funref{label_quitablement}({6},\tilde{\ell}).$
	%\sed{Check in all this section the consistency with the fact that, in general, we consider the disks to be open while we use $\overline{D}_{i}$ for their closures.}
	\autoref{label_constitutivos} guarantees the existence of
	a   {\sf tm}-pair $(\hat{M},\hat{T})$  of $\tilde{G}$  and a $b\in [\tilde{r}]$ such that

	\begin{enumerate}
		\item    $\hat{M}\setminus D_{b}$ is a subgraph of $\tilde{M}\setminus D_{b},$
		\item ${\sf ann}({\cal C}_{b,b+d-1})\cap (\tilde{T} \cup\hat{T})=\emptyset,$
		\item   $(\hat{M},\hat{T})\Cap \overline{D}_{b+d}$ is a {\sf tm}-pair of $R_W$ that is  safely ${6}$-dispersed in $R_W$
		      and none of the vertices of $\hat{T}\cap \overline{D}_{b+d}$ is within face-distance less than {six}  in $R_W$
		      from some vertex of $C_{b+d}\cup C_{\tilde{r}},$
		      %\item $\hat{M}\cap \inter(D_{\funref{label_tourterelles}(c,\ell)})=\emptyset,$
		\item $\hat{M}\cap D_{\tilde{r}}=\emptyset,$
		      %\item $\hat{M}$ does not intersect the vertices of the compass of the central $z$-subwall of $W.$
		\item the compass of the central $z$-subwall of $W$ is a subset of $D_{\tilde{r}},$ and
		\item there is a $Q$-respecting contraction-mapping $\phi$ of ${\sf diss}(\tilde{M},\tilde{T})$ to ${\sf diss}(\hat{M},\hat{T}).$

	\end{enumerate}

	Our next step is to modify the  {\sf tm}-pair $(\hat{M},\hat{T})$ to obtain another
		{\sf tm}-pair $(\hat{M}^+,\hat{T}^+)$ so that, in addition to Property~1 and Property~4 above (where $\hat{M}$ is replaced by $\hat{M}^+$), $(\hat{M}^+,\hat{T}^+)$ satisfies the
	following stronger version of Property~6:

	\begin{enumerate}
		\item[$6^+.$] there is a $Q$-respecting contraction-mapping $\phi^+$ of ${\sf diss}(\tilde{M},\tilde{T})$ to ${\sf diss}(\hat{M}^+,\hat{T}^+)$ such that for every branch  flap-vertex $v_F\in Q,$  if $\phi^+(v_F)=\{v_{\hat{F}}\},$ then
		      ${\hat{\ell}}\mbox{\sf -folio}({\bf F}^{A}) = {\hat{\ell}}\mbox{\sf -folio}({\bf \hat{F}}^{A}).$
	\end{enumerate}
	%
	%\newpage

	To force Property~$6^+,$ we will make strong use of the fact that
	the flatness pair $(W,\frak{R})$ is  $\hat{\ell}$-homogeneous  with respect to  $(G,A).$ This will permit us
	to  modify the subdivision paths of $(\hat{M},\hat{T})$ by pointing to images of flaps in $Q$ that have been displaced after the application of~\autoref{label_constitutivos}.\medskip

	We
	proceed as follows: let $v_F \in Q$ be  some flap-vertex, and let $\phi(v_F)=\{w\}$ be its image given by the function $\phi$ guaranteed by Property~6.
	%and let $R$ be a brick of $W'$ whose closed disk contains $F.$
	%Clearly, $\ell\mbox{-}{\sf color}(F)\in (A,\ell)\mbox{-}{\sf palette}(R).$
	From Property~2,  $w$ cannot belong to ${\sf ann}({\cal C}_{b,b+d-1}).$
	Also, from Property~4,  $w$ cannot belong to $D_{\tilde{r}}.$
	If   $w\in {\sf ann}({\cal C}_{1,b}),$ it follows from
	Property~1 that  $w={v_F}$ and in this case Property $6^+$ already holds trivially.
	The only remaining case is when  $w\in {\sf ann}({\cal C}_{b+d,\tilde{r}}).$
	In this case, from Property~3, $w$ is a vertex of $R_W\cap {\sf ann}({\cal C}_{b+d,\tilde{r}})$
	that is within face-distance at least {six} from both $C_{b+d}$ and $C_{\tilde{r}}.$

	\begin{figure}[h!]
		\begin{center}%\vspace{-.1cm}
			\includegraphics[width=.6\textwidth]{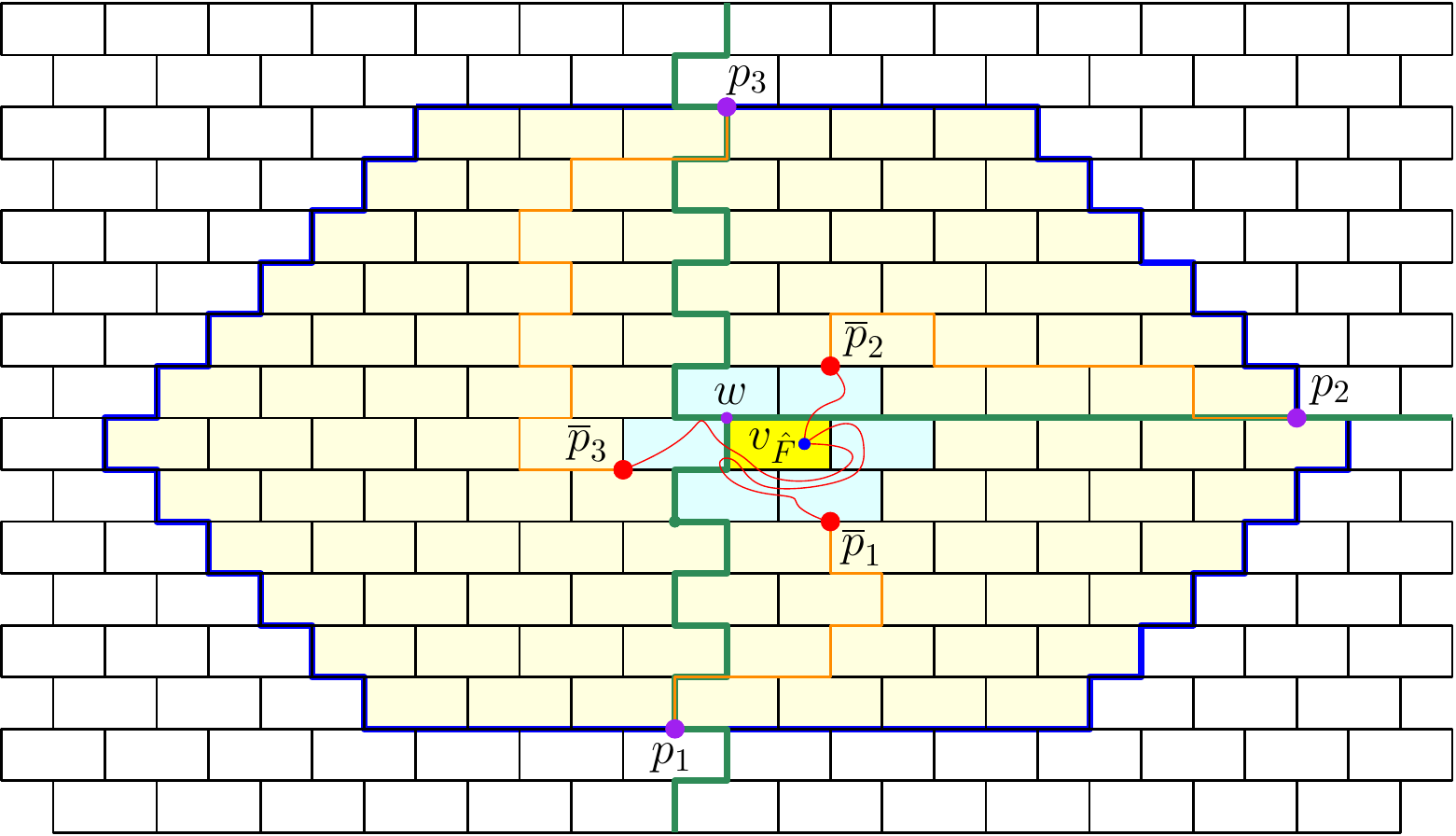}
		\end{center}\vspace{-.25cm}
		\caption{Some part of $R_W\cap {\sf ann}({\cal C}_{b+d,\tilde{r}})$ around a vertex $w$ of $\hat{T}.$ The interior of the brick $\hat{R}$ is depicted in yellow. The cycle $C_w$   is depicted in blue. }
		%\ig{replace $F'$ by $v_{F'}$ in the figure}}
		%\ig{put the yellow and green paths thicker} \ig{face-distance is 6!}}
		%\ig{is the central part of \autoref{label_disposiciones} visible enough?}}
		\label{label_disposiciones}
		\vspace{.35cm}
	\end{figure}

	Let $R$ and $\hat{R}$ be bricks of $R_W$ such that $v_F \in {\sf vflaps}_{R_W}(R)$ and $w \in {\sf vflaps}_{R_W}(\hat{R}),$ respectively. Certainly, both $R$ and $\hat{R}$ are internal bricks of $R_W.$ Moreover, $w$ is a vertex of the cycle  $\hat{R}$ (bounding the disk depicted in bright yellow in \autoref{label_disposiciones}).
	By  \autoref{label_comprometida}, there exists a vertex $v_{\hat{F}}$ in $\hat{R}$ such that  ${\hat{\ell}}\mbox{\sf -folio}({\bf F}^{A}) = {\hat{\ell}}\mbox{\sf -folio}({\bf \hat{F}}^{A}).$
	%As both $R$ and $\hat{R}$ are bricks of ${W}$  and  ${W}$ is an
	%	$(\bar{\ell},a)$-homogeneous wall,  it follows that  ${\sf palette}(R)={\sf palette}({R}'),$ therefore
	%there is a branch  flap-vertex $F'$  in the closed disk bounded by $\hat{R}’$
	%	such that $(A,\ell)\mbox{\sf-color}(F)=(A,\ell)\mbox{\sf-color}(F').$
	Let $y$ be the degree of $w$ in $\hat{M}$ (which equals the degree of $v_F$ in $\tilde{M}$) and let  $P_{1},\ldots,P_{y}$ be the subdivided paths
	of the model $(\hat{M},\hat{T})$ that have $w$ as an endpoint, depicted as fat green lines in \autoref{label_disposiciones}. Note than $y\leq 3.$

	Let $C_{w}$ be the cycle of $R_{W}$ induced by the vertices that are within face-distance exactly {six}  from $w$ (in \autoref{label_disposiciones}, this cycle is depicted with thick blue edges).
	From Property 3, $C_w$ is within face-distance at least $(2\cdot {6} + 1)- {6} = {7}$ from any other vertex in $\hat{T}\cap \overline{D}_{b+d}$ in $R_{W}.$ For $i \in [y],$ we define $P_{i}'$ as the unique path, among the paths obtained from $P_{i}$ by the removal the open disk bounded by $C_{w},$ that has only one endpoint in the cycle $C_w.$
	%Notice now that if we remove from    $P_{1}\cup \cdots\cup P_{y}$
	%	the open disk bounded by $C_{w},$ then we obtain $y$ paths $P_{1}',\ldots,P_{y}'.$ \ig{only the end of the paths!!}
	For $i\in[y],$ we denote by $p_{i}$ the endpoint of $P_{i}'$ that is not an endpoint of $P_{i}$ (depicted in purple  in \autoref{label_disposiciones}). Notice that $p_{1},\ldots,p_{y}$
	are internal pegs of $C$ ordered as they appear in the perimeter in counter-clockwise order.

	Consider now the brick-neighborhood, say $X,$ of $\hat{R}.$
	Notice that the face-distance between the perimeter of $X$ and $C_{w}$ is at least {four}.
	From~\autoref{label_ofrecimiento}, there are $y$ internally vertex-disjoint paths $\bar{P}_{1},\ldots,\bar{P}_{y}$ from $v_{\hat{F}}$ to the external pegs of the perimeter $C_{X}$ of $X$ (depicted in red lines in \autoref{label_disposiciones}). Let these pegs be $\bar{p}_{1},\ldots,\bar{p}_{y}$ ordered as they appear in the perimeter in counter-clockwise order.
	We are now in position to apply \autoref{label_deliberadament} to $C_{w},$ $C_{X},$ $p_{1},\ldots,p_{y},$ and $\bar{p}_{1},\ldots,\bar{p}_{y}$ and
	find $y$ pairwise vertex-disjoint paths $\hat{P}_{1},\ldots,\hat{P}_{y}$ each joining
	$p_{i}$ with $\bar{p}_{i}$  (depicted in orange lines in \autoref{label_disposiciones}). We now remove from $\hat{M}$ the open disk bounded  by $C_{w}$
	and we add the graph $\bar{P}_{1}\cup\cdots\cup \bar{P}_{y}\cup \hat{P}_{1}\cup\cdots\cup \hat{P}_{y}.$ We also modify $\hat{T}$ by substituting $w$ by $v_{\hat{F}}.$ By repeating this procedure for every flap-vertex of $Q,$ and defining $(\hat{M}^+,\hat{T}^+)$ as the updated {\sf tm}-pair obtained when this algorithm terminates, we enforce Property~$6^+.$ Namely the function $\phi^+$ is defined from modifying $\phi$ so that, whenever $\phi(v_F) \neq \{v_F\},$ we set $\phi^+(v_F) = \{v_{\hat{F}}\}$ as above.
	Note that $(\hat{M}^+,\hat{T}^+)$ indeed satisfies, besides Property~$6^+,$ Properties~1 and~4.

	\medskip

	Our next (and last) step is to further  modify the {\sf tm}-pair $(\hat{M}^+,\hat{T}^+)$ of $\tilde{G}$ so to obtain a {\sf tm}-pair $(\overline{M}^{},\overline{T}^{})$ of the original graph $G$ so that ${\sf diss}({M},{T})$ is a minor of ${\sf diss}(\overline{M}^{},\overline{T}^{}).$
	Notice that each edge of $\hat{M}^+$
	that is an edge of $R_W$ (resp. ${W}_{\frak{R}}$) has one endpoint that is a ground-vertex of $R_W$ (resp. ${W}_{\frak{R}}$) and another one that is a flap-vertex of $R_W$ (resp. ${W}_{\frak{R}}$).

	Let $v_{\hat{F}}$ be a flap-vertex of $\hat{M}^+.$ We modify $(\hat{T}^+,\hat{M}^+)$ by distinguishing the following cases:

	\begin{itemize}
		\item $v_{\hat{F}}\not\in \hat{T}^+.$ Then the neighbors of $v_{\hat{F}}$ in $R_W$
		      are two ground-vertices $g$ and $g'.$
		      %, and at most one of them belongs to $\hat{T}^+.$
		      From tightness property~\ref{label_mantenimientos} of a rendition, there is a path  $P_{\hat{F}}$ in $\hat{F}$ with $g$ and $g'$ as
		      endpoints. We substitute the edges  $\{g,\hat{F}\}$ and $\{\hat{F},g'\}$ of $R_W$ with $P_{\hat{F}}.$ Notice that  $P_{\hat{F}}$ is a path of ${\sf compass}_{\frak{R}}(W).$

		\item $v_{\hat{F}}\in \hat{T}^+$ and $v_{\hat{F}} \in \phi^+(x)$ for some $x\not\in Q.$ In this case, from tightness property~\ref{label_participates} of a rendition, $\hat{F}$ has $y\in[3]$ neighbors $v_{1},\ldots,v_{y},$
		      and a vertex $z_{\hat{F}}$ that is connected to $v_{1},\ldots,v_{y}$ via $y$ internally vertex-disjoint paths.  We substitute the edges $\{z_{\hat{F}},v_{1}\},\ldots,\{z_{\hat{F}},v_{y}\}$ and the vertex $z_{\hat{F}}$ of $R_W$ with the union of these $y$ internally vertex-disjoint paths. Notice that these paths are also paths of ${\sf compass}_{\frak{R}}(W).$ We also update $\hat{T}^+:=\hat{T}^+\setminus\{v_{\hat{F}}\}\cup\{z_{\hat{F}}\}$ and we call the vertex $z_{\hat{F}}$ {\em replacement} of $v_{\hat{F}}.$

		\item $v_{\hat{F}}\in \hat{T}^+$ and $\{v_{\hat{F}}\}=\phi^+(v_{F})$ for some $v_{F}\in Q.$ Let ${\bf M}_{F}=({F}\cap M,\partial F,\rho_{A}\cup\rho_{F})$ and let $T_F = T \cap F.$
		      From Property~$6^+,$ ${\hat{\ell}}\mbox{\sf -folio}({\bf \hat{F}}^{A}) = {\hat{\ell}}\mbox{\sf -folio}({\bf F}^{A}),$ which implies  that  there is a {\em {\sf btm}-pair} $({\bf M}_{\hat{F}},{T}_{\hat{F}})$ of ${\bf \hat{F}}^A$
		      % that is a topological minor of ${\bf F}'^A$ and
		      such that ${\sf diss}({\bf M}_{\hat{F}},{T}_{\hat{F}})={\sf diss}({\bf M}_{F},T_{F}).$ We denote by $M_{\hat{F}}$ the underlying graph of
		      ${\bf M}_{\hat{F}}$ and we substitute the vertex $v_{\hat{F}}$ and its incident edges in $W_{\frak{R}}$
		      with the graph $M_{\hat{F}}$ (that is a subgraph of ${\sf compass}_{\frak{R}}(W)$), by identifying the boundary vertices according to the functions $\rho_{A}\cup\rho_{F}.$ We also update $\hat{T}^+:=\hat{T}^+\setminus\{v_{\hat{F}}\}\cup {T}_{\hat{F}}.$
	\end{itemize}

	The above operations create a {\sf tm}-pair of $G$ that  we denote henceforth by  $(\overline{M},\overline{T}).$ Since $(\overline{M},\overline{T})$ satisfies Property~4 and $\tilde{r}\geq z,$ $\overline{M} \cap \cupall{\sf influence}_{{\frak{R}}}(W') = \emptyset,$ and therefore $(\overline{M},\overline{T})$
	is a {\sf tm}-pair of $G\setminus V({\sf compass}_{\frak{R}}(\tilde{W}')).$ It is worth mentioning that Property~1 implies the strong property that, throughout the rerouting procedure, the part of the topological minor model outside of the  $\frak{R}$-compass of $W$ in $G \setminus A$ can only be reduced; formally, $\overline{M} \cap X \subseteq M \cap X$ (see~\autoref{label_enthusiastically}  after the end of  this proof).

	\begin{figure}[h]
		\begin{center}
			\includegraphics[width=.33\textwidth]{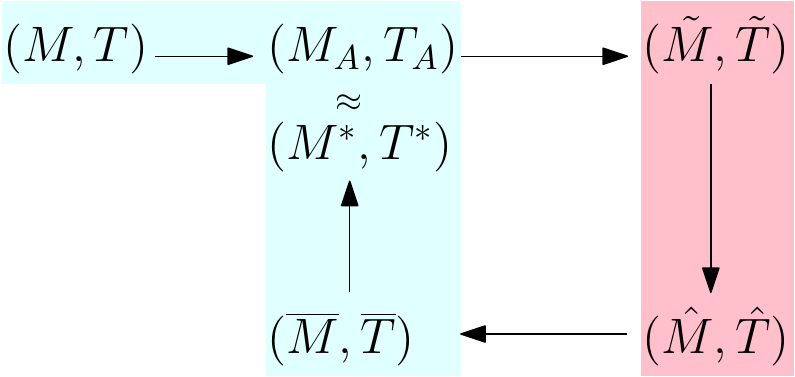}
		\end{center}
		\caption{The ${\sf tm}$-pairs  considered in the proof of \autoref{label_pretendientes}.
			The cyan-shadowed pairs are {\sf tm}-pairs of $G$ and the pink-shadowed pairs are {\sf tm}-pairs of  $\tilde{G}.$}
		\label{label_irresponsibility}
	\end{figure}

	As  $H$ is a minor of $H'={\sf diss}(M,T),$ hence a minor of ${\sf diss}(M_A,T_A)$ as well, and $G\setminus V({\sf compass}_{\frak{R}}(\tilde{W}'))={\bf K}\oplus (Z\setminus V({\sf compass}_{\frak{R}}(\tilde{W}')), B,\rho),$ it remains to verify  that ${\sf diss}(\overline{M},\overline{T})$ contains ${\sf diss}(M_A,T_A)$ as a minor.
	For this, consider every  $x\in \tilde{T}\setminus Q$ and
	construct the set $T_{x}$ by taking $T_{x}:=\phi^+(x)$ and substituting each branch flap-vertex $v_{\hat{F}}\in T_{x}$ with its replacement $z_{\hat{F}}$ defined as in the second case of the above case analysis. Notice  that each set $T_{x}$ is a subset of $\overline{T}.$
	We now construct $({M}^{\star},{T}^{\star})$ as follows: $\tilde{M}^{\star}$ is obtained by contracting, for each $x\in \tilde{T}\setminus Q,$ all subdivision paths in $\overline{M}$ that have as endpoints two
	vertices in $T_{x}$ to a single vertex,  which we again call $x.$ This operation identifies, for each $x\in \tilde{T}\setminus Q,$ all vertices of $T_{x}$ into $x,$ thus
	we set ${T}^{\star}=(\overline{T}\setminus (\bigcup_{x\in \tilde{T}\setminus Q}T_{x}))\cup (\tilde{T}\setminus Q).$ That way $({M}^{\star},{T}^{\star})$ can be seen as
	the {\sf tm}-pair $(\tilde{M},\tilde{T})$ where each of its branch flap-vertices
	has been substituted  as in the third case of the above case analysis.
	This permits us to verify that ${\sf diss}({M}_A,{T}_A)$ and ${\sf diss}({M}^{\star},{T}^{\star})$ are isomorphic (see \autoref{label_irresponsibility} for the relation between the models that we have defined). As  ${\sf diss}({M}^{\star},{T}^{\star})$  is a minor of   ${\sf diss}(\overline{M},\overline{T}),$  the theorem follows.
\end{proof}

%\ig{OBSERVATION: in the rerouting, the part outside the compass of $W$ does NOT change}

\begin{remark}\label{label_enthusiastically}
	In the above proof,  the rerouted topological minor model $(\overline{M},\overline{T})$ obtained from the original model $(M,T)$ satisfies the following property: the part of the topological minor model outside of the $\frak{R}$-compass of $W$ in $G \setminus A$ can only be reduced with respect to the original one; formally, $\overline{M} \cap X \subseteq M \cap X,$ where $X$ is the ``external'' set corresponding to the separation of the considered apex-wall triple $(A,W,\frak{R}).$
\end{remark}

By definition of the set ${\cal R}_{h}^{(t)},$ its elements are of minimum size, and therefore a boundaried graph ${\bf G}=(G,B,
	\rho)\in {\cal R}_{h}^{(t)}$
does not contain any $3h$-irrelevant vertex. To see this,
recall that in~\autoref{label_unquenchable} the equivalence is defined in terms of graphs $H$ with detail at most $h$ (i.e., with at most $h$ vertices and  at most  $h$ edges), and that by \autoref{label_daskalojannes}  every graph in ${\sf ext}(H)$ has detail at  most $3h.$ On the other hand, \autoref{label_tranquillity} implies that, in the setting of \autoref{label_pretendientes}, there is a vertex belonging to the compass of every $W'$-tilt of $(W,\frak{R}),$ where $W'$ is the
central $z$-subwall of $W.$ Thus, from~\autoref{label_pretendientes} for the particular case $z=3$ and $\ell=3h,$ $B$
should affect every $(a,\funref{label_encompassing}(a,3,3h),\funref{label_intervention}(a,3h))$-apex-wall triple  of $G,$ for every value of $a.$ As $|B| = t,$ we conclude the following.
\begin{corollary}%[{\sf BOUNDARY AFFECTS  ALL WALLS}]
	\label{label_dispensaries}
	If $t,h,a\in \Bbb{N}$  and ${\bf G}=(G,B,
		\rho)$ is a  boundaried graph in ${\cal R}_{h}^{(t)},$ then  $B$ affects every $(a,\funref{label_encompassing}(a,3,3h),\funref{label_intervention}(a,3h))$-apex-wall triple  of $G,$ in particular, ${\bf p}_{a,\funref{label_encompassing}(a,3,3h),\funref{label_intervention}(a,3h)}(G)\leq t.$
\end{corollary}

%
%\end{document}
\section{Bounding the size of the representatives}
\label{label_pertenecerles}

In this section we use the results obtained in the previous sections to prove that every representative in ${\cal R}_{h}^{(t)}$ has size linear in $t.$ For this, we first prove in \autoref{label_intolerantly} that every representative in ${\cal R}_{h}^{(t)}$ has a set of $\Ocal_h(t)$ vertices containing its boundary whose removal leaves a graph with treewidth bounded by a constant depending only on the collection $\Fcal$; such a set is called a \emph{treewidth modulator}.

Once we have the treewidth modulator, we can use known results from the protrusion machinery to achieve our goal. Namely, in \autoref{label_sechseckigen} we show how to obtain a linear protrusion decomposition of a representative, and we reduce each of the linearly many protrusions in the decomposition to an equivalent protrusion of constant size. Once we have this, a dynamic programming algorithm similar to that of~\cite{BasteST20-monster1} yields \autoref{label_desgobernada}.

%Finally,
%in \autoref{label_reconstituirse} we give upper bounds on the constants depending on the collection $\Fcal$ involved in our algorithm. These upper bounds depend explicitly on the
%parametric dependencies of the Unique Linkage Theorem~\cite{KawarabayashiW2010asho,RobertsonSGM22}.

\subsection{Finding a treewidth modulator of linear size}
\label{label_intolerantly}

%\paragraph{Balanced separations.} %Let $G$ be  a graph. We say that a pair $(L,R)\in 2^{V(G)}\times 2^{V(G)}$ is a {\em separation} of $G$ if $L \cup R = V(G)$ and there is no edge between $L\setminus R$ and $R\setminus L.$
Given a graph $G$ and a set $S\subseteq V(G),$  we say that a separation  $(L,R)$ of $G$ is a {\em $2/3$-balanced separation} of $S$ in $G$  if $|(L\setminus R)\cap S|, |(R\setminus L)\cap S|\leq \frac{2}{3}|S|.$
%Balanced separators have been extensively studied in the context of graph algorithms (see e.g.~\cite{LiptonT79,AST90,LiptonT80,AlonST94}).
We need the following well-known property of graphs of bounded treewidth (see e.g.  \cite{Bodlaender98,CyganFKLMPPS15}).

\begin{lemma}
	\label{label_instituciones}
	Let $G$ be a graph and let $S\subseteq V(G).$ There is a $2/3$-balanced separation $(L,R)$ of $S$ in $G$ of order at most $\tw(G)+1.$
\end{lemma}

Using \autoref{label_substantiality}, \autoref{label_concerniente}, \autoref{label_dispensaries}, and \autoref{label_instituciones} we prove the following result, whose proof uses Akra-Bazzi Theorem~\cite{AkraB98}, in particular its extended  provided by Leighton~\cite{LeightonAB}. We stress that ${\bf p}$ is not a bidimensional parameter in the precise way that is defined in~\cite{F.V.Fomin:2010oq,DemaineFHT05sube,FominDHT16}, therefore
\autoref{label_murmuradores} cannot be derived by directly
applying the results of \cite{F.V.Fomin:2010oq}.

%\ig{REMOVE $q$ (almost) everywhere}

%\ig{Dimitrios, in the theorem below, do we really need to distinguish between $h$ and $q$?? In fact, by the way that the equivalence relation $\equiv_{h}$ is defined, we have that $h=q$ here!!!}

\begin{lemma}%[{\sf MODULATOR$\star$}]
	\label{label_murmuradores}
	There exist two functions $\newfun{label_eavesdropper},\newfun{label_confrontation}: \Bbb{N}^2\to\Bbb{N}$ such that
	if $t,q,h\in \Bbb{N}$  and ${\bf G}=(G,B,\rho)$ is
	a $K_{q}$-minor-free boundaried graph  in ${\cal R}_{h}^{(t)},$ then
	$G$ contains an $\funref{label_eavesdropper}(q,h)$-treewidth modulator that contains $B$
	and has  at most $\funref{label_confrontation}(q,h) \cdot t$ vertices.
	Moreover, it holds that $\funref{label_eavesdropper}(q,h)=\Ocal((\funref{label_superintendent}(q,\funref{label_encompassing}(\funref{label_scheinbeziehungen}(q),3,3h),\funref{label_intervention}(\funref{label_scheinbeziehungen}(q),3h)))^2).$
\end{lemma}

%$\funref{label_scheinbeziehungen}(q),\funref{label_encompassing}(\funref{label_scheinbeziehungen}(q),3,3h),\funref{label_intervention}(\funref{label_scheinbeziehungen}(q),3h)$

\begin{proof}
	Let $q,h\in \Bbb{N}.$
	We use ${\sf q}$ as a shortcut for  the triple $(\funref{label_scheinbeziehungen}(q),\funref{label_encompassing}(\funref{label_scheinbeziehungen}(q),3,3h),\funref{label_intervention}(\funref{label_scheinbeziehungen}(q),3h))$
	and we set $s=\funref{label_superintendent}(q,\funref{label_encompassing}(\funref{label_scheinbeziehungen}(q),3,3h),\funref{label_intervention}(\funref{label_scheinbeziehungen}(q),3h)).$
	%For every $t\in \Bbb{N},$ we define
	%Let ${\cal G}$ be the class of $K_{q}$-minor free boundaried graphs  in ${\cal R}_{h}^{(t)}.$
	Let $t_0=\max\{\min\{t'\mid s\cdot  \sqrt{t'}+s+1\leq t'/\log^{2}t'\},42534179953\}$ and let $x=s\cdot  \sqrt{t_0}+s.$ We
	%42534179952.477947235107421875 is also OK!
	%\sed{\tiny\hspace{-10mm} 42534179952.477947235107421875 is also OK!}
	define the  function ${\bf z}:\Bbb{N}\to\Bbb{N}$ where
	\begin{eqnarray}
		{\bf z}(t)\!\! &\!\!\!\! =\!\!\!\! &\!\! \min \{z\mid  \forall G\ \forall B\subseteq V(G) \mbox{~if $G$ is $K_{q}$-minor-free, $|B|\leq t,$ and $B$ affects every ${\sf q}$-apex-wall triple of $G,$~}\nonumber\\
		& & ~~~~~~~~~~~~~~~~~~~~~~~~~~~~~~~~~~~~~~~~~~~~~~~~~~~~~~~~~~~~~~~\mbox{then }\exists Z \subseteq V(G):  |Z|\leq z, B\subseteq Z, \tw(G\setminus Z)\leq x\}\nonumber.
	\end{eqnarray}
	Let $G$ be a $K_{q}$-minor-free graph and let $B\subseteq V(G)$ such
	$|B|\leq t$ and $B$ affects every ${\sf q}$-apex-wall triple of $G.$
	%We make the following observations.
	%\begin{enumerate}
	From \autoref{label_substantiality}, $\tw(G)\leq s\cdot  \max\{1,\sqrt{|B|}\}\leq s\cdot  \max\{1,\sqrt{t}\}\leq s\cdot  \sqrt{t}+s.$
	From~\autoref{label_instituciones}, $G$ has a 2/3-balanced separator $(L,R)$ of $B$ where $|L\cap  R|\leq  s\cdot  \sqrt{t}+s
		+1.$ This means that $|(L\setminus R)\cap B|, |(R\setminus L)\cap B|\leq \frac{2}{3}|B|.$
	We set $G_{L}=G[L],$ $G_{R}=G[R],$
	$B_{L}=L\cap (R \cup B) \subseteq V(G_{L}),$ and
	$B_{R}=R\cap (L \cup B)\subseteq V(G_{R})$
	and observe that $B\subseteq B_L\cup B_{R}$ and that both $G_{L}$ and $G_{R}$ are $K_{q}$-minor-free.
	Notice that $B_L=\big((L\setminus R)\cap B\big)\cup (L\cap R)$ and $B_R=\big((R\setminus L)\cap B\big)\cup (R\cap L),$ therefore  there is some $α\in [\frac{1}{2},\frac{2}{3}],$ such that
	$|B_L|\leq α\cdot t+s\cdot  \sqrt{t}+s+1$ and $|B_{R}|\leq (1-α)\cdot t+s\cdot  \sqrt{t}+s+1.$
	From~\autoref{label_concerniente},
	$B_{L}$ affects every ${\sf q}$-apex-wall triple of $G_{L}$
	and $B_{R}$ affects every ${\sf q}$-apex-wall triple of $G_{R}.$ This means that there exists some $Z_{L} \subseteq V(G_L)$ such that
	$|Z_{L}|\leq {\bf z}(α\cdot t+s\cdot  \sqrt{t}+s+1),$
	$B_{L}\subseteq Z_{L},$ and $\tw(G_{L}\setminus Z_{L})\leq x.$ Also,  there exists some $Z_{R}  \subseteq V(G_R)$ such that
	$|Z_{R}|\leq {\bf z}((1-α)\cdot t+s\cdot  \sqrt{t}+s+1),$
	$B_{R}\subseteq Z_{R},$ and $\tw(G_{R}\setminus Z_{R})\leq x.$
	We set $Z=Z_{L}\cup Z_{R}$ and observe that
	$|Z| \leq  {\bf z}(α\cdot t+s\cdot  \sqrt{t}+s+1)+{\bf z}((1-α)\cdot t+s\cdot  \sqrt{t}+s+1).$
	Moreover, $B\subseteq B_L\cup B_{R}\subseteq Z_{L}\cup Z_{R}=Z$ and, since $L \cap R \subseteq Z,$  it holds that ${\sf tw}(G\setminus Z)\leq x.$
	%\end{enumerate}
	We obtain that \begin{eqnarray}{\bf z}(t) & \leq  &{\bf z}(α\cdot t+s\cdot  \sqrt{t}+s+1)+{\bf z}((1-α)\cdot t+ s\cdot  \sqrt{t}+s+1).\label{label_accidentally}\end{eqnarray}
	Let now $f:\Bbb{N}_{\geq 1}\to\Bbb{N}$ be the solution of the following recurrence:

	$$f(t) = \begin{cases} f\big(\frac{t}{3}+s\cdot  \sqrt{t}+s+1\big)+f\big( \frac{2t}{3}+ s\cdot  \sqrt{t}+s+1\big) & \mbox{if } t> t_0            \\
			t_0                                                                                        & \mbox{if }1\leq t \leq t_{0}\end{cases}$$

	By the choice of $t_0,$ it holds that $t_0 = \Theta_{s}(1).$ Also, the choice of $t_0$ is made so that the conditions for applying the extended version of Akra-Bazzi Theorem \cite{AkraB98} provided by Leighton~\cite{LeightonAB} are satisfied\footnote{We verified these conditions using an elementary {\sf MATLAB} program, from which the number 42534179953 was generated.}. Consequently, the solution of the above recurrence is $f(t)=\Theta_{s}(t^{α})$ where
	$α$ is the unique solution of the  equation $(1/3)^α+(2/3)^α=1.$ Therefore $f(t)=Θ_{s}(t).$

	Note that, if $1 \leq t \leq t_{0},$ then from  \autoref{label_substantiality} we have that
	${\bf z}(t)\leq t \leq t_0= f(t).$ On the other hand,  by convexity, the right part of \eqref{label_accidentally} is upper-bounded by  $f(t),$ so for all $t > t_0$
	we have that ${\bf z}(t)\leq f(t).$ Summarizing, we have that ${\bf z}(t)\leq f(t)$ for all $t \geq 1.$

	Let now ${\bf G}=(G,B,\rho)$ be
	a $K_{q}$-minor-free boundaried graph  in ${\cal R}_{h}^{(t)}.$ Applying~\autoref{label_dispensaries} with $a=\funref{label_scheinbeziehungen}(q),$  we obtain that $B$ affects every ${\sf q}$-apex-wall triple  of $G.$ Therefore $G$ contains a $x$-treewidth modulator that contains $B$ and has ${\bf z}(t)\leq f(t)=\Ocal_{s}(t)$ vertices, as required. Therefore, $ {\bf z}(t) \leq \funref{label_confrontation}(q,h) \cdot t$ for some function $\funref{label_confrontation}: \Bbb{N}^2\to\Bbb{N}.$  Observe  that $t_{0}=\Ocal(s^2),$ therefore $x=\Ocal(s^2)$ as well.
	The lemma follows with $\funref{label_eavesdropper}(q,h) :=x=\Ocal(s^2).$
	% x=\Ocal(s^2)=(\funref{label_inevitability}(h))^{{2^{2^{\Ocal((q+h)\cdot \log(q+h))}}}}.
\end{proof}

Note that the above proof does not give any estimation on the function $\funref{label_confrontation}(q,h).$ In the Appendix (\autoref{label_confederarlos}) we provide an improved version of \autoref{label_murmuradores}, namely \autoref{label_gleichsinnigen}, with $\funref{label_confrontation}(q,h)=2.$
This will permit us to make a better
estimation of the contribution of $h$ in the running time of our algorithm (cf.~\autoref{label_reconstituirse}). The proof of \autoref{label_gleichsinnigen} is an adaptation to our setting of the one of  \cite[Lemma~3.6]{F.V.Fomin:2010oq}.

\subsection{Finding a linear protrusion decomposition and reducing protrusions}
\label{label_sechseckigen}

Equipped with \autoref{label_gleichsinnigen}, the next step is to construct an appropriate protrusion decomposition of a representative. We first need to define protrusions and protrusion decompositions of graphs and boundaried graphs.

\paragraph{Protrusion decompositions of unboundaried graphs.}
Given a graph $G,$ a set $X \subseteq V(G)$ is a {\em $\beta$-protrusion} of $G$ if $|\partial(X)|\leq \beta$ and $\tw(G[X]) \leq \beta-1.$
Given $\alpha,t\in\Bbb{N},$
an $(\alpha,\beta)${\em -protrusion decomposition} of $G$
is  a sequence ${\cal P}=\langle R_{0},R_{1},\ldots,R_{\ell}\rangle$ of pairwise disjoint subsets of $V(G)$
such that \begin{itemize}
	\item $\bigcup_{i\in[\ell]}=V(G),$
	\item $\max\{\ell, |R_{0}|\}\leq \alpha,$
	\item for $i\in [\ell],$ $N[R_i]$  is a $\beta$-protrusion of $G,$ and
	\item  for $i\in [\ell],$ $N(R_{i})\subseteq R_0.$
	      %SED: actually, this last condition is not necessary! But makes things  more visualizable!
\end{itemize}
We call the sets $N[R_i]$ $i\in [\ell],$ the  {\em protrusions} of ${\cal P}$ and the set $R_0$ the {\em core} of ${\cal P}.$

The above notions can be naturally generalized to boundaried graphs, just by requiring that both boundaries --of the host graph and of the protrusion-- behave as one should expect, namely that
the intersection of the protrusion with the boundary of the considered graph is a subset of the boundary of the protrusion.

\paragraph{Protrusions and protrusion decompositions of boundaried graphs.} Given a boundaried graph ${\bf G}=(G,B,\rho),$ a \emph{tree decomposition} of ${\bf G}$ is any tree decomposition of $G$ with a bag containing  $B.$
The {\em treewidth} of a boundaried graph ${\bf G},$ denoted by $\tw({\bf G}),$ is the minimum width of a  tree decomposition of ${\bf G}.$
A boundaried graph ${\bf G}'=(G',B',\rho')$  is a {\em $\beta$-protrusion} of ${\bf G}$
if
\begin{itemize}
	\item $V(G')$ is a $\beta$-protrusion of $G,$
	\item $\tw({\bf G'}) \leq \beta-1,$
	\item $\partial(V(G'))\subseteq B',$ and
	\item $B\cap V(G')\subseteq B'.$
\end{itemize}

Given a boundaried graph ${\bf G}=(G,B,\rho)$ and $\alpha,t\in\Bbb{N},$
an $(\alpha,\beta)${\em -protrusion decomposition} of ${\bf G}$
is  a sequence ${\cal P}=\langle R_{0},R_{1},\ldots,R_{\ell}\rangle$ of pairwise disjoint subsets of $V(G)$
such that \begin{itemize}
	\item $\bigcup_{i\in[\ell]}=V(G),$
	\item $\max\{\ell, |R_{0}|\}\leq \alpha,$
	\item $B \subseteq R_0,$
	\item for $i\in [\ell],$ $(G(N[R_i]), \partial(N[R_i]), \rho\restr{\partial(N[R_i])})$  is a $\beta$-protrusion of ${\bf G},$ and
	\item  for $i\in [\ell],$ $N(R_{i})\subseteq R_0.$
	      %SED: actually, this last condition is not necessary! But makes things  more visualizable!
\end{itemize}
As in the unboundaried case, we call the sets $N[R_i]$ $i\in [\ell],$ the  {\em protrusions} of ${\cal P}$ and the set $R_0$ the {\em core} of ${\cal P}.$

The following theorem is a reformulation using our notation of one of the main results of Kim et al.~\cite{KimLPRRSS16line}, which is stronger than what we need, in the sense that also applies to {\sl topological}-minor-graphs. It is worth mentioning that, for $H$-minor-free-graphs, an appropriate protrusion decomposition can also be found using the results in {\cite[Lemma 3.10]{F.V.Fomin:2010oq}}.

\begin{theorem}\label{label_misconceptions}
	Let $c,\beta, t$ be positive integers, let $H$ be a $q$-vertex graph, and let $G$ be an $n$-vertex $H$-topological-minor-free graph. If we are given a set $M \subseteq V(G)$ with $|M| \leq c \cdot t$ such that $\tw(G-M) \leq \beta,$ then we can compute in time $\Ocal(n)$ an $((\alpha_{H} \cdot \beta \cdot c)\cdot t, 2\beta + q)$-protrusion decomposition ${\cal P}$ of $G$ with $M$ contained in the core of ${\cal P},$ where $\alpha_{H}$ is a constant depending only on $H$ such that $\alpha_H \leq 40 q^2 2 ^{5 q \log q}.$
\end{theorem}

Having stated the above definitions, the following lemma is an easy consequence of \autoref{label_gleichsinnigen} and \autoref{label_misconceptions}.

%\ig{again, I think that in \autoref{label_unintentional} below we can just set $q = h$}

\begin{lemma}\label{label_unintentional}
	There exists a function $\newfun{label_athanasopoulos}: \Bbb{N}^2\to\Bbb{N}$ such that
	if $t,q,h\in \Bbb{N}$  and ${\bf G}=(G,B,\rho)$ is
	a $K_{q}$-minor-free boundaried graph  in ${\cal R}_{h}^{(t)},$ then
	${\bf G}$ admits a $(\funref{label_athanasopoulos}(q,h)\cdot t,\funref{label_athanasopoulos}(q,h))$-protrusion decomposition.
	Moreover, it holds that
	$\funref{label_athanasopoulos}(q,h)=  \funref{label_eavesdropper}(q,h) \cdot \funref{label_confrontation}(q,h) \cdot 2^{\Ocal(q \log q)}.$

	%\red{$\funref{label_athanasopoulos}(q,h)=(\funref{label_inevitability}(h))^{{2^{2^{\Ocal((q+h)\cdot \log(q+h))}}}}.$}
\end{lemma}
%\sed{Please reduce protrusion bound! \ig{what do you mean?}}
\begin{proof}
	By  \autoref{label_gleichsinnigen}, $G$ contains an $\funref{label_compartmentalization}(q,h)$-treewidth modulator $M$ that contains $B$ and has  at most $ 2t$ vertices. We can now apply \autoref{label_misconceptions} to $G$ and $M$ with $H=K_q,$ $c=\funref{label_confrontation}(q,h),$ and $\beta=\funref{label_eavesdropper}(q,h),$ obtaining a $(\funref{label_athanasopoulos}(q,h)\cdot t,\funref{label_athanasopoulos}(q,h))$-protrusion decomposition ${\cal P}$ of $G$ with $M$ contained in the core of ${\cal P}$ and  $\funref{label_athanasopoulos}(q,h):=  \funref{label_eavesdropper}(q,h) \cdot \funref{label_confrontation}(q,h)  \cdot 40 q^2 2 ^{5 q \log q}.$ Since $B \subseteq M$ and $M$ contained in the core of ${\cal P},$ it can be easily checked that ${\cal P}$ is also a $(\funref{label_athanasopoulos}(q,h)\cdot t,\funref{label_athanasopoulos}(q,h))$-protrusion decomposition of ${\bf G}.$
\end{proof}

%\subsection{Reducing protrusions}
%\label{fissures}

Once we have the protrusion decomposition given by \autoref{label_unintentional}, all that remains is to replace the protrusions by equivalent ones of size depending only on the collection $\Fcal.$ The protrusion replacement technique, which is nowadays part of the basic toolbox of parameterized complexity, originated in the meta-theorem of Bodlaender et al.~\cite{BodlaenderFLPST16}, whose objective was to produce linear kernels for a wide family of problems on graphs of bounded genus. This technique was later extended to graphs excluding a fixed minor by Fomin et al.~\cite{F.V.Fomin:2010oq} and then to graphs excluding a fixed topological minor by Kim et al.~\cite{KimLPRRSS16line}.  We could directly apply the results of Fomin et al.~\cite{F.V.Fomin:2010oq} to the protrusion decomposition of a representative given by  \autoref{label_unintentional}, hence reducing each protrusion to an equivalent one of size $\Ocal_{\Fcal}(1),$ yielding an equivalent representative of size $\Ocal_{\Fcal}(t).$ However, the drawback of the results in~\cite{F.V.Fomin:2010oq} (and also in~\cite{BodlaenderFLPST16,KimLPRRSS16line}) is that they do not provide {\sl explicit} bounds on the hidden constants. In order to be able to do so (cf. \autoref{label_reconstituirse}), we apply the protrusion replacement used by Baste et al.~\cite{BasteST20-monster1}, which is suited for the \textsc{$\Fcal$-M-Deletion} problem. This yields explicit constants because it uses ideas similar to the ones presented by Garnero et al.~\cite{GarneroPST15} (later generalized in~\cite{GarneroPST19}) for obtaining kernels with explicit constants.\medskip

Given a function $\xi:\Bbb{N}^2\to\Bbb{N}$ and a $t$-boundaried graph ${\bf G},$ we say that
${\bf G}$ is {\em $\xi$-protrusion-bounded} if, for every $t'\in\Bbb{N},$
all $\beta$-protrusions of ${\bf G}$ have at most $\xi(\beta)$ vertices. The following lemma is again a reformulation using our notation of one of the results of  Baste et al.~\cite{BasteST20-monster1}. Namely, it is a consequence of the proof\footnote{In the statement of~\cite[Lemma 7.2]{BasteST20-monster1} it is required that the family $\Fcal$ contains a planar graph, an assumption that is not true anymore in our case. However, in the proof this fact is only used to guarantee that the considered protrusion has treewidth bounded by a function depending only on $\Fcal.$ Thanks to \autoref{label_unintentional}, we can assume that this also holds in our setting.
	% \ig{Dimitrios, do you think that this comment about \autoref{label_intentionnel} is too much hand-waving?}
} of~\cite[Lemma 7.2]{BasteST20-monster1}.

%\ig{In \autoref{label_intentionnel} below we may also set $q = h,$ right? Otherwise, we have to say that there is a function for every $q,$ as I wrote below}

\begin{lemma}%[{\sf JOHNSON}]
	\label{label_intentionnel}
	There exists a function $\newfun{label_chronological}: \Bbb{N}^2\to\Bbb{N}$ such that
	if $t,q,h\in \Bbb{N}$  and ${\bf G}=(G,B,\rho)$ is
	a $K_{q}$-minor-free boundaried graph  in ${\cal R}_{h}^{(t)},$ then
	${\bf G}$ is $\funref{label_chronological}(q,h)$-protrusion-bounded. Moreover, $\funref{label_chronological}(q,h) = 2^{2^{2^{\Ocal(\funref{label_athanasopoulos}(q,h)\cdot \log \funref{label_athanasopoulos}(q,h))}}}.$% \ig{put explicit bound!!}.
\end{lemma}

Using  \autoref{label_unintentional}
and \autoref{label_intentionnel},
we can easily prove \autoref{label_afincamiento}, that is the main result on which the algorithm of \autoref{label_instintivamente} is based (cf. \autoref{label_voixoderncrow}). In particular, it implies \autoref{label_privateering}.

%Some consequences of \autoref{label_soggiugnendo} are the following.

%\ig{is it the right place to state this theorem?}

\begin{theorem}%[{\sf REPRESENTATIVES}]
	\label{label_afincamiento}
	There exists a function $\newfun{label_annihilation}:\Bbb{N}^2\to\Bbb{N}$ such that, for every $t \in \Bbb{N}$ and $q,h\in \Bbb{N}_{\geq 1},$ if ${\bf G}=(G,B,\rho)$ is
	a $K_{q}$-minor-free boundaried graph  in ${\cal R}_{h}^{(t)},$  then  $|V(G)|\leq \funref{label_annihilation}(q,h)\cdot t.$ Moreover, it holds that $\funref{label_annihilation}(q,h) \leq \funref{label_athanasopoulos}(q,h)\cdot (\funref{label_chronological}(q,h)+1).$
\end{theorem}
\begin{proof}
	By  \autoref{label_unintentional}, ${\bf G}$ admits an $(\funref{label_athanasopoulos}(q,h)\cdot t,\funref{label_athanasopoulos}(q,h))$-protrusion decomposition $\Pcal.$ By \autoref{label_intentionnel}, each of the protrusions of $\Pcal$ has at most $\funref{label_chronological}(q,h)$ vertices. Therefore,
	$$
		|V(G)| \leq \funref{label_athanasopoulos}(q,h)\cdot t + \funref{label_athanasopoulos}(q,h) \cdot \funref{label_chronological}(q,h) \cdot t ,
	$$
	and the theorem follows with $\funref{label_annihilation}(q,h) := \funref{label_athanasopoulos}(q,h)\cdot (\funref{label_chronological}(q,h)+1).$
\end{proof}

Let $h:= \max_{F \in {\mathcal F}} \{ \max_{H \in {\sf ext}(F)} {\sf detail}(H)  \}.$
The following corollary is an immediate consequence of \autoref{label_afincamiento}, by using the fact that all $t$-representatives in ${\cal R}_{h}^{(t)},$ except one, are $K_h$-minor-free, hence they have $\Ocal(\funref{label_annihilation}(h,h)\cdot h\sqrt{\log h})\cdot t$  edges; see for instance~\cite{MyersT05}. Note that are at most ${n^2\choose m}=2^{\Ocal(n\log m)}$ different graphs on $n$ vertices and $m$ edges and that, if $(G,B,
	\rho)\in {\cal R}_{h}^{(t)},$ then  \autoref{label_afincamiento} implies that $|V(G)|\leq \funref{label_annihilation}(h,h)\cdot t.$ Note also that there are ${|V(G)| \choose t}=2^{\Ocal(t\log |V(G)|)}$ choices for $B,$ and $t!=2^{\Ocal(t\log t)}$ choices for
$\rho.$ Therefore,  $|{\cal R}_{h}^{(t)}|=2^{{\cal O}\big(\funref{label_annihilation}(h,h)\cdot t\cdot \log (\funref{label_annihilation}(h,h)\cdot h\sqrt{\log h}\cdot t)+\funref{label_annihilation}(h,h)\cdot t\log (\funref{label_annihilation}(h,h)\cdot t)+t\log t\big)}$ and we can conclude the following.

\begin{corollary}
	\label{label_encyclopedia}
	There exists a function $\newfun{label_liberalization}:\Bbb{N}\to\Bbb{N}$ such that for every $t\in\Bbb{N}_{\geq 1},$  $|{\cal R}_{h}^{(t)}| \leq 2^{\funref{label_liberalization}(h) \cdot t \cdot \log t}.$ In particular, the relation $\equiv_{h}$ partitions  ${\cal B}^{(t)}$
	into at most  $2^{\funref{label_liberalization}(h) \cdot t \cdot \log t}$ equivalence classes.  Moreover, it holds that $\funref{label_liberalization}(h) =\Ocal(\funref{label_annihilation}(h,h) \cdot \log(\funref{label_annihilation}(h,h) \cdot h \sqrt{\log h})).$
	%=\Ocal( \funref{label_chronological}(h,h)).$
\end{corollary}

\paragraph{The dynamic programming algorithm.} Having proved \autoref{label_encyclopedia}, we can apply~\cite[Theorem 8.1]{BasteST20-monster1}
to compute the parameter ${\bf m}_{\cal F}(G)$ within the claimed running time.

For the sake of completeness, let us comment some details of this dynamic programming algorithm, whose details can be found in~\cite[Section 8]{BasteST20-monster1}. First of all, to run the algorithm we need to have the set ${\cal R}_{h}^{(t)}$ of representatives at hand. This can be done easily relying on \autoref{label_afincamiento}, by generating all $t$-boundaried graphs on at most $\funref{label_annihilation}(h,h)\cdot t$ vertices and $\Ocal(\funref{label_annihilation}(h,h)\cdot h\sqrt{\log h})\cdot t$  edges, partitioning them into equivalence classes according to $\equiv_{h},$ and picking an element of minimum size in each of them; see~\cite[proof of Lemma 7.1]{BasteST20-monster1} for more details. To simplify the description of the dynamic programming update operations, the main algorithm in~\cite{BasteST20-monster1} is written in terms of \emph{branchwidth} instead of treewidth. Without defining branchwidth here, it is enough to say that it is linearly equivalent to treewidth, in the sense that both parameters differ by a constant  factor and whose corresponding decompositions can be easily transformed from one to the other~\cite{Rob91}. Also, the main algorithm in~\cite{BasteST20-monster1} is written in terms of {\sl topological} minors, that is, given a finite graph class ${\cal F}'$ and a graph $G,$
it computes ${\bf tm}_{\cal F'}(G),$ that is the  minimum-size set of vertices $S \subseteq V(G)$ whose removal leaves a graph without any of the graphs in a fixed collection $\Fcal'$ as a {\sl topological} minor. This works for our purposes
because of the translation of the question on minors to one on topological minors,
provided by   \autoref{label_daskalojannes}. %,  ${\bf m}_{\cal F}(G) = {\bf tm}_{{\sf ext}(\Fcal)}(G).$
%It is easy to see that computing this parameter suffices for computing ${\bf m}_{\cal F}(G),$ since, as observed in~\red{\cite[Lemma 4]{monster-arXiv}}, for every proper collection $\Fcal$ and every graph $G,$ it holds that ${\bf m}_{\cal F}(G) = {\bf tm}_{\cal F'}(G),$ where $\Fcal' = {\sf ext}(\Fcal)$ (see~\autoref{label_daskalojannes}).
The dynamic algorithm  computes, in a typical bottom-up manner, at every bag separator $B$ of the branch decomposition associated with a $t$-boundaried graph ${\bf G}_B$ and for every representative ${\bf R} \in {\cal R}_{h}^{(t)},$ the minimum size of a set $S \subseteq V({\bf G}_B)$ such that ${\bf G}_B \setminus S \equiv_h {\bf R}.$ These values can be computed in a standard way by combining the values associated with the children of a given node; cf.~\cite[Theorem 8.1]{BasteST20-monster1}.
The overall running time is bounded by $\Ocal(|{\cal R}_{h}^{(t)}|^2 \cdot |E(G)|),$ and taking into account that $|E(G)|  \leq \tw(G) \cdot |V(G)|,$ from \autoref{label_encyclopedia} we obtain the following theorem, which  is a more precise reformulation of \autoref{label_instintivamente}.

\begin{theorem}
	\label{label_desgobernada}
	Let $t,h\in \Bbb{N},$ ${\cal F}$ be a proper  collection of size at most $h,$ and $G$
	be an $n$-vertex graph of treewidth at most $t.$ Then ${\bf m}_{\cal F}(G)$
	can be computed  by an algorithm that runs in $2^{\Ocal( \funref{label_liberalization}(h)  \cdot  t \log t)}  \cdot n$ steps.
	%$2^{\Ocal(t)}\cdot |{\cal R}_{h}^{(t)}|\cdot n$ steps.
\end{theorem}

In \autoref{label_reproductions} we give upper bounds on the constants depending on the collection $\Fcal$ involved in our algorithm. These upper bounds depend explicitly on the
parametric dependencies of the Unique Linkage Theorem~\cite{KawarabayashiW2010asho,RobertsonSGM22}.

\section{Further research}
\label{label_pelopponesian}

We presented an algorithm for solving the  $\mathcal{F}$-\textsc{M-Deletion} problem in time $\Ostar(2^{\Ocal(\tw  \cdot \log \tw)})$ for every  collection $\Fcal.$
This algorithm together with the single-exponential algorithms and lower bounds presented in previous papers of this series~\cite{BasteST20-monster2,BasteST20-monster3} yield a complete classification of the asymptotic complexity of $\mathcal{F}$-\textsc{M-Deletion} parameterized by treewidth assuming the \ETH, when $\Fcal = \{H\}$ and $H$ is connected (\autoref{label_purposelessness}).
However, we do not have a complete classification when $|\Fcal| \geq 2,$ even for connected $\Fcal.$ To ease the presentation, let us call a connected graph $H$ \emph{easy} (resp. \emph{hard}) if  $\{H\}$-M-\textsc{Deletion} is solvable in time $\Ostar(2^{\Theta(\tw)})$ (resp. $\Ostar(2^{\Theta(\tw  \cdot \log \tw)})$). Suppose that  $\Fcal= \{H_1,H_2\}$ with both $H_1$ and $H_2$ being connected. Using the recent results of Baste~\cite{Baste19}, it is possible to prove that if both $H_1$ and $H_2$ are easy, then $\Fcal$ is easy as well (easiness of graph collections is defined in the obvious way). However, if both $H_1$ and $H_2$ are hard, then strange things may  happen. For instance, Bodlaender et al.~\cite{BodlaenderOO18} presented an algorithm running in time $\Ostar(2^{\Ocal(\tw)})$ for \textsc{Pseudoforest Deletion}, which consists in,  given a graph $G$ and an integer $k,$ deciding whether one can delete at most $k$ vertices from $G$ to obtain a
\emph{pseudoforest}, i.e.,  a graph
where each connected component contains at most one cycle. Note that \textsc{Pseudoforest Deletion} is equivalent to \textsc{$\{\ourdiamond, \butterfly\}$-M-Deletion}. While both the \ourdiamond and the \butterfly are hard graphs (cf. \autoref{label_thoughtlessness}), $\{\ourdiamond, \butterfly\}$ is an easy collection. The cases where $H_1$ is easy and $H_2$ is hard seem even trickier.
Obtaining (tight) lower bounds when $\Fcal$ may contain disconnected graphs is
another challenging avenue for further research.

It is also interesting to consider the version of the problem where the graphs in $\Fcal$ are forbidden as {\sl topological} minors; we call this problem $\mathcal{F}$-\textsc{TM-Deletion}. While most of the lower bounds that we presented in~\cite{BasteST20-monster3} also hold  for $\mathcal{F}$-\textsc{TM-Deletion}, the algorithm in time $\Ostar(2^{\Ocal(\tw  \cdot \log \tw)})$ of this paper does {\sl not} work for topological minors. In this direction, the algorithm in time $\Ostar(2^{\Ocal(\tw  \cdot \log \tw)})$ for $\mathcal{F}$-\textsc{M-Deletion}  when $\Fcal$ contains a {\sl planar} graph) given in~\cite{BasteST20-monster1} also works for $\mathcal{F}$-TM-\textsc{Deletion}, if we additionally require $\Fcal$ to contain a {\sl subcubic} planar graph (in order to bound the treewidth of the representatives).
The main obstacle for applying  our approach in order  to achieve
a time $\Ostar(2^{\Ocal(\tw  \cdot \log \tw)})$ for {\sl every} collection $\Fcal,$  is that topological-minor-free graphs do not enjoy the
flat wall structure that is omnipresent  in our proofs.
Another reason is that in our rerouting procedure, in order to find an irrelevant vertex (\autoref{label_constitutivos}),  we may find a {\sl different} topological minor model that corresponds to the same minor. Nevertheless, we think that this latter difficulty can be overcome for planar graphs --or even minor-free graphs-- by making use of the rerouting potential of \autoref{label_interessiert},  as this is done in \cite{GolovachST20-SODA} for planar graphs.

Finally, it is worth mentioning that the algorithm presented in this paper, as well as the main combinatorial result (\autoref{label_pretendientes}), have been used in~\cite{ICALP-versions} (see~\cite{SauST21kapiII} for the full version) to obtain a fixed-parameter algorithm for the $\mathcal{F}$-\textsc{M-Deletion} problem parameterized by $k.$ \autoref{label_pretendientes} has also been used in~\cite{SauST21kapiI} in order to provide explicit upper bounds on the size of the minor-obstructions of the set of {\sf yes}-instances of the $\mathcal{F}$-\textsc{M-Deletion} problem, as a function of $\mathcal{F}$ and $k.$

%, and that for every
%(connected) collection $\mathcal{F},$ the $\mathcal{F}$-TM-\textsc{Deletion} problem is also solvable in time $\Ostar(2^{\Ocal(\tw \cdot \log \tw)})$ \ig{do you agree with this statement, or is it too risky?}.% under the {\sf ETH}.

%\newpage
%{\small
%\bibliographystyle{abbrv}
\bibliography{Biblio-Fdeletion}
%}

\newpage
\begin{appendix}

	\section{Illustration of the complexity dichotomy }
	\label{label_apreciadores}

	\begin{figure}[h!]
		\begin{center}%\vspace{-.1cm}
			\includegraphics[width=.76\textwidth]{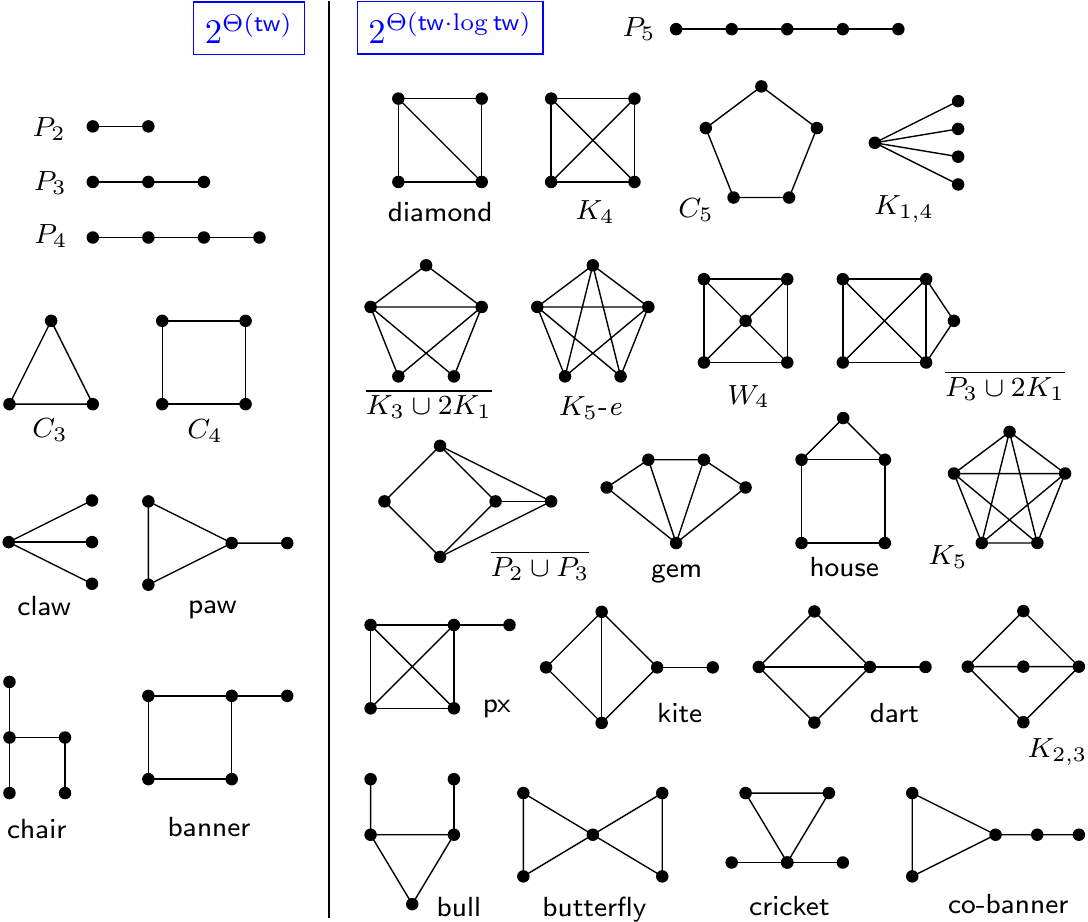}
		\end{center}%\vspace{-.25cm}
		\caption{Classification of the complexity of \textsc{$\{H\}$-M-Deletion} for all connected simple graphs $H$ with $2 \leq |V(H)|\leq 5,$ according to our results: for the nine graphs on the left, the problem is solvable in time $2^{\Theta(\tw)} \cdot n^{\Ocal(1)}$ under the \ETH. For the 21 graphs on the right and for all the connected graphs on at least six vertices, the problem is solvable in time $2^{\Theta(\tw \cdot \log \tw)} \cdot n^{\Ocal(1)}$ under the \ETH.}
		%For \textsc{$\{H\}$-TM-Deletion}, $K_{1,4}$ should be on the left.}
		% \ig{Say that $K_5$ is the only nonplanar one} \ig{Say also that the picture changes for the topological minor version (stars)} \ig{Put in conclusions?}.
		\label{label_thoughtlessness}
		% \vspace{.35cm}
	\end{figure}

	\section{An estimation of the constants depending on $\Fcal$ in our algorithm}
	\label{label_reproductions}

	The main result of this paper is that the \textsc{${\cal F}$-M-Deletion}  problem
	can be solved in time $2^{{\cal O}(f(h)\cdot \tw\cdot \log\tw)} \cdot n$ on $n$-vertex graphs of treewidth at most $\tw,$  for some computable
	function $f:\Bbb{N}\to\Bbb{N},$ where here $h$ is an upper bound on the size of the
	graphs in ${\cal F}.$ This appendix is dedicated to an estimation of an upper bound on the function $f.$
	Notice that almost all the statements of the results in this paper are accompanied with specific bounds
	on the involved functions, usually in terms of functions defined in previous statements. However, there are two exceptions. The first one is the function $\funref{label_inevitability}$ of \autoref{label_wiederzugeben}, which we discuss in \autoref{label_reconstituirse}. The second one is \autoref{label_murmuradores}, where no explicit
	bound for $\funref{label_confrontation}(q,h)$
	is given. This is because the existence of  $\funref{label_confrontation}(q,h)$  follows by
	applying  Akra-Bazzi Theorem~\cite{AkraB98}  as a black box and this does not provide any estimation of $\funref{label_confrontation}.$ To circumvent this issue, in \autoref{label_confederarlos} we provide an improved version of \autoref{label_murmuradores}, namely  \autoref{label_gleichsinnigen}, whose proof uses a direct induction, without invoking the {Akra-Bazzi Theorem \cite{AkraB98,LeightonAB}}. This alternative proof is strongly based on the proof of  \cite[Lemma~3.6]{F.V.Fomin:2010oq}.  Finally, in \autoref{label_reconstituirse} we provide an upper bound on the constants involving $\Fcal$ in our algorithm. For this,  we will use the stronger version of  \autoref{label_murmuradores} given in \autoref{label_confederarlos}.

	\subsection{An improved version of \autoref{label_murmuradores}}
	\label{label_confederarlos}

	In this section we provide an improved version of \autoref{label_murmuradores}, whose proof is  an adaptation of the proof of  \cite[Lemma~3.6]{F.V.Fomin:2010oq}.

	\begin{lemma}%[{\sf MODULATOR$\star$}]
		\label{label_gleichsinnigen}
		There exists a function $\newfun{label_compartmentalization}: \Bbb{N}^2\to\Bbb{N}$ such that
		if $t,q,h\in \Bbb{N}$  and ${\bf G}=(G,B,\rho)$ is
		a $K_{q}$-minor-free boundaried graph  in ${\cal R}_{h}^{(t)},$ then
		$G$ contains an $\funref{label_compartmentalization}(q,h)$-treewidth modulator that contains $B$
		and has  at most $ 2 t$ vertices. Moreover, it holds that $\funref{label_compartmentalization}(q,h)=\Ocal((\funref{label_superintendent}(q,\funref{label_encompassing}(\funref{label_scheinbeziehungen}(q),3,3h),\funref{label_intervention}(\funref{label_scheinbeziehungen}(q),3h)))^2).$
		%(\funref{label_inevitability}(h))^{{2^{2^{\Ocal((q+h)\cdot \log(q+\red{h}))}}}}.$
	\end{lemma}

	%\removed{
	\begin{proof}
		For simplicity, we use ${\sf q}$ as a shortcut for  the triple $(\funref{label_scheinbeziehungen}(q),\funref{label_encompassing}(\funref{label_scheinbeziehungen}(q),3,3h),\funref{label_intervention}(\funref{label_scheinbeziehungen}(q),3h)).$
		%${\bf p}$ as a shortcut of ${\bf p}_{\funref{label_scheinbeziehungen}(q),\funref{label_encompassing}(\funref{label_scheinbeziehungen}(q),3,3h),\funref{label_intervention}(\funref{label_scheinbeziehungen}(q),3h)}.$
		We   define the constants $s=\funref{label_superintendent}(q,\funref{label_encompassing}(\funref{label_scheinbeziehungen}(q),3,3h),\funref{label_intervention}(\funref{label_scheinbeziehungen}(q),3h)),$ $t_0 = 256 s^{2},$ and  $c=s\cdot  \sqrt{t_0}.$
		We define the relation $\leq_{\circ}$  so that $a\leq_{\circ} b$ means that $a\leq \max\{0,b\}.$
		We first prove, by induction on $t,$ the following  statement.\medskip

		\noindent{\em Claim:} For every  non-negative integer $t$
		% \geq \frac{1}{3}t_0$
		and  every $K_{q}$-minor-free  graph $G,$
		if ${\bf p}_{\sf q}(G)\leq t$ then  $G$ has a $ c$-treewidth modulator $Z$ with $|Z|\leq_{\circ}  {t} - 16 s\cdot \sqrt{t}.$ \medskip

		\noindent{\em Proof of claim:}
		In the base case we consider any $t$ with  $0\leq  t \leq t_0.$
		As ${\bf p}_{\sf q}(G)\leq t,$ \autoref{label_substantiality} implies that $\tw(G)\leq s\cdot \max\{\sqrt{t},1\}\leq s\cdot\sqrt{t_0}.$ % (recall that $k\geq t_0\geq 1$).
		Thus $G$ has a $ c$-treewidth modulator of size $0\leq_{\circ}  {t} - 16 s\cdot \sqrt{t},$ and the claim follows.

		\smallskip

		For the inductive step, let $t > t_0$ and suppose that the claim is true for  every $t'$ with $0\leq t'\leq t-1.$ We prove that the claim holds also for $t.$ Consider a
		graph $G$ with ${\bf p}_{\sf q}(G)\leq t$ and let $S$ be a set  of at most $t$ vertices affecting  every ${\sf q}$-apex-wall triple of  $G.$
		Because of  \autoref{label_substantiality}, ${\bf p}_{\sf q}(G)\leq t$ implies that   $\tw(G) \leq  s\cdot \max\{\sqrt{t},1\}=s\cdot \sqrt{t}.$

		By applying \autoref{label_instituciones} to $G$ and $S,$ there is a 2/3-balanced separation $(L,R)$ of $S$ in $G$ such that $|L\cap R|\leq \tw(G) + 1 \leq s \cdot \sqrt{t} + 1$ and there exists some $\alpha\in  [\frac{1}{3},\frac{2}{3}]$ such that $|(L\setminus R) \cap S| \leq \alpha\cdot |S|\leq \alpha\cdot t$ and  $|(R\setminus L)\cap S| \leq (1-\alpha)\cdot |S|\leq (1-\alpha)\cdot t.$

		Since $S$ affects  every ${\sf q}$-apex-wall triple of $G,$ \autoref{label_concerniente} gives that
		the set
		$L\cap (R \cup S)$ affects  every  ${\sf q}$-apex-wall triple of $G[L].$ This implies
		\begin{align*}
			{\bf p}_{\sf q}(G[L]) & \leq |L\cap (R \cup S)| = |(L\setminus R) \cap S| + |L\cap R| \leq \alpha\cdot t + (s\cdot  \sqrt{t} + 1) \leq \alpha\cdot t + 2s\cdot  \sqrt{t}.
		\end{align*}
		Here the last inequality follows from the assumption that $t \geq t_0 \geq 1.$% and the choice of $2s.$
		\medskip

		In order to apply the inductive hypothesis, note that   $t':=\alpha\cdot t + 2s\cdot  \sqrt{t} \leq t - 1$ for $t \geq t_0.$ This can be verified by using the fact that $s\geq 1,$ $
			\alpha\leq	 \frac{2}{3},$ and checking that $\frac{2}{3} t + 2s\cdot  \sqrt{t} \leq t - 1$ for $t \geq t_0.$
		Indeed, the inequality holds whenever $\sqrt{t}\geq 3s+\frac{3}{2}\sqrt{(2s)^2+\frac{4}{3}}$ and this is the case
		as $\sqrt{t}\geq \sqrt{t_{0}}=16s\geq 3\cdot s+\frac{3}{2} \sqrt{(2s)^2+\frac{4}{3}}.$

		Therefore we can  apply the induction hypothesis to $G[L]$ and $t'$ and obtain a $c$-treewidth modulator $Z_L$ of $G[L],$ such that
		\begin{align*}
			|Z_L| & \leq_{\circ}  t' - 16s\cdot  \sqrt{t'}  \leq_{\circ}  (\alpha\cdot t + 2s\cdot  \sqrt{t}) - 16s\cdot \sqrt{\alpha\cdot t + 2s\cdot  \sqrt{t}} \leq_{\circ}  (\alpha\cdot t + 2s\cdot  \sqrt{t}) - 16s\cdot  \sqrt{\alpha\cdot t}.
		\end{align*}
		A symmetric argument applied to $G[R]$ yields a treewidth modulator $Z_R$ of $G[R],$ such that
		\begin{align*}
			|Z_R| & \leq_{\circ}  \left((1-\alpha)\cdot t + 2s\cdot  \sqrt{t} \right) - 16s\cdot  \sqrt{  (1-\alpha)\cdot t}.
		\end{align*}

		We now construct a $ c$-treewidth modulator $Z$ of $G$ as follows by setting $Z:= Z_L \cup (L\cap R) \cup Z_R.$ The set $Z$ is a $ c$-treewidth-modulator of $G$ because every connected component of $G- Z$ is a subset of either $(L\setminus (L\cap R))\setminus Z_{L}$ or $(R\setminus (L\cap R))\setminus Z_{R},$ and $Z_L$ and $Z_R$ are $ c$-treewidth modulators for $G[L]$ and $G[R]$ respectively. Finally we bound the size of $Z.$
		\begin{align*}
			|Z| & \leq |Z_L| + |Z_R| + |S|                                                                                                                                                                                \\
			    & \leq_{\circ}   (\alpha\cdot t + 2s\cdot  \sqrt{t}) - 16s\cdot  \sqrt{t\cdot  \alpha} +  \left((1-\alpha)\cdot t + 2s\cdot  \sqrt{t} \right) - 16s\cdot  \sqrt{t\cdot (1-\alpha)} + s \cdot \sqrt{t} + 1 \\
			    & \leq_{\circ}   (\alpha\cdot t + 2s\cdot  \sqrt{t}) - 16s\cdot  \sqrt{t\cdot  \alpha} +  \left((1-\alpha)\cdot t + 2s\cdot  \sqrt{t} \right) - 16s\cdot  \sqrt{t\cdot (1-\alpha)} + 2s\cdot  \sqrt{t}    \\
			%& =   t - 16\left(\sqrt{\alpha} +  \sqrt{1-\alpha}\right)\cdot s\cdot  \sqrt{t}  +  6 s\cdot \sqrt{t} \\
			    & \leq_{\circ}   t -   \big(16 \big(\sqrt{\alpha} +  \sqrt{1-\alpha}\big) -  6\big) \cdot s\cdot  \sqrt{t}                                                                                                \\
			    & \leq_{\circ}   t - 16s\cdot  \sqrt{t}.
		\end{align*}
		The last inequality  uses  the fact that $16 \left(\sqrt{\alpha} +  \sqrt{1-\alpha}\right) -  6\geq 16,$ for every $\alpha\in[\frac{1}{3},\frac{2}{3}].$
		The claim follows.
		\medskip

		Suppose now that ${\bf G}=(G,B,\rho)$ is
		a $K_{q}$-minor-free boundaried graph  in ${\cal R}_{h}^{(t)}.$
		From
		\autoref{label_dispensaries}, ${\bf p}_{\sf q}(G)\leq t$ and, because of the above claim, $G$ contains a $c$-treewidth modulator $Z$ where $|Z|\leq_{\circ} {t} - 16s\cdot  \sqrt{t}\leq t.$ This, in turn, implies that $B\cup Z$ is a
		$c$-treewidth modulator of $G$ that contains $B$ and has  size $2 t.$
		Also observe that $c= 16s^2.$ Therefore, the lemma holds for $\funref{label_compartmentalization}(q,h)=\Ocal(s^2).$
		%16(\funref{label_superintendent}(\funref{label_scheinbeziehungen}(q),\funref{sdfhgdfgdfgdsfgdf}(q-5,2h))+1)^2=(\funref{label_inevitability}(h^2))^{{2^{2^{\Ocal((q+h^2)\cdot \log(q+h^2))}}}}.$}
	\end{proof}
	%}

	We  stress that, as it is done in the proof of {\cite[Lemma 3.6]{F.V.Fomin:2010oq}}, it is possible to find a modulator
	of size at most $(1+\varepsilon) \cdot t,$ for every positive real $\varepsilon>0.$ Nevertheless, we have provided the proof of \autoref{label_gleichsinnigen} for the particular case $\varepsilon =2,$ which is enough for our purposes.

	\subsection{Upper bounds on the constants depending on the excluded minors}
	\label{label_reconstituirse}

	In this section we provide an estimation on the function $\funref{label_liberalization}$ in \autoref{label_desgobernada}, or equivalently on the constant $c_{\Fcal}$ in \autoref{label_instintivamente}.
	We first provide some definitions in order to introduce the Unique Linkage Theorem~\cite{KawarabayashiW2010asho,RobertsonSGM22}.

	A \emph{linkage} in a graph $G$ is a subgraph $L$ of $G$  whose connected components are all non-trivial paths. The {\em paths} of a linkage $L$ are its connected components and we denote them by ${\cal P}(L).$
	The {\em size} of $L$ is the number of its paths and is denoted by $|L|.$
	The \emph{terminals} of a linkage $L,$ denoted by $T(L),$ are the endpoints of the paths in ${\cal P}(L),$ and
	%\marg{$|L|,$ $T(L),$  ${\cal P}(L)$}
	the \emph{pattern} of $L$ is the set $$\big\{\{s,t\} \mid {\cal P}(L)\mbox{ contains some $(s,t)$-path}\big\}.$$ Two linkages $L_{1},L_{2}$ of $G$  are {\em equivalent} if they have the same pattern and we denote this fact by $L_{1}\equiv L_{2}.$
	We say that a linkage $L$ in a graph $G$ is {\em unique}
	if for every linkage $L’$ that is equivalent to $L$
	it holds that $V(L’) =V(L).$

	According to the proof of~\autoref{label_interessiert} in~\cite{GolovachST20-SODA}, the function
	$\funref{label_inevitability}$ emerges from the  following result, known as the Unique Linkage Theorem.

	%We say that a linkage $L$ of a graph $G$ is {\em vital} if $V(L)=V(G)$ and there is no other linkage of $G$ that is equivalent to $L.$

	\begin{proposition}[\!\!\cite{KawarabayashiW2010asho,RobertsonSGM22}]
		\label{label_wiederzugeben}
		There exists a  function  $\funref{label_inevitability}:\Bbb{N}_{\geq  0}\to\Bbb{N}_{\geq  0}$ such that
		if $G$ is a graph and $L$ is a unique linkage of $G,$ then $\tw(G)\leq \funref{label_inevitability}(|L|).$
	\end{proposition}

	It is worth mentioning that~\cite{KawarabayashiW2010asho,RobertsonSGM22}
	do not provide the precise number of exponentiations
	involved in the function $\funref{label_inevitability},$ and therefore we will express our upper bounds in terms of this function.
	%For two positive integers $h,x,$ we recursively
	%define the quantity $h\uparrow\uparrow x$ such that $h\uparrow \uparrow0=h$
	%and, for $x\geq 1,$ $h\uparrow\uparrow x=2^{h~\uparrow\uparrow~(x-1)}.$
	%As~\cite{KawarabayashiW2010asho,RobertsonS22}
	%do not provide the precise number of exponentiations
	%involved in the function $\funref{label_inevitability},$ we define
	%$x_{\sf ul}$ as the minimum integer such that $\funref{label_inevitability}(h)=h^{\Ocal(1)} \uparrow\uparrow x_{\sf ul} .$
	Namely, in order to provide an upper bound on $\funref{label_liberalization}(h),$ we backtrack the functions involved in the intermediate results of this paper as follows:
	\begin{itemize}
		\item  According to \autoref{label_encyclopedia}, $\funref{label_liberalization}(h)=\Ocal(\funref{label_annihilation}(h,h) \cdot \log(\funref{label_annihilation}(h,h) \cdot h \sqrt{\log h})).$
		\item By
		      \autoref{label_afincamiento},  $\funref{label_annihilation}(h,h) =
			      \Ocal(\funref{label_athanasopoulos}(h,h)\cdot \funref{label_chronological}(h,h)).$
		\item By \autoref{label_intentionnel}, $\funref{label_chronological}(h,h) = 2^{2^{2^{\Ocal(\funref{label_athanasopoulos}(h,h)\cdot \log \funref{label_athanasopoulos}(h,h))}}}.$
		\item By \autoref{label_unintentional}, $\funref{label_athanasopoulos}(h,h) =\funref{label_eavesdropper}(h,h) \cdot \funref{label_confrontation}(h,h) \cdot 2^{\Ocal(h \log h)}.$
		\item By \autoref{label_gleichsinnigen}, we can take $\funref{label_confrontation}(h,h)=2.$
		\item By \autoref{label_murmuradores} (and \autoref{label_gleichsinnigen} as well), $\funref{label_eavesdropper}(h,h)=
			      \Ocal((\funref{label_superintendent}(h,\funref{label_encompassing}(\funref{label_scheinbeziehungen}(h),3,3h),\funref{label_intervention}(\funref{label_scheinbeziehungen}(h),3h)))^2).$
		\item By \autoref{label_intercanvien}, $\funref{label_scheinbeziehungen}(h)=\Ocal(h^{24}) = h^{\Ocal(1)}.$
		\item By \autoref{label_pretendientes},
		      $\funref{label_intervention}(\funref{label_scheinbeziehungen}(h),3h)=\Ocal(h^{24}) =  h^{\Ocal(1)}.$
		\item By \autoref{label_pretendientes}, $\funref{label_encompassing}(\funref{label_scheinbeziehungen}(h),3,3h)=\Ocal( (\funref{label_inevitability}(h^{\Ocal(1)}))^3).$

		\item By \autoref{label_substantiality}, $\funref{label_superintendent}(h,r,\hat{\ell}) =\funref{label_congiugnersi}(h)\cdot r^{2^{2^{\Ocal((h^{24}+\hat{\ell})\cdot \log (h^{24}+\hat{\ell}))}}},$ which is $r^{2^{2^{(h+\hat{\ell})^{\Ocal(1)}}}}$ by \autoref{label_intercanvien}.
		      If we set $r=\funref{label_encompassing}(\funref{label_scheinbeziehungen}(h),3,3h)=\Ocal( (\funref{label_inevitability}(h^{\Ocal(1)}))^3)$ and $\hat{\ell}=\funref{label_intervention}(\funref{label_scheinbeziehungen}(h),3h)=h^{\Ocal(1)},$ we have that
		      $$
			      \funref{label_eavesdropper}(h,h)=
			      \Ocal((\funref{label_superintendent}(h,r,\hat{\ell}))^2)= \big((\funref{label_inevitability}(h^{\Ocal(1)}))^3\big)^{2^{2^{(h+h^{{\cal O}(1)})^{\Ocal(1)}}}}= \big(\funref{label_inevitability}(h^{{\cal O}(1)})\big)^{2^{2^{h^{{\cal O}(1)}}}}.
		      $$

		\item We set $λ=\funref{label_inevitability}(h^{{\cal O}(1)}).$ Given that $\funref{label_athanasopoulos}(h,h)=  \funref{label_eavesdropper}(h,h) \cdot \funref{label_confrontation}(h,h) \cdot 2^{\Ocal(h \log h)}$ and that $\funref{label_confrontation}(h,h)=2,$ we obtain that   $\funref{label_athanasopoulos}(h,h) =λ^{2^{2^{h^{{\cal O}(1)}}}}.$ We now have that
		      $\funref{label_chronological}(h,h) = 2^{2^{2^{λ^{2^{2^{h^{{\cal O}(1)}}}}}}},$ which implies that $\funref{label_annihilation}(h,h) = 2^{2^{2^{λ^{2^{2^{h^{{\cal O}(1)}}}}}}}$ and thus
		      $\funref{label_liberalization}(h)= 2^{2^{2^{λ^{2^{2^{h^{{\cal O}(1)}}}}}}}.$
		      %In other words, $\funref{label_liberalization}(h)=h^{\Ocal(1)}\uparrow\uparrow (6+\max\{x_{\sf ul}-3,0\}).$
	\end{itemize}
	\medskip

	From \autoref{label_desgobernada} and the above discussion, we conclude  the following corollary, which gives an explicit upper bound  on the contribution  of the maximum size of the graphs in ${\cal F}$ in the complexity of our algorithm, depending on the function $\funref{label_inevitability}$ given by \autoref{label_wiederzugeben}.

	\begin{corollary}
		\label{label_geographical}
		Let ${\cal F}$ be a collection of graphs each of size at most $h,$ and let $G$ be a graph.
		Then  the parameter ${\bf m}_{\cal F}(G)$ can be computed in time

		$$
			2^{\Big(2^{2^{2^{λ^{2^{2^{h^{{\cal O}(1)}}}}}}} \Big)\cdot\tw(G)\cdot \log(\tw(G))}\cdot |V(G)|, \text{ where $λ=\funref{label_inevitability}(h^{{\cal O}(1)}).$}
		$$
	\end{corollary}

\end{appendix}

\end{document}